%% file: main.tex
\documentclass[sigplan,10pt]{acmart}
\acmConference[PL'18]{ACM SIGPLAN Conference on Programming Languages}{January 01--03, 2018}{New York, NY, USA}
\acmYear{2018}
\acmISBN{} 
\acmDOI{} 
\startPage{1}

\setcopyright{none}

\usepackage{booktabs}   
\usepackage{subcaption} 

\input{packages}
\input{macros-lics}
\input{macros-undecidability}
\newtheorem{remark}{Remark}
\newtheorem{notation}{Notation}

\begin{document}

\title{Bouncing threads for infinitary and circular proofs}         
\titlenote{with title note}             


\author{David Baelde}
\affiliation{
  \institution{ENS Paris-Saclay}            
}
\email{david.baelde@lsv.fr}          

\author{Amina Doumane}
\affiliation{
  \institution{CNRS, LIP, ENS Lyon}            
}
\email{amina.doumane@ens-lyon.fr}          

\author{Denis Kuperberg}
\affiliation{
  \institution{CNRS, LIP, ENS Lyon}            
}
\email{denis.kuperberg@ens-lyon.fr}          

\author{Alexis Saurin }
\affiliation{
  \institution{CNRS, IRIF, Université Paris Diderot}            
}
\email{alexis.saurin@irif.fr}          

\begin{abstract}
\input{abstract}
\end{abstract}



\maketitle

\newcommand{\Section}[1]{\vspace{-.4cm}\section{#1}\vspace{-.3cm}}
\renewcommand{\Section}[1]{\section{#1}}
\input{new-intro}

\input{proofsystem}
\input{cut-reductions}
\input{mumll-validity}

\input{bouncing-cut-elim}
\input{decid-body}

\input{conclusion}

\bibliographystyle{splncs04}
\bibliography{biblio}

\appendix

\newpage
\section{Appendices}

\input{appendix-reduction-rules}
\input{appendix-cut-elim}
\input{appendix-additives}

\input{appendix-undecidability-as}

\clearpage
\newpage
\tableofcontents

\end{document}

%% file: packages.tex
\usepackage{ebproofs,color}
\usepackage{url}
\usepackage{pifont}
\usepackage{stmaryrd}
\usepackage{tikz}
\usetikzlibrary{automata}
\usepackage{wrapfig}
\usetikzlibrary{positioning,calc,arrows,shapes}
\usepackage{todonotes}
\usepackage{listings}
\usepackage{lstcoq}

%% file: macros-lics.tex
\usepackage[utf8]{inputenc}
\usepackage[T1]{fontenc}
\usepackage{xspace}
\usepackage{xcolor}
\usepackage{framed} 
\usepackage{wrapfig} 
\usepackage{todonotes}
\presetkeys{todonotes}{fancyline, color=orange}{}
\usepackage{stmaryrd}
\usepackage{amsfonts,amstext,amsmath,amssymb}
\usepackage{thmtools,thm-restate}
\usepackage{cleveref}

\def\savespace{}

\newcommand{\restatablelemma}[2]{
  \expandafter\newcommand\csname statelemma#1\endcsname{#2}
  \begin{lemma}\label{#1}
    \expandafter\csname statelemma#1\endcsname{}
  \end{lemma}
}
\newcommand{\restatelemma}[1]{
  \begin{claim}[\cref{#1}]
    \expandafter\csname statelemma#1\endcsname{}
  \end{claim}
}

\newcommand{\restatableproposition}[2]{
  \expandafter\newcommand\csname stateproposition#1\endcsname{#2}
  \begin{proposition}\label{#1}
    \expandafter\csname stateproposition#1\endcsname{}
  \end{proposition}
}
\newcommand{\restateproposition}[1]{
  \begin{claim}[\cref{#1}]
    \expandafter\csname stateproposition#1\endcsname{}
  \end{claim}
}

\newcommand{\restatabletheorem}[2]{
  \expandafter\newcommand\csname statetheorem#1\endcsname{#2}
  \begin{theorem}\label{#1}
    \expandafter\csname statetheorem#1\endcsname{}
  \end{theorem}
}
\newcommand{\restatetheorem}[1]{
  \begin{claim}[\cref{#1}]
    \expandafter\csname statetheorem#1\endcsname{}
  \end{claim}
}

\usepackage{tikz}
\usetikzlibrary{positioning,calc,arrows,shapes}
\usetikzlibrary{automata}
\usepackage{pifont}
\usetikzlibrary{tikzmark,decorations.pathreplacing}
\usepackage{ebproofs}

\newcommand{\TODO}[1]{\textcolor{red}{[TODO #1]}}
\renewcommand{\TODO}[1]{}

\long\def\hide#1\endhide{}

\newcommand{\psd}[3]{(#1,#2,#3)}
\renewcommand{\C}[1]{\mathcal{#1}}

\newcommand{\TT}[1]{\mathtt{#1}}

\newcommand{\dbot}{\raisebox{0.05em}{\rotatebox[origin=c]{90}{$\models$}}}
\newcommand{\Aax}{\mathsf{A}}
\newcommand{\Ccut}{\mathsf{C}}
\newcommand{\Wwait}{\mathsf{W}}
\def\prem#1{\mathsf{premiss}(#1)}

\makeatletter
\newcommand{\bigpole}{%
  \mathop{\mathpalette\bigp@rp\relax}%
  \displaylimits
}

\newcommand{\bigp@rp}[2]{%
  \vcenter{
    \m@th\hbox{\scalebox{\ifx#1\displaystyle2.1\else1.5\fi}{$#1\independent$}}
  }%
}
\newcommand\independent{\protect\mathpalette{\protect\independenT}{\perp}}
\def\independenT#1#2{\mathrel{\rlap{$#1#2$}\mkern2mu{#1#2}}}
\makeatother

    \newcommand\todoAS[2][]{\todo[color=blue!60,#1]{Alexis: #2}} 
     
    \newcommand\todoDB[2][]{\todo[color=yellow,#1]{David: #2}} 
    
\renewcommand{\todo}[2][]{}

\newcommand{\orig}[2]{\begin{tikzpicture}[remember picture]\node[inner sep=0pt,outer sep=1pt] (#1)  {\ensuremath{#2}};\end{tikzpicture}}

\makeatletter
\tikzset{%
  remember picture with id/.style={%
    remember picture,
    overlay,
    save picture id=#1,
  },
  save picture id/.code={%
    \edef\pgf@temp{#1}%
    \immediate\write\pgfutil@auxout{%
      \noexpand\savepointas{\pgf@temp}{\pgfpictureid}}%
  },
  if picture id/.code args={#1#2#3}{%
    \@ifundefined{save@pt@#1}{%
      \pgfkeysalso{#3}%
    }{
      \pgfkeysalso{#2}%
    }
  }
}

\def\savepointas#1#2{%
  \expandafter\gdef\csname save@pt@#1\endcsname{#2}%
}

\def\tmk@labeldef#1,#2\@nil{%
  \def\tmk@label{#1}%
  \def\tmk@def{#2}%
}

\tikzdeclarecoordinatesystem{pic}{%
  \pgfutil@in@,{#1}%
  \ifpgfutil@in@%
    \tmk@labeldef#1\@nil
  \else
    \tmk@labeldef#1,(0pt,0pt)\@nil
  \fi
  \@ifundefined{save@pt@\tmk@label}{%
    \tikz@scan@one@point\pgfutil@firstofone\tmk@def
  }{%
  \pgfsys@getposition{\csname save@pt@\tmk@label\endcsname}\save@orig@pic%
  \pgfsys@getposition{\pgfpictureid}\save@this@pic%
  \pgf@process{\pgfpointorigin\save@this@pic}%
  \pgf@xa=\pgf@x
  \pgf@ya=\pgf@y
  \pgf@process{\pgfpointorigin\save@orig@pic}%
  \advance\pgf@x by -\pgf@xa
  \advance\pgf@y by -\pgf@ya
  }%
}

\makeatother

\newcommand{\FA}{\mathcal{A}_{\mathit{fresh}}}
\newcommand{\Addr}{\mathit{Addr}}
\newcommand{\interp}[1]{\llbracket #1 \rrbracket}

\newcommand{\weight}[1]{\mathsf{w}(#1)}

\newcommand{\visP}[1]{\mathsf{vp}(#1)}
\newcommand{\hiddenP}[1]{\mathsf{hp}(#1)}

\newcommand{\defname}[1]
{\textbf{\emph{#1}}}
\newcommand{\ie}{i.e.,\xspace}

\newcommand{\muMALL}{\ensuremath{\mu\mathsf{MALL}}\xspace}

\newcommand{\muMALLi}{\ensuremath{\mu\mathsf{MALL}^\infty}\xspace}

\newcommand{\muSMALLi}{\ensuremath{\mu\mathsf{SMALL}^\infty}\xspace}

\newcommand{\mupSMALLi}{\ensuremath{\mu\mathsf{pSMALL}^\infty}\xspace}

\newcommand{\muMLLi}{\ensuremath{\mu\mathsf{MLL}^\infty}\xspace}
\newcommand{\muMLLo}{\ensuremath{\mu\mathsf{MLL}^\omega}\xspace}
\newcommand{\muMLLim}{\ensuremath{\mu\mathsf{MLL}_\mathsf{m}^\infty}\xspace}
\newcommand{\muMALLit}{\ensuremath{\mu\mathsf{MALL^\infty_\tau}}\xspace}

\newcommand{\muMLLit}{\ensuremath{\mu\mathsf{MLL^\infty_\tau}}\xspace}
\newcommand{\muMALLim}{\ensuremath{\mu\mathsf{MALL}_\mathsf{m}^\infty}\xspace}

\newcommand{\with}{\binampersand}
\newcommand{\parr}{\bindnasrepma}

\newcommand{\tensor}{{\otimes}}
\newcommand{\plus}{{\oplus}}

\newcommand{\llzero}{\mathbf{0}}
\newcommand{\llone}{\mathbf{1}}

\newcommand{\dai}{{\footnotesize\maltese}}
\newcommand{\trans}[2]{\xrightarrow[\,#2\,]{\,#1\,}}
\newcommand{\dom}{\TT{Dom}}
\newcommand{\Ai}{\mathtt{i}}
\newcommand{\Al}{\mathtt{l}}
\newcommand{\Ar}{\mathtt{r}}
\renewcommand{\phi}{\varphi}

\newcommand{\mkrule}[1]{{\scriptsize\ensuremath{\mathsf{(#1)}}}\xspace}

\newcommand\rweaknatl{\mkrule{WNat_l}}
\newcommand\rweakbooll{\mkrule{WBool_l}}
\newcommand\rax{\mkrule{Ax}}
\newcommand\rcut{\mkrule{Cut}}

\newcommand\rweak{\mkrule{W}}

\newcommand\rbot{\mkrule{\bot}}

\newcommand\rtop{\mkrule{\top}}
\newcommand\rsigma{\mkrule{\sigma}}

\newcommand\rmul{\mkrule{\mu_{l}}}
\newcommand\rmur{\mkrule{\mu_{r}}}
\newcommand\rmu{\mkrule{\mu}}
\newcommand\rnul{\mkrule{\nu_{l}}}
\newcommand\rnur{\mkrule{\nu_{r}}}
\newcommand\rnu{\mkrule{\nu}}

\newcommand\rt{\mkrule{\tau}}
\newcommand\rtp{\mkrule{\tau'}}

\newcommand\rwith{\mkrule{\with}}
\newcommand\rwithl{\mkrule{\with_1}}
\newcommand\rwithr{\mkrule{\with_2}}
\newcommand\rwithi{\mkrule{\with_i}}
\newcommand\rparr{\mkrule{\parr}}
\newcommand\rotimes{\mkrule{\otimes}}
\newcommand\rotimesl{\mkrule{\otimes_l}}
\newcommand\rotimesr{\mkrule{\otimes_r}}
\newcommand\rtensor{\mkrule{\otimes}}
\newcommand\rtensorl{\mkrule{\otimes_l}}
\newcommand\rtensorr{\mkrule{\otimes_r}}
\newcommand\roplus{\mkrule{\oplus}}
\newcommand\roplusl{\mkrule{\oplus_1}}

\newcommand\roplusrr{\mkrule{\oplus_{r2}}}
\newcommand\roplusr{\mkrule{\oplus_2}}
\newcommand\roplusi{\mkrule{\oplus_i}}
\newcommand\roplusj{\mkrule{\oplus_j}}
\newcommand\rone{\mkrule{\llone}}
\newcommand\ronel{\mkrule{\llone_l}}

\newcommand\rdai{\mkrule{\dai}}
\newcommand\rmcut{\mkrule{mcut}}
\newcommand\rmcutone{\mkrule{mcut_1}}
\newcommand\rmcuttwo{\mkrule{mcut_2}}
\newcommand\rmcutpar{{\scriptsize\ensuremath{\mathsf{mcut(\iota, \perp\!\!\!\perp)}}}\xspace}


%% file: macros-undecidability.tex
\newcommand{\pop}{\mathrm{pop}}
\newcommand{\push}{\mathrm{push}}

\newcommand{\N}{\mathbb N}

\newcommand{\A}{\mathcal A}
\newcommand{\Athread}{\A_\mathit{thread}}
\newcommand{\short}{\mathit{short}}
\newcommand{\effect}{\mathit{effect}}

\newcommand{\gray}[1]{\textcolor{gray}{#1}}
\newcommand{\lc}{\mathbf{l}}
\newcommand{\rc}{\mathbf{r}}

\newcommand{\muvee}{\mkrule{\mu,\parr}}
\newcommand{\nuvee}{\mkrule{\nu,\parr}}

\newcommand{\mm}{\mkrule{\mu}}
\newcommand{\vv}{\mkrule{\parr}}

\newcommand{\mnv}{\mkrule{\mu,\nu,\parr}}
\newcommand{\mnvw}{\mkrule{\mu,\nu, \parr, \tensor}}
\newcommand{\mnvwgray}{\mkrule{\gray{\mu},\nu, \parr, \gray{\tensor}}}

\newcommand{\mnvnv}{\mkrule{\mu,(\nu,\parr)^2}}
\newcommand{\mw}{\mkrule{\mu,\tensor}}
\newcommand{\mvw}{\mkrule{\mu,\parr,\tensor}}
\newcommand{\we}{\mkrule{\tensor}}

\newcommand{\rinf}{\mkrule{\infty}}
\newcommand{\rexp}{\mkrule{\exp}}
\newcommand{\rci}{\rc_{\mathsf{i}}}

\newcommand{\lci}{\lc_{\mathsf{i}}}

\newcommand{\mr}{\mkrule{\rc}}
\newcommand{\mri}{\mkrule{\rci}}

\newcommand{\ml}{\mkrule{\lc}}
\newcommand{\mli}{\mkrule{\lci}}

\newcommand{\mrip}{\mkrule{\rci'}}
\newcommand{\mlip}{\mkrule{\lci'}}

\newcommand{\Acut}{\mathrm{Acut}}
\newcommand{\rAcut}{\mkrule{\Acut}}
\newcommand{\rAA}{\mkrule{\parr_A}}

\newcommand{\piaux}{\pi_{\mathsf{aux}}}

\newcommand{\gA}{\gray{A}}
\newcommand{\gB}{\gray{B}}

\newcommand{\gG}{\gray{G}}

\newcommand{\Fl}{F_\lc}
\newcommand{\Fr}{F_\rc}
\newcommand{\Gl}{G_\lc}
\newcommand{\Gr}{G_\rc}
\newcommand{\Fll}{F_{\lc\lc}}
\newcommand{\Flr}{F_{\lc\rc}}
\newcommand{\Frl}{F_{\rc\lc}}
\newcommand{\Frr}{F_{\rc\rc}}
\newcommand{\Gll}{G_{\lc\lc}}
\newcommand{\Glr}{G_{\lc\rc}}
\newcommand{\Grl}{G_{\rc\lc}}
\newcommand{\Grr}{G_{\rc\rc}}

\newcommand{\Inc}{\mathrm{Inc}}
\newcommand{\Dec}{\mathrm{Dec}}
\newcommand{\T}{\mathrm{Test}}
\newcommand{\Act}{\mathrm{Act}}


%% file: abstract.tex
We generalize the validity criterion for the infinitary proof system of the multiplicative additive linear logic with fixed points. Our criterion is designed to take into account axioms and cuts. We show that it is sound and enjoys the cut elimination property. We finally study its decidability properties, and prove that it is undecidable in general but becomes decidable under some restrictions.

%
%

%% file: new-intro.tex
\section{Introduction}\label{sec:Intro}
\todoAS[inline]{
  -  add reference to Mints continuous cut elimination\\
  - take examples into account\\
- develop motivation}
Fixed point theory has proved to be a valuable tool in computer
science, in particular for reasoning formally about software systems.
It is pervasive in programming language semantics, 
concurrency, automata theory and software verification techniques.

In the setting of fixed-point logics, infinitary ({\it ie.} non-well\-founded) and circular 
proof systems~\cite{dax06fsttcs,janin95mfcs,santocanale02fossacs,brotherston11jlc} have received much 
attention in recent years.
Such proof systems allow non-wellfounded proof trees and impose some global 
validity condition in order to ensure soundness. Typically, it requires that
every infinite branch is supported by some \emph{thread} tracing some 
formula in a bottom-up manner and witnessing infinitely many progress points of 
a coinductive property.

\todoAS[inline]{Fixed points are also present in most programming languages as
recursive types. More interestingly, the Curry-Howard correspondence,
which allows to view proofs as programs and formulas as types,
has been extended in various ways to encompass fixed point types,
e.g.\ in System F extended with least and greatest fixed point 
types~\cite{mendler91apal}, 
in Coq's calculus of (co)\-inductive constructions~\cite{gimenez98icalp},
and in typing disciplines for functional reactive programming
based on the linear-time $\mu$-calculus~\cite{cave2014popl}.}

On the programming side of the Curry-Howard correspondence, 
fixed-point formulas correspond to inductive and coinductive types:
one finds programming languages equipped with 
(co)\-recursion constructs whose typing naturally reflects the Kozen-Park 
(co)\-induction rules~\cite{mendler91apal,cave2014popl}.
Writing programs in these systems may be difficult, as it involves
coming up with complex (co)\-invariants. These difficulties are only
partially lifted through the use of guarded 
(co)\-recursion or sized types~\cite{gimenez98icalp,abel2007mixed}
in Coq or Agda respectively.
Furthermore, (co)\-recursion involves a suspended computation which makes it 
difficult to analyze the behavior of a program.
As an alternative, one could naturally consider infinitary (or circular)
programs, equipped with a global validity condition ensuring that
they behave well -- in particular that they are terminating, or productive
for inhabitants of coinductive types.
There is surprisingly little work following this approach\footnote{
  We note the work of Hyvernat~\cite{hyvernat2014lmcs} whose use of 
  size-change termination can be seen as a form of validity checking.
  It would be interesting to
  precisely compare it with our style of circular
  proof systems.
}, and
foundations are missing.

This lack of studies can be understood from the fact that the aforementioned
infinitary proof systems for fixed point logics are all cut-free; hence,
the role of the validity condition in (syntactic) cut-elimination remains
unclear from these works. This shortcoming has been addressed first by Santocanale
and Fortier: in~\cite{fortier13csl} they consider an infinitary sequent
calculus for purely additive logic, featuring cuts and an extended
notion of validity, and they show that cuts can be eliminated from valid
proofs -- in that setting, cut-elimination is not terminating but productive,
and converges to a (valid) cut-free derivation.
A key insight of this work is that the same validity condition that ensures
soundness in infinitary proof systems also ensures the productivity of 
cut-elimination.
The result has been generalized later to the multiplicative and additive linear
logic, at the cost of a more complex argument, by Baelde, Doumane and Saurin
~\cite{BaeldeDS16}.
Through these syntactic cut-elimination results, infinitary proofs for
the multiplicative-additive $\mu$-calculus \muMALL are given a computational
content, which is an important first step towards an interesting Curry-Howard
correspondence for that logic.


Unfortunately, existing notions of validity impose a quite limited use of cuts 
in non-wellfounded proofs and many proofs that could be accepted as valid
are rejected.
In particular, this prevents writing circular proofs in a compositional manner, as exemplified in the following (supported by Figure~\ref{fig:drop}):

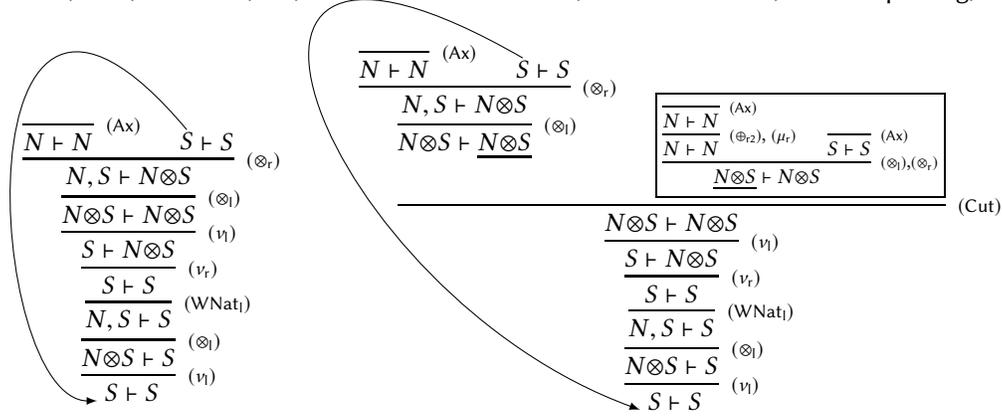
\begin{figure*}[t]
  $\scalebox{1}{
    $\begin{array}{cc}
      \begin{prooftree}
        \vspace{1cm}
        \Hypo{\begin{array}{c}
            ~\\
            ~\\
        \end{array}}
    \Infer{1}[\rax]{N \vdash N}
    \Hypo{
     \orig{introa}S \vdash S}
    \Infer{2}[\rotimesr]{N,S\vdash N\tensor S}
     \Infer{1}[\rotimesl]{N \tensor S\vdash N\tensor S}
    \Infer{1}[\rnul]{S\vdash N\tensor S}
    \Infer{1}[\rnur]{S\vdash S}
    \Infer{1}[\rweaknatl]{N, S\vdash S}
    \Infer{1}[\rotimesl]{N \tensor S\vdash S}
    \Infer{1}[\rnul]{
        \begin{tikzpicture}[remember picture,overlay]
      \node[inner sep=0pt,outer sep=0pt] (introb) {\ensuremath{}};
      \draw [->,>=latex] (introa.north) .. controls +(130:4.8cm) and +(180:1.5cm) .. (introb.west);
        \end{tikzpicture}~ S\vdash S
    }
      \end{prooftree}
&  \qquad
  \begin{prooftree}
    \Infer{0}[\rax]{N \vdash N}
    \Hypo{
      \orig{introc}S \vdash S}
    \Infer{2}[\rotimesr]{N,S\vdash N\tensor S}
    \Infer{1}[\rotimesl]{N \tensor S\vdash \underline{N\tensor S}}
\Hypo{\scalebox{.8}{\fbox{\begin{prooftree}\Infer{0}[\rax]{N\vdash N}
    \Infer{1}[\roplusrr, \rmur]{N\vdash N}
    \Infer{0}[\rax]{S\vdash S}
    \Infer{2}[\rotimesl,\rotimesr]{\underline{N\tensor S}\vdash N\tensor S}
  \end{prooftree}}}}
  \Infer{2}[\rcut]{N\tensor S\vdash N\tensor S}
    \Infer{1}[\rnul]{S\vdash N\tensor S}
    \Infer{1}[\rnur]{S\vdash S}
    \Infer{1}[\rweaknatl]{N, S\vdash S}
     \Infer{1}[\rotimesl]{N \tensor S\vdash S}
     \Infer{1}[\rnul]{
        \begin{tikzpicture}[remember picture,overlay]
      \node[inner sep=0pt,outer sep=0pt] (introd) {\ensuremath{}};
      \draw [->,>=latex] (introc.north) .. controls +(150:5.8cm) and +(160:4.5cm) .. (introd.west);
        \end{tikzpicture}~S\vdash S}
  
  \end{prooftree}
\end{array}$}
  $
  \caption{\label{fig:drop}Example of a valid and an invalid circular pre-proof.}
  \end{figure*}

\vspace{-.3cm}
\begin{example}\label{ex:streams}
  Consider formulas $N =\mu X. 1 \oplus X$ and $S = \nu Y. N\otimes Y$ respectively encoding natural numbers and streams of natural numbers in \muMALL.
Figure~\ref{fig:drop} presents two circular derivation trees, in the two-sided version of \muMALLi sequent calculus\footnote{We follow this convention for the example in order to exhibit more clearly the computational interpretation, even though the rest of the paper will be developed in the one-sided sequent calculus which is more concise.}. The computational interpretation of left-hand derivation is that of a function from streams of nats to streams of nats which drops its elements in odd position, keeping half of its elements only\footnote{Notice that natural numbers are erasable and duplicable as inputs in \muMALLi, hence the use of the $\rweaknatl$ admissible rule.}.
The rightmost proof has a slightly different computational interpretation: it drops one element every two but also increments the element it keeps in the output streams. This is achieved by using a cut with a proof precisely doing this increment (depicted in the box).
Although the leftmost proof is valid, the rightmost one is not for the validity condition introduced in~\cite{BaeldeDS16}: indeed, the cut that is introduced belongs to a cycle and as a result, no valid thread inhabits the infinite branch.

These examples correspond (at a somehow informal level) to the Coq coinductive terms {\tt drop} and {\tt incdrop} of Fig.~\ref{fig:guard}.
\begin{figure}[b]
  {\scriptsize\begin{lstlisting}[language=Coq]
CoInductive Stream := Cons : nat -> Stream -> Stream.
CoFixpoint drop (s : Stream) : Stream := match s with 
  | Cons a (Cons b s') => Cons b (drop s') end.
Definition hdinc (s: Stream) : Stream := match s with
  | Cons a s' => Cons (S a) s' end.
CoFixpoint incdrop (s : Stream) : Stream := match s with
  | Cons a (Cons b s') => hdinc (Cons a (incdrop s')) end.
                   (***********)
CoInductive BStream := BCons : bool -> BStream -> BStream.
Definition neghd (s: BStream) : BStream := match s with 
  | BCons a s' => BCons (negb a) s' end.
Definition k := 3.
CoFixpoint filter1everyk (m : nat) (s : BStream) : BStream :=
  match (m,s) with 
  | (0,BCons a s') => BCons a (filter1everyk k s')
  | (S m',BCons a s') => neghd (filter1everyk m' s')  end.
\end{lstlisting}}
  \caption{\label{fig:guard}Guarded and non-guarded examples of coinductive definitions.}\end{figure}
\end{example}

The present paper contributes to a line of research aiming at
providing and analyzing the computational content of circular
and non-wellfounded proofs, and improving their compositionality
of circular proofs.
From the Curry-Howard perspective, considering more relaxed validity
criteria is an interesting and important challenge as it affords a more
flexible way in constructing
circular proofs that, through the lense of Curry-Howard, more
flexibility to write valid programs on coinductive types. 

Indeed, while our previous cut-elimination result~\cite{BaeldeDS16}
is a significant step it goes only half-way due to
strong restrictions on the use of cuts along non-wellfounded
branches (or cycles in proofs) as described above. 
We  introduce here a new validity condition for \muMALLi,
the infinitary proof system for multiplicative additive linear logic with fixed 
points. Taking inspiration from Geometry of Interaction \cite{girardGoI89},  this criterion generalizes the existing one by enriching the structure of threads and relaxing their 
geometry: bouncing threads can leave the branch they validate and ``bounce''
({\it i.e.}\ change direction, moving upward but also downward along proof 
branches) on axioms and cut rules.

We define a new \emph{bouncing} validity condition in the multiplicative 
fragment and show that the obtained proofs enjoy cut elimination and soundness.
This new validity condition is undecidable but can be decomposed into an
infinite hierarchy of decidable conditions, via a parameter called ``height''. Moreover, it naturally extends
to a validity condition ensuring cut elimination for the multiplicative
and additive fragment.

The following example illustrates with a simpler example the intuitive idea
behind validation by  bouncing threads: it is not valid according to straight threads since its only infinite branch conains no infinite thread at all. On the other hand, one can trace the unfolding of the coinductive by following formulas $\nu X.X$ upwards and formulas $\mu X.X$ downwards while changing directions and moving from a formula to its dual when entering axioms and cuts, as represented on the portion of the proof that is represented. Indeed, after reducing twice the cut in each repetition of the cycle,
it yields a cut-free proof which is validated by a straight thread, which can be viewed as the ``straightened'' version of the above mentioned {\it bouncing thread}:
\vspace{-.2cm}
$$\scalebox{.8}{
  $\pi_0 =  \begin{prooftree}
  \Infer{0}[\rax]{\vdash \tikzmark{ad3}\nu X.X, \mu X.X\tikzmark{ad4}}
  \Infer{1}[\rmu]{\vdash \tikzmark{ad2}\nu X.X, \mu X.X\tikzmark{ad5}}
  \Hypo{\vdash \nu X.\orig{introA}X\tikzmark{ad7}}
  \Infer{1}[\rnu${\times}2$]{\vdash \nu X.X\tikzmark{ad6}}
  \Infer{2}[\rcut]{
    \vdash \tikzmark{ad1}\nu X.X
    \begin{tikzpicture}[remember picture,overlay]
      \node[inner sep=0pt,outer sep=0pt] (introB) {\ensuremath{}};
      \draw [->,>=latex] (introA.north) .. controls +(0:2.3cm) and +(-8:3cm) .. (introB.west);
    \end{tikzpicture}
  }
  \end{prooftree}
  \begin{tikzpicture}[overlay,remember picture,-,line cap=round,line width=0.1cm]
    \draw[rounded corners, smooth=2,red, opacity=.25] ($(pic cs:ad1)+(-.1cm,.1cm)$) to ($(pic cs:ad2)+(.1cm,-.2cm)$)to ($(pic cs:ad3)+(-.1,.2cm)$) to($(pic cs:ad4)+(-.2cm,.2cm)$) to($(pic cs:ad5)+(-.2cm,0cm)$) to($(pic cs:ad6)+(-.1cm,0cm)$) to($(pic cs:ad7)+(-.1cm,0cm)$) ;
  \end{tikzpicture}
  \;\;
  \begin{array}{cc}
    \rightarrow_{\mathsf{cut}}^2~
  &\begin{prooftree}
    \Hypo{\pi_0}
    \Infer{1}[\rnu]{\vdash \nu X.X}
  \end{prooftree}
  \\[16pt]~\rightarrow_{\mathsf{cut}}^\omega~
   & \begin{prooftree}
  \Hypo{\vdash \nu X.\orig{introa}X}
  \Infer{1}[\rnu]{
    \vdash \nu X.X
    \begin{tikzpicture}[remember picture,overlay]
      \node[inner sep=0pt,outer sep=0pt] (introb) {\ensuremath{}};
      \draw [->,>=latex] (introa.north) .. controls +(15:.8cm) and +(-5:.8cm) .. (introb.west);
    \end{tikzpicture}}
    \end{prooftree}
\end{array}      $}
$$

\emph{Summary of the contributions.}
We introduce a new decidability criterion for \muMALLi preproofs, based on bouncing threads. We prove that it guarantees soundness and productivity of the cut-elimination process. The criterion is compatible with simple compositions with cuts, as shown in \cref{ex:streams}. Contrarily to \cite{BaeldeDS16}, we also provide a treatment of axioms, which introduces substantial additional difficulties in the proof of soundness and cut-elimination.
We finally show that a parameter can be associated to our criterion (the ``height'' of the bouncing thread) and that every valid circular proofs is validated with a bounded height. moreover, we show that the criterion with fixed height is decidable in the multiplicative case, but that without bounding the height, the criterion becomes undecidable already in the multiplicative case.

\emph{Organization of the contributions.}
In Section~\ref{sec:ProofSystemMuMLL} we recall the basic definitions for the non-wellfounded proof system \muMALLi.
We then define, in \cref{sec:CutElimRules} the cut-elimination procedure. In \cref{sec:ValiditeMuMLL}, we introduce our new bouncing validity 
condition, and show in \cref{sec:muMLLCutElimThm} that it guarantees soundness of the system and productivity of the cut-elimination procedure. We finally study in \cref{sec:decid-body} the decidability of our criterion in the multiplicative case. Proofs and developments omitted due to space constraints can be found in appendix.

%% file: proofsystem.tex
\Section{The pre-proofs of \muMALLi }\label{sec:ProofSystemMuMLL}

In this section we introduce the multiplicative additive 
linear logic extended with least and greatest fixed point operators,
and a system of infinitary (pre-)proofs for that logic.

\begin{definition}\label{def:formulas}
Given infinite sets of \defname{atoms} $\C{A}=\{a,b,\dots \}$ and of \defname{fixed-point variables} $\C{V}=\{X,Y,\dots \}$,
\defname{\muMALLi-formulas} are built over the following syntax:
\renewcommand{\|}{\mathrel{|}}
$$
\begin{array}{rrll}   \null\hfill
  \phi,\psi &::=&
  a \| a^\bot \quad a\in {\C A}, &\quad \text{(atoms)}\\
  &\|& \mu X. \phi \| \nu X. \phi \| X\quad  X \in {\C V} &\quad \text{(fixed points)}\\
  &\|&  \bot \| \llone \| \phi \parr \psi \| \phi \tensor \psi &\quad \text{(multiplicatives)}\\
   &\| &\llzero \| \top \| \phi \oplus \psi \| \phi \with \psi. &\quad \text{(additives)}
\end{array}   $$
The connectives $\mu$ and $\nu$ bind the variable $X$ in $\phi$.
The formulas of \defname{\muMLLi}, the multiplicative fragment, are those \muMALLi formulas which do not contain $\with, \oplus, \top$ nor $\llzero$ ({\it ie.} restricting the grammar to the first two lines).
\end{definition}

\begin{definition}[Negation]
  $(\_)^\bot$
  is the involution on formulas satisfying:
  $$\begin{array}{c}a^{\bot\bot} = a;\  X^\bot = X;\  (\nu X. \phi)^\bot = \mu X. \phi^\bot;\\
   \bot^\bot = \llone;\ (\phi\parr\psi)^\bot =  \phi^\bot\tensor \psi^\bot;\ (\phi\oplus\psi)^\bot =  \phi^\bot\with \psi^\bot.
   \end{array}$$
%
\end{definition}

Setting $X^\bot=X$ would be incorrect when considering formulas with free
variables, but it yields the proper dualization for closed formulas,
e.g.\ $(\mu X.X)^\bot=\nu X.X$.
Note that, since negation is not a connective, our formulas enjoy the
positivity condition by construction: all fixed-point expressions are
monotonic.

There are several presentations of sequents in the literature: a sequent can 
be defined as a set of formulas, a multiset of formulas, a list of formulas or 
a set of named formulas. The first two presentations  (sets and multisets of 
formulas) are not suitable in a Curry-Howard perspective as they identify 
proofs having completely different computational behaviours. 
The last two presentations are the most used in the proofs-as-programs 
framework. Considering sequents as lists of formulas requires a constant use 
of the exchange rule, which is very heavy. In this paper, we made the choice 
to work with sequents as sets of named formulas, also called \emph{formula 
occurrences}.
We recall next their formal definition, in the style of \cite{BaeldeDS16}.

A formula occurrence is a formula together with an address. In a
derivation, all the conclusion (and cut) formula occurrences will have
pairwise distinct addresses.
When a rule is applied to a formula occurrence, the 
addresses of its sub-occurrences will be extended by $\{\Al,\Ar, \Ai\}$
(standing for left, right and inside respectively)
in order to record their provenance. 
This is of  great importance for our developments: our validity criterion 
traces the evolution of formulas, and this evolution is completely explicit in 
their addresses. 

\begin{example}\label{ex:ChoiceSequents} We show in the following an example of an application of the $\parr$ rule in the framework of sequents as sets, as multisets and as set of formula occurrences respectively: 
$$\begin{array}{ccc}
\scalebox{.9}{\begin{prooftree}
\Hypo{\vdash \phi}
\Infer{1}[\rparr]{\vdash \phi\parr\phi}
\end{prooftree}}
\quad&\quad
\begin{prooftree}
\Hypo{\vdash \phi, \phi}
\Infer{1}[\rparr]{\vdash \phi\parr\phi}
\end{prooftree}
\quad&\quad
\scalebox{.9}{\begin{prooftree}
\Hypo{\vdash \phi_{\alpha\Al}, \phi_{\alpha\Ar}}
\Infer{1}[\rparr]{\vdash (\phi\parr\phi)_\alpha}
\end{prooftree}}
\end{array}
$$
In  the first case, the two subformulas of $\phi\parr\phi$ collapse into one formula, in the second framework we keep track of the multiplicity but we cannot distinguish between the formula coming from the right and the one coming from the left. In the framework of formula occurrences, we can do this thanks to the tags $\Al$ and $\Ar$ in their addresses.
\end{example}

\begin{definition}
  Let $\FA$ be an infinite set of \defname{atomic addresses}, $\FA^\perp=\{\alpha^\perp\mid \alpha \in \FA\}$, and $\Sigma=\{\Al,\Ar,\Ai\}$.
  An \defname{address} is a word of the form $\alpha.w$, where $\alpha\in \FA\cup\FA^\perp$ and $w\in\Sigma^*$. Let us call $\Addr$ the set of addresses.
  We say that $\alpha'$ is a \defname{sub-address} of $\alpha$
  when $\alpha$ is a prefix of $\alpha'$,
  written $\alpha \sqsubseteq \alpha'$.
  We say that $\alpha$ and $\beta$ are \defname{disjoint} when
  $\alpha$ and $\beta$ are incomparable wrt.\ ${\sqsubseteq}$.
\end{definition}
The intuition is that atomic addresses and their duals will be assigned to the conclusions and cut formulas, and all the addresses appearing in our proofs will be sub-addresses of these addresses. 

\begin{definition}
  A \defname{formula occurrence}, or simply \defname{occurrence},
  is given by a formula $\phi$ and an address $\alpha$, and written 
  $\phi_\alpha$.
  Occurrences will be denoted by $F$, $G$, $H$.
  Occurrences are \defname{disjoint} when their addresses
  are.
  The occurrences $\phi_\alpha$ and $\psi_\beta$ are
  \defname{structurally equivalent},
  written $\phi_\alpha \equiv \psi_\beta$, if $\phi=\psi$.
\end{definition}


Note that the relation of sub-address is the inverse of the prefix relation. 
This may seem surprising, but it is coherent with the sub-formula relation. 
For instance, in \cref{ex:ChoiceSequents}, $\phi$ is a sub-formula of $\phi\parr\phi$, but its address is $\alpha\Ar$ while the address of $\phi\parr\phi$ is $\alpha$. 

We now define the rules of linear logic with fixed points in the framework of sequents as sets of occurrences.
 As seen in \cref{ex:ChoiceSequents}, a rule will look at the structure of the 
 formula underlying an occurrence, decompose it following a standard \muMALL 
 rule, then assign addresses to its subformulas in the obvious way. This necessitates always making the addresses explicit. In order to lighten 
 notation, we will have the syntax of \muMALLi operate directly on 
 occurrences:

\begin{definition}
  Logical connectives are lifted to operations on occurrences as:
\begin{itemize}
\item For any $\star\in\{\parr,\tensor,\oplus,\with\}$, if  $F = \phi_{\alpha 
  \Al}$ and  $G = \psi_{\alpha \Ar}$ then $F\star G = (\phi\star\psi)_\alpha$.
\item For any $\sigma\in\{\mu,\nu\}$, if $F = \phi_{\alpha\Ai}$ then 
  $\sigma X. F = (\sigma X. \phi)_\alpha$.
\end{itemize} 
 \end{definition}

%
%

\begin{definition}
  We define a duality over $\Addr$ by setting $(\alpha.w)^\perp=\alpha^\perp.w$ and $(\alpha^\perp.w)^\perp=\alpha.w$ for all $\alpha\in\FA$ and $w\in\Sigma^*$.
 We then define $(\phi_\alpha)^\bot = (\phi^\bot)_{\alpha^\bot}$,
 and write $F\bot G$ when $F^\bot=G$.
 We define substitution over occurrences as follows:
  $(\phi_\alpha)[\psi_\beta/X] = (\phi[\psi/X])_\alpha$.
\end{definition}





\begin{figure*}[t]
\def\infparr{\begin{prooftree}
\Hypo{\vdash F, G, \Gamma}
\Infer{1}[\rparr]{\vdash F\parr G, \Gamma}
\end{prooftree}}
\def\inftensor{\begin{prooftree}
\Hypo{\vdash F,\Gamma}
\Hypo{\vdash G,\Delta}
\Infer{2}[\rotimes]{\vdash F\tensor G, \Gamma, \Delta}
\end{prooftree}}
\def\infbot{\begin{prooftree}
\Hypo{\vdash \Gamma}
\Infer{1}[\rbot]{\vdash \bot, \Gamma}
\end{prooftree}}
\def\infone{\begin{prooftree}
\Hypo{}
\Infer{1}[\rone]{\vdash \llone}
\end{prooftree}}
\def\infmu{\begin{prooftree}
\Hypo{ \vdash F[\mu X.F/X], \Gamma}
\Infer{1}[\rmu]{ \vdash \mu X.F, \Gamma}
\end{prooftree}}
\def\infnu{\begin{prooftree}
\Hypo{\vdash G[\nu X.G/X], \Gamma}
\Infer{1}[\rnu]{\vdash \nu X.G, \Gamma}
\end{prooftree}}
\def\infax{
\begin{prooftree}
\Hypo{F \equiv G}
\Infer{1}[\rax]{\vdash F, G^\perp}
\end{prooftree}}
\def\infcut{
  \begin{prooftree}
  \Hypo{\vdash \Gamma, F}
  \Hypo{\vdash F^\perp, \Delta}
\Infer{2}[{\rcut}]{\vdash \Gamma, \Delta}
\end{prooftree}}
\def\infwith{\begin{prooftree}
\Hypo{\vdash F, \Gamma}
\Hypo{\vdash G, \Gamma}
\Infer{2}[\rwith]{ \vdash F\with G, \Gamma}
\end{prooftree}} 
\def\infplus{\begin{prooftree}
\Hypo{\vdash F_i, \Gamma}
\Infer{1}[\roplusi, $i\in\{1,2\}$]{\vdash F_1\oplus F_2, \Gamma}
\end{prooftree}}
\def\inftop{\begin{prooftree}
\Hypo{}
\Infer{1}[\rtop]{\vdash \top, \Gamma}
\end{prooftree}}
\def\infzero{\text{(no rule for $\llzero$)}}
$$\begin{array}{cccc|c|c}
   \infparr  & \infbot  & \infwith  & \inftop & \infnu\quad  &\infax \\[12pt]
 \inftensor  & \infone & \infplus  & \infzero\quad  & \infmu\quad  & \infcut\\[12pt]

 \end{array}$$
\vspace{-.5cm}  \caption{Rules of the proof system \muMALLi}  
\label{fig:muMALLi}
\label{fig:AdditiveRules}
\end{figure*}

We are now ready to introduce our infinitary sequent calculus.

\begin{definition}
  A \defname{sequent} is a set of disjoint occurrences.
  A \defname{\muMALLi pre-proof} is a possibly infinite tree, coinductively
  generated by the rules of Fig.~\ref{fig:muMALLi}.
%
%
  Given a sequent $s$ in a pre-proof $\pi$, we denote by $\prem s$
  the set of sequents which are premisses of the rule of conclusion $s$
  in $\pi$.
  Rules other than $\rax$ and $\rcut$ are called \defname{logical rules}.
  For every instance of one such rule
  we call \defname{principal occurrence} the occurrence in its conclusion
  sequent that is decomposed to obtain the premisses.
\end{definition}

The infinite derivations of \muMALLi may be quite complex trees,
possibly not even computable. In practical uses one would turn to sub-systems, 
typically the fragment of circular pre-proofs
\cite{fortier13csl,santocanale02fossacs,brotherston11jlc,doumane16lics}.
In a nutshell, a circular derivation is an infinite derivation which
has only finitely many distinct sub-trees up to renaming of 
addresses~\cite{DoumanePhD}. 

\begin{notation}[Two-sided notation]
  While it is proof-theoretical-ly convenient to work with one-sided sequents
  as in the previous definition, it is more illustrating for some examples,
  especially when we shall aim at illustrating computational interpretation
  of some proofs, to allow to use the usual two-sided sequent calculi. 
  In the following (and in the examples of the introduction), two sided
  sequents may be used:
\begin{center}   $F_1,\ldots,F_n\vdash\Gamma$ should be read as
  $\vdash \Gamma, F_1^\bot,\ldots,F_n^\bot$.
\end{center}
  Regarding the labelling of inference rules, we allow ourselves two conventions: either the inference rules are written with the labels introduced in Fig.~\ref{fig:muMALLi} or, as in the introductory example, we use their two-sided names, for instance $\rtensorl$ and $\rtensorr$, in which case this is a notation for the corresponding rule in the one-sided sequent calculus, respectively $\rparr$ and $\rtensor$ in this case.

  \end{notation}
\begin{example}
We show in Fig.\ref{fig:pidup} two examples of pre-proof,  $\pi_\mathsf{succ}$ and  $\pi_\mathsf{dup}$. (The reader can check their validity and cut-elimination behaviours once those concepts will have been introduced in the following sections.)
\begin{figure*}[t]
  $$
  ~\hspace{5cm}
  \pi_\mathsf{succ}\quad= \quad\begin{prooftree}
    \Infer{0}[\rax]{N \vdash N''}
    \Infer{1}[\roplusr]{N \vdash \llone\plus N''}
    \Infer{1}[\rmu]{N \vdash N'}
\end{prooftree} $$
  
  \vspace{.7cm}
  $$
  ~\hspace{3cm}
  \pi_{\mathsf{dup}} \quad=\quad   \begin{prooftree}
      \Infer{0}[\rmu,\roplusl,\rone]{\vdash N_1}
      \Infer{0}[\rmu,\roplusl,\rone]{\vdash N_2}
      \Infer{2}[\rbot,\rtensor]{\llone \vdash N_1\tensor N_2}
      \Hypo{
        \orig{opidup}{N' \vdash N'_1\tensor N'_2}}
      \Hypo{\pi_\mathsf{succ}}
      \Hypo{\pi_\mathsf{succ}}
      \Infer{2}[\rparr,\rtensor]{N'_1\tensor N'_2 \vdash N_1\tensor N_2}
      \Infer{2}[\rcut]{N' \vdash N_1\tensor N_2}
      \Infer{2}[\rnu,\rwith]{
        N \vdash N_1\tensor N_2
        \begin{tikzpicture}[remember picture,overlay]
        \node[inner sep=0pt,outer sep=0pt] (tpidup)  {\ensuremath{}};
        \draw [->,>=latex] (opidup.north east) .. controls +(40:6cm) and +(-5:10cm) .. (tpidup.east);
      \end{tikzpicture}\;}
    \end{prooftree}
  $$
  \caption{\label{fig:pidup}Examples of pre-proofs  $\pi_\mathsf{succ}$ and  $\pi_\mathsf{dup}$.}
\end{figure*}
  \end{example}

Pre-proofs are obviously
unsound: it is easy to derive the empty sequent. Hence, a validity
condition shall be required for a pre-proof to be called a proof.



%% file: cut-reductions.tex
\Section{The cut elimination process}\label{sec:CutElimRules}

In this section we introduce the cut-elimination rules for \muMALLi pre-proofs. 
In general, the cut-elimination procedure is not productive. 
However, we will show in Section~\ref{sec:muMLLCutElimThm} that when we restrict to valid pre-proofs (that will be defined in Section~\ref{sec:ValiditeMuMLL}), the process is productive and outputs a valid pre-proof.


\savespace
\subsection{The multicut rule}

In finitary proof theory, cut elimination may proceed by reducing topmost cuts. In the infinitary setting however, by non-wellfoundedness, there is no such thing, in general, as a topmost cut inference. 
In \cite{fortier13csl,BaeldeDS16}, this issue is dealt with by reducing 
\defname{bottom-most cuts}, and when encountering during the reduction a cut which is immediately above another one, instead of permuting two consecutive cuts, merging them into a new rule called 
\defname{multicut} and noted \rmcut. A multicut can be seen as a metarule to represent a finite tree of cuts.

We will also use this multicut approach\footnote{Note that there are various approach to cut-elimination in infinitary settings, for non-wellfounded derivations or for logics including an $\Omega$-rule, in paritcular Mints continuous cut-elimination~\cite{mints1978finite}.}, but we now have to deal with axiom/cut reductions. This leads us to enrich the structure of multicuts, by allowing those to perform a renaming.

A multicut is a rule written as: 
$$\scalebox{.9}{\begin{prooftree}
\Hypo{\vdash\Gamma_1\quad \dots\quad  \vdash\Gamma_n}
\Infer{1}[\rmcutpar]{\vdash\Gamma}
\end{prooftree}}$$
and comes with a function $\iota$ which shows how the occurrences of the conclusion are 
distributed over the premisses (modulo renaming), and a relation 
$\dbot$ specifiying which occurrences are cut-connected.
Below is an example of a multicut rule: the function $\iota$ is represented by the red lines, the relation bottom  is represented by the blue ones.
$$\scalebox{.9}{\begin{prooftree}
\Hypo{\vdash F'\tikzmark{PC1}, G}
\Hypo{\vdash G^\bot\tikzmark{PC2},H}
\Hypo{\vdash H^\bot\tikzmark{PC3},K}
\Infer{3}[\rmcutpar]{\vdash F\tikzmark{CC},K}
\end{prooftree}
\begin{tikzpicture}[overlay,remember picture,-,line cap=round,line width=0.1cm]
   \draw[rounded corners, smooth=2,red, opacity=.25] ($(pic cs:CC)+(-.2cm,.1cm)$)to ($(pic cs:PC1)+(-.2cm,-.1cm)$)to ($(pic cs:PC1)+(-.2cm,.1cm)$); 
   \draw[rounded corners, smooth=2,red, opacity=.25] ($(pic cs:CC)+(.3cm,.1cm)$)to ($(pic cs:PC3)+(.3cm,-.1cm)$)to ($(pic cs:PC3)+(.3cm,.1cm)$); 
  \draw[rounded corners, smooth=2,cyan, opacity=.25] ($(pic cs:PC1)+(.25cm,.1cm)$)to ($(pic cs:PC1)+(.25cm,-.1cm)$)to ($(pic cs:PC2)+(-.4cm,-.1cm)$) to ($(pic cs:PC2)+(-.4cm,.1cm)$);
    \draw[rounded corners, smooth=2,cyan, opacity=.25] ($(pic cs:PC2)+(.3cm,.1cm)$)to ($(pic cs:PC2)+(.3cm,-.1cm)$)to ($(pic cs:PC3)+(-.4cm,-.1cm)$) to ($(pic cs:PC3)+(-.4cm,.1cm)$);  
   \end{tikzpicture}
   }$$ 
Precise definitions and more explanations are given in 
\cref{app:multiCutRule}. Later, if clear from the context, we omit to specify $\iota$ and $\dbot$ in the rule name.

Now, we add the multicut rule to our proof system in order to perform cut-elimination.
  
\begin{definition}
We call \muMALLim the infinitary proof system obtained from \muMALLi by adding the multicut rule.
\end{definition}

%


\savespace
\subsection{Reduction rules and strategy}


The reduction rules are the same as in \cite{BaeldeDS16,fortier13csl}, adapting them in a straighforward way to account for the extra labellings $\iota,\dbot$ in multicut rules. We give examples of such reductions in this section.

There are two kinds of cut reductions: \emph{external}
ones that push the multicut deeper in the pre-proof (Example in \cref{fig:ExInternal}.a),
and \emph{internal} ones, that keep the multicut at the same level, and are not productive (Example in \cref{fig:ExInternal}.b). 
The rules in the first category are said to be \emph{productive}, since they contribute to the output of the process. Intuitively, the cut-elimination process succeeds if infinitely many productive rules occur on each branch of the proof.
%
%
An exhaustive description of the \muMALLim cut-reduction rules is given in
\cref{app:CutElimRules}.

\begin{figure*}[t]
(a) Example of an external reduction rule:
  \begin{center}
\scalebox{.9}{$\begin{array}{clc}
\begin{prooftree}
\Hypo{\C C}
\Hypo{\vdash \Delta, F'[\mu X. F'/X]}
\Infer{1}[\rmu]{\vdash \Delta, \mu X.\tikzmark{IRPl} F'} 
\Infer{2}[\rmcut]{\vdash \Sigma, \mu X.\tikzmark{IRCl} F}
\end{prooftree}
\quad&\quad\underset{}{\longrightarrow}\quad&\quad 
\begin{prooftree}
\Hypo{\C C}
\Hypo{\vdash \Delta,F'[\mu X.\tikzmark{IRPr} F'/X]}
\Infer{2}[\rmcut]{\vdash \Sigma, F[\mu X.\tikzmark{IRCr} F/X]}
\Infer{1}[\rmu]{\vdash \Sigma, \mu X. F}
\end{prooftree}   
\end{array}$}
\\[10pt]
%
\end{center}
\label{fig:ExExternal}
%
(b) Examples of internal reduction rules:
  \begin{center}

$$\begin{array}{ccc}
 \scalebox{.9}{ \begin{prooftree}
      \Hypo{\C C}
      \Hypo{\vdash \Delta,F}
      \Hypo{\vdash \Gamma,F^\bot}
\Infer{2}[\rcut]{\vdash \Delta,\Gamma}     
\Infer{2}[\rmcut]{\vdash \Sigma}
  \end{prooftree}}
  \quad&\quad
  \underset{}{\longrightarrow}
  \quad&\quad
  \begin{array}{c}
	\scalebox{.9}{    
    \begin{prooftree}
        \Hypo{\C C}
             \Hypo{\vdash \Delta,\tikzmark{MPC2}F}
      \Hypo{\vdash F^\bot,\tikzmark{MPC3}\Gamma}
\Infer{3}[\rmcut]{\vdash \Sigma}
    \end{prooftree}
        }

    \\
    \end{array}
\\
  \scalebox{.9}{
  \begin{prooftree}
    \Hypo{\C C}
        \Hypo{\vdash \Delta,F[\mu X.F/X]}
        \Infer{1}[\rmu]{\vdash \Delta,\mu X.\tikzmark{MPl1}F}
        \Hypo{\vdash F'^\bot[\nu X.F'^\bot/X],\Gamma}
        \Infer{1}[\rnu]{\vdash \nu X.F'^\bot,\tikzmark{MPl2}\Gamma} 
        \Infer{3}[\rmcut]{\vdash \Sigma}
  \end{prooftree}} 
  \quad&\quad
  \underset{}{\longrightarrow}
  \quad&\quad
  \scalebox{.9}{
  \begin{prooftree}
    \Hypo{\C C}
    \Hypo{\vdash \Delta,\tikzmark{MPr1}F[\mu X.F/X]}
    \Hypo{\vdash F'^\bot[\nu X.\tikzmark{MPr2}F'^\bot/X],\Gamma}
    \Infer{3}[\rmcut]{
      \vdash \Sigma}
  \end{prooftree}
  }
\end{array}  $$

  \end{center}
%
%
%
%
%
\caption{Examples of external and internal reduction rules.}
\label{fig:ExInternal}
\end{figure*}




%
We now describe a procedure to eliminate cuts from \muMALLi proofs, using as an intermediary framework the system with multicuts.
We start by embedding \muMALLi in \muMALLim by adding a unary multicut at the root of the pre-proof, with the identity as $\iota$ and $\dbot=\emptyset$. We then apply internal and external reduction rules to this multicut. 
We will require reduction sequences to be \defname{fair}, in the sense that every redex is eventually fired.

We introduce in the following section the validity condition, that will guarantee productivity of this cut elimination process.

%% file: mumll-validity.tex
\Section{Bouncing threads and pre-proof validity}\label{sec:ValiditeMuMLL}

We now formally introduce our bouncing threads and the corresponding
notion of validity for pre-proofs.
Given an alphabet $A$, we denote by $A^\omega$ the set of infinite
words over $A$, and define $A^\infty$ to be $A^* \cup A^\omega$.
We will make use of the letter $\lambda$ to denote ordinals in $\omega+1$,
i.e.\ either $\omega$ or a finite ordinal in $\mathbb{N}$.
For such an ordinal, recall that $1+\lambda=\lambda$ iff $\lambda=\omega$.
Finally, we will make use of a special concatenation:
given $u=(u_i)_{i\leq n<\omega}$ and $v=(v_i)_{i\in\lambda}$
such that $u_n = v_0$, we define $u\odot v$ as the standard concatenation
of $u$ and $v$ without its first element, i.e.\ $u \cdot 
(v_i)_{i\in\lambda\setminus\{0\}}$.
For example $aba\odot aab= abaab$.

\savespace\subsection{Threads}

We start with a naive notion of pre-thread,
defined as a sequence of pointed sequents (i.e. sequents with a marked formula) together with a direction:
a pre-thread follows occurrences in consecutive sequents,
travelling up- or downwards.

\begin{definition}\label{def:pre-thread}
  A \textbf{pre-thread} is a sequence $\psd{F_i}{s_i}{d_i}_{i\in\lambda}$ of tuples of a formula, a sequent and a direction, such that for all $i\in\lambda$,  $F_i \in s_i$, $d_i \in 
  \{\uparrow,\downarrow\}$ and if $i+1\in\lambda$ one of the following
  clauses holds:
  \begin{itemize}
    \item
      $d_i = d_{i+1} = {\uparrow}$,
      $s_{i+1} \in\prem {s_i}$, and $F_{i+1} \sqsubseteq F_i$;
    \item
      $d_i = d_{i+1} = {\downarrow}$,
      $s_i \in\prem{s_{i+1}}$,
      and $F_i \sqsubseteq F_{i+1}$;
      \item
        $d_i = {\downarrow}$, $d_{i+1} = {\uparrow}$,
      $s_i$ and $s_{i+1}$ are the two premisses of the same cut rule,
      and $F_i = F^\bot_{i+1}$;
    \item
      $d_i = {\uparrow}$, $d_{i+1} = {\downarrow}$ and
      $s_i = s_{i+1} = \{ F_i, F_{i+1} \}$ is the conclusion of an axiom rule (so that $F_i \equiv F^\bot_{i+1})$.

  \end{itemize}
  If $\lambda=n+1$ is finite we call $F_0$ and $F_n$ the \textbf{endpoints}
  of the pre-thread.
\end{definition}

\begin{example}\label{ex:pre-thread}
  Consider the formulas $\phi=\nu X. X$, $F=\phi_\alpha$, $F'=\phi_\beta$, $F''=\phi_{\beta.\Ai}$ where $\alpha$ and $\beta$ are disjoint addresses. Let $G, G'$ be two disjoint occurrences such that $G\equiv G'$. In the following pre-proof, the red and blue lines are two pre-threads\footnote{... which respectively correspond to the following sequences:\\
    $t_r = (F\parr G; \vdash F\parr G; \uparrow~\!)\cdot
 (F\parr G; \vdash F\parr G, F'^\perp\tensor G'^\perp; \uparrow~\!)\cdot
  (F; \vdash F, F'^\perp;\uparrow~\!)\cdot
  (F'^\perp; \vdash F, F'^\perp;\downarrow~\!)\cdot
  (F'^\perp\tensor G'^\perp; \vdash F\parr G, F'^\perp\tensor G'^\perp; \downarrow~\!)\cdot
  (F'\parr G'; \vdash F'\parr G'; \uparrow~\!)\cdot
  (F'; \vdash F', G'; \uparrow~\!)\cdot
  (F''; \vdash F'', G'; \uparrow~\!)$ and\\
  $t_b = (F\parr G; \vdash F\parr G; \uparrow~\!)\cdot
 (F\parr G; \vdash F\parr G, F'^\perp\tensor G'^\perp; \uparrow~\!)\cdot
  (G; \vdash G, G'^\perp;\uparrow~\!)\cdot
  (G'^\perp; \vdash G,G'^\perp;\downarrow~\!)\cdot
  (F'^\perp\tensor G'^\perp; \vdash F\parr G, F'^\perp\tensor G'^\perp; \downarrow~\!)\cdot
  (F'\parr G'; \vdash F'\parr G'; \uparrow~\!)\cdot
  (F'; \vdash F', G'; \uparrow~\!)\cdot
  (F''; \vdash F'', G'; \uparrow~\!)$.}:

\vspace{-.7cm}$$\scalebox{.9}{\begin{prooftree}
\Infer{0}[\rax]{\vdash \tikzmark{r3}F, F'^\bot\tikzmark{r4}}
\Infer{0}[\rax]{\vdash \tikzmark{b3}G, G'^\bot\tikzmark{b4}}
\Infer{2}[\mkrule{\parr,\tensor}]{\vdash F\tikzmark{r2}\parr G, F'^\bot\otimes\tikzmark{r5} G'^\bot}
\Hypo{\vdots}
\Infer{1}[]{\vdash \tikzmark{r8}F'', G'}
\Infer{1}[\rnu]{\vdash \tikzmark{r7}F', G'}
\Infer{1}[\rparr]{\vdash F'\tikzmark{r6}\parr G'}
\Infer{2}[\rcut]{\vdash F\tikzmark{r1}\parr G}  
  \end{prooftree}
  \begin{tikzpicture}[overlay,remember picture,-,line cap=round,line width=0.1cm]
    \draw[rounded corners, smooth=2,cyan, opacity=.25] ($(pic cs:r1)+(-.1cm,.1cm)$) to ($(pic cs:r2)+(.1cm,-.2cm)$)to ($(pic cs:r2)+(-.1,.2cm)$) to($(pic cs:r3)+(0,.2cm)$) to($(pic cs:r4)+(0,.2cm)$) to($(pic cs:r5)+(0,.1cm)$) to($(pic cs:r6)+(.1cm,.1cm)$) to($(pic cs:r7)+(0,.2cm)$) to($(pic cs:r8)+(0,.2cm)$);
  \draw[rounded corners, smooth=2,red, opacity=.25] ($(pic cs:r1)+(.2cm,.2cm)$) to ($(pic cs:r2)+(.2cm,-.15cm)$) to ($(pic cs:r2)+(.2cm,.2cm)$) to($(pic cs:b3)+(0,.2cm)$) to($(pic cs:b4)+(0,.2cm)$) to($(pic cs:r5)+(0,.2cm)$) to($(pic cs:r6)+(.2cm,.2cm)$) to($(pic cs:r7)+(.1cm,.2cm)$) to($(pic cs:r8)+(.1cm,.2cm)$);
  \end{tikzpicture}
}$$
\end{example}


 
We shall define threads as pre-threads satisfying a particular condition that will make them compatible with cut reduction, in the sense that they will have residuals after cut-elimination steps. In Example~\ref{ex:pre-thread}, the red thread has no residual
if one performs a cut elimination step on $F'\parr G'$, because it
comes from the right-hand subformula of  $F'^\perp \tensor G'^\perp$
and goes
to the left-hand subformula of $F' \parr G'$. 
In contrast,
the blue thread can meaningfully be simplified to persist over cut elimination
steps: its residual is well-defined. Geometry of Interaction \cite{girardGoI89} provides a formalization of
these notions, assigning weights to pre-threads and determining which weights
correspond to meaningful computations.
We follow this inspiration, adapting it to our framework.

\begin{definition}
  Let $t=\psd{F_i}{s_i}{d_i}_{i\in 1+\lambda}$ be a pre-thread.
  The \textbf{weight} of $t$
  is a word $(w_i)_{i\in \lambda}\in
  \{\Al, \Ar, \Ai, \bar{\Al}, \bar{\Ar}, \bar{\Ai}, \Wwait, \Aax, \Ccut\}^\infty$, written $\weight{t}$ and  defined as follows.
  For every $i\in\lambda$ one of the following clauses holds:
  \begin{itemize}
    \item  $w_i = x$ if $F_{i} = \phi_\alpha$ and $F_{i+1} = \psi_{\alpha x}$
      for $x \in \{\Al,\Ar,\Ai\}$;
    \item $w_i = \bar{x}$ if $F_{i} = \phi_{\alpha x}$ and $F_{i+1} = \psi_{\alpha}$
      for $x \in \{\Al,\Ar,\Ai\}$;
    \item $w_i = \Aax$
      if $d_{i} = {\uparrow}$ and $d_{i+1} = {\downarrow}$
      (corresponding to bouncing on an axiom rule);
    \item $w_i = \Ccut$
      if $d_{i} = {\downarrow}$ and $d_{i+1} = {\uparrow}$
      (corresponding to bouncing on a cut rule);
    \item $w_i = \Wwait$ if $F_{i} = F_{i+1}$.
  \end{itemize}
\end{definition}

The weight should be seen as a bracketed expression, where each symbol
$x\in \{{\Al},{\Ar},{\Ai}\}$ is an opening bracket with matching closing
bracket $\bar{x}$.
When defining threads from pre-threads,
we will be particularly interested in the following classes of
well-bracketed words:

\begin{definition}\label{def:thread}
  Let $\C{B}$ and $\C{H}$ be the set of words defined inductively as follows:
  $$
  \C{B}:= \Ccut \ \mid \ \C{B}\Wwait^*\Aax\Wwait^*\C{B} \ \mid
          \ \bar{x}\Wwait^*\C{B}\Wwait^*x
  \quad\quad\quad
  \C{H}:= \epsilon\ \mid \ \Aax\Wwait^*\C{B}
  $$
  A (finite) pre-thread is called a \textbf{$b$-path} if $\weight{t}\in\C{B}$.
  It is called an \textbf{$h$-path} if $\weight{t}\in\C{H}$.
\end{definition}

The $b$-paths start downwards and end upwards: they consist
of a series of U-shapes centered around cuts, glued together by axioms.
The endpoints of $b$-paths are negations of each other (up to renaming).
The $h$-paths start and end going upwards, and their endpoints are
structurally equivalent (up to renaming).
Intuitively, $h$-paths will be simplified during cut elimination,
and eventually disappear completely.

\begin{definition}
  A pre-thread $t$ is a thread when it can be written
  $\odot_{i\in 1+\lambda} (H_i\odot V_i)$ where for all $i\in 1+\lambda$:
  \begin{itemize}
  \item $\weight{V_i}\in \{\Al,\Ar, \Ai, \Wwait\}^{\infty}$
    and it is non-empty if $i\neq\lambda$;
  \item $\weight{H_i}\in \C{H}$ and it is non-empty if $i\neq 0$.
  \end{itemize}
  Notice that such a decomposition is unique.
  We call $(V_i)_{i\in 1+\lambda}$ the \defname{visible part} of $t$,
  and we denote it by $\visP{t}$,
  and $(H_i)_{i\in 1+\lambda}$ its \defname{hidden part}
  and we denote it by $\hiddenP{t}$.
  A thread is \defname{stationary} when its visible part is a finite sequence (of finite words), or when there exists $k\in 1+\lambda$
  such that $\weight{V_i}\in\{\Wwait\}^{\infty}$ for all
  $k\leq i \in 1+\lambda$.
\end{definition}


For instance if a pre-thread $t=(F_i,s_i,\uparrow)_{i\in\lambda_t}$ of length $\lambda_t$ goes only 
upwards with $\weight{t}\in\{\Al,\Ar,\Ai,\Wwait\}^{\lambda_t}$, then the above decomposition is given by $\lambda=0$, $H_0=(F_0,s_0,\uparrow)$ and $V_0=t$. 

\begin{example}
Let us consider the blue pre-thread of Example~\ref{ex:pre-thread}. We can decompose it into 
a visible part (plain line) and a hidden part (dashed line) as shown below:
$$\scalebox{.9}{\begin{prooftree}
\Infer{0}[\rax]{\vdash \tikzmark{ar3}F, F'^\bot\tikzmark{ar4}}
\Infer{0}[\rax]{\vdash \tikzmark{ab3}G, G'^\bot\tikzmark{ab4}}
\Infer{2}[\mkrule{\parr,\tensor}]{\vdash F\tikzmark{ar2}\parr G, F'^\bot\otimes\tikzmark{ar5} G'^\bot}
\Hypo{\vdots}
\Infer{1}[]{\vdash \tikzmark{ar8}F'', G'}
\Infer{1}[\rnu]{\vdash \tikzmark{ar7}F', G'}
\Infer{1}[\rparr]{\vdash F'\tikzmark{ar6}\parr G'}
\Infer{2}[\rcut]{\vdash F\tikzmark{ar1}\parr G}  
  \end{prooftree}
  \begin{tikzpicture}[overlay,remember picture,-,line cap=round,line width=0.1cm]
  \draw[rounded corners, smooth=2,cyan, opacity=.25] ($(pic cs:ar1)+(.2cm,.1cm)$)to ($(pic cs:ar2)+(.3cm,-.3cm)$) to ($(pic cs:ar2)+(.2cm,.1cm)$) to($(pic cs:ar3)+(0,.1cm)$);
   \draw[rounded corners, smooth=2,cyan, opacity=.25,dash pattern=on 3pt off 6pt]($(pic cs:ar3)+(.1cm,.1cm)$) to($(pic cs:ar4)+(0,.1cm)$) to ($(pic cs:ar5)+(0,.1cm)$) to($(pic cs:ar6)+(.1cm,.1cm)$) to($(pic cs:ar7)+(.1cm,0cm)$);
    \draw[rounded corners, smooth=2,cyan, opacity=.25] ($(pic cs:ar7)+(.1cm,.1cm)$) to($(pic cs:ar8)+(.1cm,.1cm)$);
  \end{tikzpicture}}$$
The blue pre-thread is then indeed a thread. On the contrary,
the red pre-thread from \cref{ex:pre-thread} admits no such decomposition.
\end{example}

If we consider the sequence of formulas followed by a non-stationary thread
on its visible part, ignoring its hidden parts (which have equivalent 
formulas on their endpoints), and skipping the steps in the visible parts
corresponding to $\Wwait$ weights, we obtain an infinite sequence of formulas
as in \cite{BaeldeDS16} where each formula is an immediate subformula or an
unfolding of the previous formula.
It is then well known \cite{DoumanePhD} that the formulas appearing infinitely often in 
that sequence admit a minimum w.r.t.\ the subformula ordering. We call this 
formula the \defname{minimal formula of the thread}.

%
  
\begin{definition} 
  A non-stationary thread is \defname{valid} if its minimal
  formula is a $\nu$-formula.
\end{definition}

Consider for example the formula $F=\mu X.\nu Y.X$. The minimal formula obtained by unfolding $F$ infinitely many times is $F$ itself, a $\mu$-formula, so the corresponding thread would be invalid.



\savespace\subsection{Pre-proof validity: the multiplicative case} \label{subsec:validity}

The previous notion of valid thread suggests a first extension of the
notion of valid proof based on straight threads~\cite{BaeldeDS16}:
one might say that a branch $\beta$ is valid when there is a valid bouncing thread 
which meets $\beta$ infinitely often,
and declare a pre-proof valid when all its branches are. 
%
However, this notion of \emph{weak validity} turns out to allow
unsound proofs, as shown next.

\begin{example}\label{ex:unsound}
  We set $T:=\nu X. X$ and $F:=\mu X. X$. The following is a weakly valid proof of the empty sequent. The hidden part of the decomposition
\vspace{-.3cm}
$$
\scalebox{.85}{\begin{prooftree}
	\Infer 0[\rax]{\vdash F_{\alpha \Al}, T_{\beta \Al}}
	
    \Infer 0[\rax]{\vdash T\tikzmark{un5}_{\alpha \Ar\Ai}, F\tikzmark{un6}_{\beta \Ar}}    
    \Infer 1[\rnu]{\vdash T\tikzmark{un4}_{\alpha \Ar}, F\tikzmark{un7}_{\beta \Ar}}
    
    \Infer 2[\rtensor]{\vdash (F\tensor\tikzmark{un3}T)_\alpha, 
    T_{\beta \Al}, F\tikzmark{un8}_{\beta \Ar}}
    \Infer 1[\rparr]{\vdash (F\tensor\tikzmark{un2}T)_\alpha, 
    (T\parr\tikzmark{un9}F)_{\beta}}
    \Hypo{
      \vdash (F\tensor \orig{a}T)_{\beta^\perp}}
    \Infer 2 [\rcut]{\vdash F\tensor\tikzmark{un1}T_\alpha
    \begin{tikzpicture}[remember picture,overlay]
    \node[inner sep=0pt,outer sep=0pt] (b)  {\ensuremath{}};
    \draw [->,>=latex] (a.north east) .. controls +(10:1.5cm) and +(-5:5cm) .. (b.north east);
  \end{tikzpicture}\;}
    \Infer 0[\rax]{ \vdash T_{\alpha^\perp \Al}, F_{\alpha^\perp \Ar}}
  \Infer 1[\rparr]{ \vdash (T\parr F)_{\alpha^\perp}}
\Infer 2[\rcut] {\vdash }
  \end{prooftree}
   \begin{tikzpicture}[overlay,remember picture,-,line cap=round,line width=0.1cm]
   \draw[rounded corners, smooth=2,green, opacity=.4] ($(pic cs:un4)$) to($(pic cs:un5)+(0cm,.1cm)$); 
 \draw[rounded corners, smooth=2,cyan, opacity=.4,dash pattern=on 3pt off 5pt]
 ($(pic cs:un1)$) to ($(pic cs:un2)$) to ($(pic cs:un3)$) to($(pic cs:un4)$);
 \draw[rounded corners, smooth=2,cyan, opacity=.4,dash pattern=on 3pt off 6pt]
($(pic cs:un5)+(0cm,.1cm)$) to($(pic cs:un6)+(0cm,.1cm)$) to($(pic cs:un7)$) to($(pic cs:un8)+(.1cm,0cm)$)to($(pic cs:un9)+(-.1cm,.1cm)$)to ($(a)$) to ($(a)+(0cm,.1cm)$);
  \end{tikzpicture}
  }
$$
\end{example}

A proper notion of validity must therefore be more constraining.
We shall consider the following one, which requires that the visible
part of the valid thread $t$ is contained in the infinite branch $\beta$.


\begin{definition} Let $\pi$ be a \muMLLi pre-proof. An infinite branch  $\beta$ of $\pi$ is said to be \defname{valid} if there is a valid thread $t$ starting  from one of its sequents, whose visible part is contained in this branch. 
\noindent  A \defname{\muMLLi proof} is a \muMLLi pre-proof in which every infinite branch is valid.
\end{definition}

\begin{example}\label{ex:validproofs}
  We show below examples of valid and invalid pre-proofs:
$$
\scalebox{.85}{\begin{prooftree}
  \Infer{0}[\rax]{\vdash (\nu X\tikzmark{bb2}. X)_\alpha, (\mu X.\tikzmark{bb3} 
  X)_\beta}
\Hypo{
  \vdash (\nu X\tikzmark{bb5}. \orig{a} X)_{\beta^\bot\Ai}}
\Infer{1}[\rnu]{\vdash (\nu X\tikzmark{bb4}. X)_{\beta^\bot}}
\Infer{2}[\rcut]{
  \vdash (\nu X\tikzmark{bb1}. X)_\alpha \begin{tikzpicture}[remember picture,overlay]
    \node[inner sep=0pt,outer sep=0pt] (b)  {\ensuremath{}};
    \draw [->,>=latex] (a.north east) .. controls +(20:1.2cm) and +(-5:5.5cm) .. (b.east);
  \end{tikzpicture}\;}
\end{prooftree}
}$$
$$\qquad 
\scalebox{.85}{\begin{prooftree}
  \Infer{0}[\rax]{\vdash (\nu X\tikzmark{rr3}. X)_{\alpha\Ai}, (\mu 
  X.\tikzmark{rr4} X)_\beta}
  \Infer{1}[\rnu]{\vdash (\nu X\tikzmark{rr2}. X)_\alpha, (\mu X.\tikzmark{rr5} 
  X)_\beta}
\Hypo{
  \vdash (\nu X\tikzmark{rr6}. \orig{aa}X)_{\beta^\bot}}
\Infer{2}[\rcut]{
  \vdash (\nu X\tikzmark{rr1}. X)_\alpha
  \begin{tikzpicture}[remember picture,overlay]
    \node[inner sep=0pt,outer sep=0pt] (bb)  {\ensuremath{}};
    \draw [->,>=latex] (aa.north east) .. controls +(20:1.2cm) and +(-5:5cm) .. (bb.east);
  \end{tikzpicture}\;
}
\end{prooftree}
 \begin{tikzpicture}[overlay,remember picture,-,line cap=round,line width=0.1cm]
   \draw[rounded corners, smooth=2,cyan, opacity=.25] ($(pic cs:bb1)+(0cm,.1cm)$) to ($(pic cs:bb2)+(.2cm,-.4cm)$) to ($(pic cs:bb2)+(-.2cm,.1cm)$) to ($(pic cs:bb2)+(.5cm,.5cm)$) to($(pic cs:bb3)+(0,.1cm)$) to($(pic cs:bb4)+(.1cm,0cm)$) to($(pic cs:bb5)+(0cm,0cm)$)to($(pic cs:bb5)+(.5cm,0.3cm)$); 
   \draw[rounded corners, smooth=2,red, opacity=.25] ($(pic cs:rr1)+(-.1cm,.2cm)$) to ($(pic cs:rr2)+(.2,-.4cm)$) to ($(pic cs:rr2)+(-.1,.2cm)$) to($(pic cs:rr3)+(0,.2cm)$) to($(pic cs:rr4)+(0,.2cm)$) to($(pic cs:rr5)+(0,.1cm)$) to($(pic cs:rr6)+(0cm,0cm)$) to($(pic cs:rr6)+(.1cm,0cm)$)to($(pic cs:rr6)+(.5cm,.3cm)$);

  \end{tikzpicture}
  }
$$
  The pre-proof on the left is valid: its infinite
  branch is supported by the valid blue thread,
  whose visible part belongs to the infinite branch.
  The right pre-proof is not valid, because the red thread,
  though valid, has a visible part that is not contained in
  the infinite branch.
\end{example}

\input{sec-additives}

%% file: sec-additives.tex
\subsection{Pre-proof validity: accomodating the additives}
\label{sec:Additives}

The previous definition of validity is too weak to ensure cut-elimination for \muMALLi,
%
%
which is not a strictly linear sequent calculus (as $\muMLLi$ is)
 since commutation/external reductions for the \rwith{} connective induce the duplication of a subproof. As a result, the extension of the validity condition in Section~\ref{subsec:validity} fails to ensure productivity and validity of cut-elimination as shown in figure~\ref{fig:additives}.(i).
\begin{figure*}[t]
$  \begin{array}{l|l|l}
    \hspace{-.6cm}    \pi_k =
    \scalebox{.85}{
      \begin{prooftree}
    \Hypo{}
    \Infer 1[\rax$\dagger$]{\vdash T\tikzmark{add16},T^\perp\tikzmark{add17}}
    \Infer 1[\rbot]{\vdash \bot,T\tikzmark{add15},T^\perp\tikzmark{add18}}
    \Hypo{}
    \Infer 1[\rax$\ddagger$]{\vdash T,T^\perp}
    \Infer 1[\rbot]{\vdash \bot,T,T^\perp}
    \Infer 1[\scriptsize{$\color{red}{\rmu^k}$}]{{\vdash \bot, T,T^\perp}}
    \Infer 2 [\rwith]{\vdash \bot\with\bot, T\tikzmark{add14},T^\perp\tikzmark{add19}}
    \Hypo{\vdash S,\orig{a}T\tikzmark{add20}}
    \Infer 2[\rcut]{\vdash \bot\with\bot, S,T\tikzmark{add13}}
    \Infer 1[\rmu,\rparr]{\vdash S,T\tikzmark{add12}}
    \Infer 1[\rnu]{\vdash S,T\tikzmark{add11}
          \begin{tikzpicture}[remember picture,overlay]
    \node[inner sep=0pt,outer sep=0pt] (b)  {\ensuremath{}};
    \draw [->,>=latex] (a.north east) .. controls +(20:1.2cm) and +(0:4.5cm) .. (b.east);
  \end{tikzpicture}\;}
  \end{prooftree} \begin{tikzpicture}[overlay,remember picture,-,line cap=round,line width=0.1cm]
   \draw[rounded corners, smooth=2,cyan, opacity=.25] ($(pic cs:add11)+(0cm,.1cm)$) to ($(pic cs:add12)+(-0.1cm,.1cm)$) to($(pic cs:add13)+(0,.1cm)$) to($(pic cs:add14)+(.1cm,-0.1cm)$) to($(pic cs:add15)+(0,-.1cm)$)to($(pic cs:add16)+(-0.1cm,.3cm)$)to($(pic cs:add17)+(-0.1cm,.3cm)$)to($(pic cs:add18)+(-.2cm,.1cm)$)to($(pic cs:add19)+(0,.1cm)$)to($(pic cs:add19)+(1cm,-.1cm)$)to($(pic cs:add20)+(-0.1cm,-.1cm)$)to($(pic cs:add20)+(0,.3cm)$); 
  \end{tikzpicture}}
\hspace{.2cm}    &
\qquad
    \scalebox{.85}{
\begin{prooftree}
    \Hypo{\orig{a}\vdash S,T}
\Infer 1[\rbot]{\vdash \bot,S,T}
    \Hypo{\vdash S,\orig{aa}T}
    \Infer 1[\rbot]{\vdash \bot, S,T}
    \Infer 2[\rwith]{\vdash \bot\with\bot, S,T}
    \Infer 1[\rmu,\rparr]{\vdash S,T
              \begin{tikzpicture}[remember picture,overlay]
    \node[inner sep=0pt,outer sep=0pt] (bb)  {\ensuremath{}};
    \draw [->,>=latex] (aa.north east) .. controls +(10:2cm) and +(0:4.5cm) .. (bb.east);
  \end{tikzpicture}\;}
    \Infer 1[\rnu]{          \begin{tikzpicture}[remember picture,overlay]
    \node[inner sep=0pt,outer sep=0pt] (b)  {\ensuremath{}};
    \draw [->,>=latex] (a.north east) .. controls +(160:2cm) and +(180:2.5cm) .. (b.east);
  \end{tikzpicture}\;\vdash S,T}
\end{prooftree}
} \hspace{1cm}
    &Sl(\pi_k) \ni
    \scalebox{.85}{
      \begin{prooftree}
    \Hypo{}
    \Infer 1[\rax]{\vdash T,T^\perp}
    \Infer 1[\rbot]{\vdash \bot,T,T^\perp}
    \Infer 1[\scriptsize{$\rmu^k$}]{{\vdash \bot, T,T^\perp}}
    \Infer 1
           [\rwithr]{\vdash \bot\with\bot, T,T^\perp}
    \Hypo{\vdash S,\orig{a}T}
    \Infer 2[\rcut]{\vdash \bot\with\bot, S,T}
    \Infer 1[\rmu,\rparr]{\vdash S,T}
    \Infer 1[\rnu]{\vdash S,T
          \begin{tikzpicture}[remember picture,overlay]
    \node[inner sep=0pt,outer sep=0pt] (b)  {\ensuremath{}};
    \draw [->,>=latex] (a.north east) .. controls +(20:1.2cm) and +(0:4.5cm) .. (b.east);
  \end{tikzpicture}\;}
  \end{prooftree}}    
    \\[-.5cm]
    (i) & (ii) & (iii)
  \end{array}
    $
  \caption{\label{fig:additives}
    (i) Pre-proof family $(\pi_k)_{k\in\mathbb{N}}$ with  $S= \mu Y.((\bot\with\bot)\parr Y), T=\nu X.X$. Note that we omit the occurrences and that $k$ is a parameter fixing how many times the $\mu$ rule (in red) should be applied to the sequent $\vdash \bot, T,T^\perp$. (ii) Result of applying (infinitary) cut-elimination to $\pi_1$. (iii) Example of a slice of $\pi_k$.}
  \end{figure*}
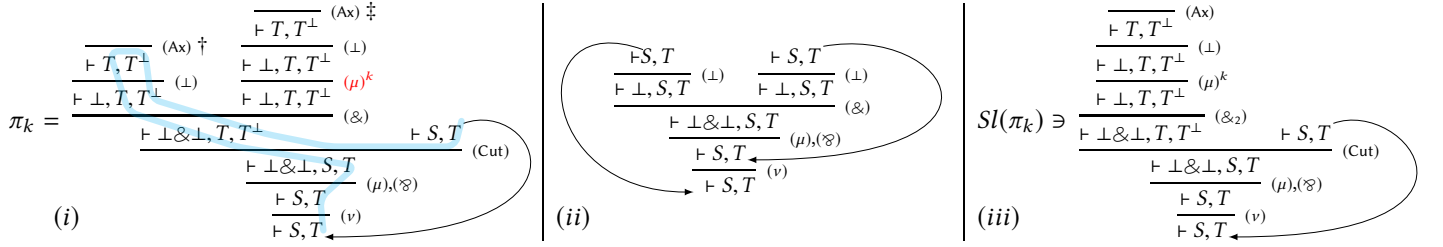
The result of cut-elimination on the proofs in the sequence $(\pi_k)_{k\geq 0}$ can be split into the following cases:
\\
(i) from $\pi_0$, cut-elimination is productive and produces a valid cut-free proof;\\
(ii) from $\pi_1$, cut-elimination  produces an \emph{invalid} pre-proof (see Figure~\ref{fig:additives}.(ii)): 
 any infinite branch following only finitely many times the left back-edge is invalid;
\\
(iii) from $\pi_k$, for $k\geq 2$ it is not even productive.
Indeed, in these examples, each $\pi_k$ contains exactly one infinite
branch which is supported by a thread on $T$ bouncing on the left-most
axiom and this thread is valid.
\hide
Still, this does not ensure productivity nor validity.
For instance, $\pi_1$ reduces to the following pre-proof:
$$  \begin{prooftree}
    \Hypo{\orig{a}\vdash S,T}
\Infer 1[\rbot]{\vdash \bot,S,T}
    \Hypo{\vdash S,\orig{aa}T}
    \Infer 1[\rbot]{\vdash \bot, S,T}
    \Infer 2[\rwith]{\vdash \bot\with\bot, S,T}
    \Infer 1[\rmu,\rparr]{\vdash S,T
              \begin{tikzpicture}[remember picture,overlay]
    \node[inner sep=0pt,outer sep=0pt] (bb)  {\ensuremath{}};
    \draw [->,>=latex] (aa.north east) .. controls +(10:2cm) and +(0:4.5cm) .. (bb.east);
  \end{tikzpicture}\;}
    \Infer 1[\rnu]{          \begin{tikzpicture}[remember picture,overlay]
    \node[inner sep=0pt,outer sep=0pt] (b)  {\ensuremath{}};
    \draw [->,>=latex] (a.north east) .. controls +(160:2cm) and +(180:2.5cm) .. (b.east);
  \end{tikzpicture}\;\vdash S,T}
  \end{prooftree}
  $$
which is not valid since any infinite branch which takes only finitely many times the left back-edge is invalid. 
\endhide

To understand the problem, consider the first step of cut-reduction (from $\pi_k$, for any $k$):
it is a $\rcut/\rwith$ commutation step, which copies the right-premiss of
the cut ({\it ie.} the
non-wellfounded part of the proof): after this step, the pre-proof contains
two infinite branches, but only one thread to validate them.
While the leftmost copy can be validated by the original thread, the
rightmost copy does not contain a residual of the original thread. Of course,
one might consider a thread originated in the cut inference, but that
will not suffice to ensure validity, nor productivity, as $\pi_2$ examplifies: its rightmost branch produces bottom rule.

\subsubsection{Sliced proof system and its cut-reduction}

This issue is solved by refining the criterion using slices~\cite{1987:Girard:LL,girard01locus,HvG,Terui11} and requiring that there exists
a supporting thread not only for every infinite branch of the proof, but also for every infinite branch \emph{of every persistent slice} of the pre-proof. 
In linear logic, an additive slice is a subtree of a sequent proof obtained by removing, for any of its $\rwith$ inference, the subtree rooted in one of its premisses (see Appendix~\ref{app:Additives} for details and precise definitions).

\begin{definition}
  \muSMALLi is obtained by extending \muMALLi with the following three inference rules:
  $$\begin{prooftree}
    \Hypo{\vdash A, \Gamma}
    \Infer 1[\rwithl]
           {\vdash A\with B, \Gamma}
  \end{prooftree}
  \qquad \begin{prooftree}
    \Hypo{\vdash B, \Gamma}
    \Infer 1[\rwithr]
           {\vdash A\with B, \Gamma}
  \end{prooftree}
  \qquad
  \begin{prooftree}
    \Hypo{}
    \Infer 1[$\rdai$]{\vdash \Gamma}
  \end{prooftree}
  $$
\end{definition}

\begin{definition}[Additive slice]
  \defname{Partially sliced} pre-proofs are the non-wellfounded  
  \muSMALLi pre-proofs. 
A(n additive) \defname{slice} is a $\rwith$-free, $\rdai$-free, \mupSMALLi-preproof. \end{definition}

    To a $\muMALLi$ sequent (pre-)proof, one can associate a set of slices
    by keeping, for each $\rwith$ inference, only one of its premisses and replacing the $\rwith$ with  the corresponding inference in $\rwithl,\rwithr$
    by applying corecursively one of the following two reductions:
$$    \scalebox{.9}{
$
    \begin{prooftree}
      \Hypo{\pi_1}
      \Infer 1{\vdash A_1, \Gamma}
      \Hypo{\pi_2}
      \Infer 1{\vdash A_2, \Gamma}
      \Infer 2[\rwith]{\vdash A_1\with A_2, \Gamma}
    \end{prooftree}
    ~\longrightarrow~
    \begin{prooftree}
      \Hypo{\pi_1}
      \Infer 1{\vdash A_1, \Gamma}
      \Infer 1[\rwithl]
             {\vdash A_1\with A_2, \Gamma}
        \end{prooftree}, \begin{prooftree}
      \Hypo{\pi_2}
      \Infer 1{\vdash A_2, \Gamma}
      \Infer 1[\rwithr]
             {\vdash A_1\with A_2, \Gamma}
        \end{prooftree}
$}$$
More precisely: 

\begin{definition}[Slicing of a pre-proof]
  The set of \defname{slices of $\pi$}, $Sl(\pi)$, is defined corecursively by
$$
\scalebox{.9}{$
Sl\left(\begin{prooftree}
      \Hypo{\pi_1}
      \Infer 1{\vdash A_1, \Gamma}
      \Hypo{\pi_2}
      \Infer 1{\vdash A_2, \Gamma}
      \Infer 2[\rwith]{\vdash A_1\with A_2, \Gamma}
    \end{prooftree}
\right) =
\left\{\begin{prooftree}
\Hypo{\pi'_i}
      \Infer 1{\vdash A_i, \Gamma}
      \Infer 1[\rwithi]
             {\vdash A_1\with A_2, \Gamma}
\end{prooftree},
\begin{array}{c}
  \pi'_i\in Sl(\pi_i),\\
    i\in\{1,2\}
\end{array}
\right\}
$}
$$        
(The other inferences are treated homomorphically.)
\end{definition}

\begin{example}
  Fig.~\ref{fig:additives}.(iii) gives an example of a slice.
\end{example}
  
\subsubsection{Cut-reductions for sliced proofs}
 
Cut-reduction rules for (partial) slices of $\muSMALLi$ extend those for
$\muMALLi$ with specific rules for sliced additives and $\rdai$.
Indeed, one may have a problematic situation, {\it e.g.} when a $\rwithl$ shall interact
 with a $\roplusr$: cut-elimination cannot be performed by relying on subproofs.  
%
%
%
%
%
%

 \begin{definition}[Cut reductions for slices]
   The sliced additive principal case is reduced as follows, 
   \hide
 \noindent\scalebox{.75}{
$\begin{prooftree}
\Hypo{\pi_i}
      \Infer 1{\vdash A_i, \Gamma}
      \Infer 1[\rwithi]
             {\vdash A_1\with A_2, \Gamma}
\Hypo{\pi'_j}
      \Infer 1{\vdash A^\perp_j, \Gamma}
      \Infer 1[\roplusj]
             {\vdash A^\perp_1\oplus A^\perp_2, \Delta}
             \Infer 2 [\rcut]{\vdash \Gamma, \Delta}
    \end{prooftree}
    $
    $\rightarrow \left\{
\begin{array}{cc}
   \begin{prooftree}\Hypo{}\Infer 1[$\rdai$]{\vdash \Gamma,\Delta}\end{prooftree} & \text{if }i\neq j\\
\begin{prooftree}
\Hypo{\pi_i}
      \Infer 1{\vdash A_i, \Gamma}
\Hypo{\pi'_i}
      \Infer 1{\vdash A^\perp_i, \Delta}
             \Infer 2 [\rcut]{\vdash \Gamma, \Delta}
    \end{prooftree}  & \text{if }i = j\\      
    \end{array}
    \right.$
 }
 \endhide
if $\{A^\perp_1\with A^\perp_2, {A'}_1\oplus {A'}_2\}\in\dbot$, with $r = (\mathsf{princ}, \{A^\perp_1\with A^\perp_2, {A'}_1\oplus {A'}_2\})$.
$$\begin{prooftree}
\Hypo{\C C}
\Hypo{\pi_i}
\Infer 1
{\vdash A^\perp_i, \Gamma}
\Infer 1[\rwithi]
{\vdash A^\perp_1\with A^\perp_2, \Gamma}
\Hypo{\pi'_j}
\Infer 1{\vdash {A'}_j, \Gamma}
\Infer 1[\roplusj]
{\vdash {A'}_1\oplus {A'}_2, \Delta}
\Infer 3 [{\scriptsize\ensuremath{\mathsf{mcut(\iota, \perp\!\!\!\perp)}}}\xspace]{\vdash \Sigma}
    \end{prooftree}
    $$

 $$\underset{r}{\longrightarrow}
 \left\{
\begin{array}{cc}
   \begin{prooftree}\Hypo{}\Infer 1[$\rdai$]{\vdash \Sigma}\end{prooftree} & \text{if }i\neq j\\[12pt]
\begin{prooftree}
\Hypo{\C C}
\Hypo{\pi_i}
      \Infer 1{\vdash A^\perp_i, \Gamma}
\Hypo{\pi'_i}
      \Infer 1{\vdash A'_i, \Delta}
             \Infer 3 [{\scriptsize\ensuremath{\mathsf{mcut(\iota, \perp\!\!\!\perp')}}}\xspace]{\vdash \Sigma}
\end{prooftree}  & \text{if }i = j\\
\text{where }\dbot'=\dbot\cup\{\{A^\perp_i,{A'}_i\}\}&
    \end{array}
    \right.$$
$$\text{or}\quad\begin{prooftree}
\Hypo{\C C}
\Hypo{}
\Infer 1[\rdai]
{\vdash \Gamma}
\Infer 2 [{\scriptsize\ensuremath{\mathsf{mcut(\iota, \perp\!\!\!\perp)}}}\xspace]{\vdash \Sigma}
    \end{prooftree}
\underset{r}{\longrightarrow}
\begin{prooftree}
\Hypo{}
\Infer 1[\rdai]
{\vdash \Sigma}
    \end{prooftree} \quad \text{with $r = (\mathsf{princ}, \dai)$.}$$
 \end{definition}

Notions of $b$-paths and $\epsilon$-paths can be naturally extended to additive slices.

\subsubsection{Persistent slices}

Persistent slices are introduced precisely as those in which no case of the above mismatch ever occurs:
 
 \begin{definition}[Persistent slice]
   Given a slice $\pi$, a $\rwithi$ rule of principal formula $A_1\with A_2$
   occurring in $\pi$ is said to be
   \defname{well-sliced} if
   no $b$-path starting down from the $A_1\with A_2$ occurrence of
   this sequent ends in a formula
   $A^\perp_1\oplus A^\perp_2$ that is the principal formula for a $\roplusj$ inference with $i\neq j$.
   A slice is \defname{persistent} if all its $\rwithi$ occurrences are well-sliced. \end{definition}

 \begin{example}
Pre-proof in Fig.~\ref{fig:additives}.(iii) is (obviously) a persistent slice.
 \end{example}

The following two properties of persistent slices are the key for the cut-elimination property:

 \begin{proposition}
   All reducts of a persistent slice are $\rdai$-free, and therefore are slices.
 \end{proposition}

 \begin{proposition}[Pull-back property]
   If $\pi \rightarrow^* \pi'$ (resp. $\pi \rightarrow^\omega \pi'$) and 
   $S' \in Sl(\pi')$, then there is a
   $S\in Sl(\pi)$ 
   such that $S\rightarrow^* S'$ (resp. $S \rightarrow^\omega S'$). 
\end{proposition}

\subsubsection{Additive validity}

Def~\ref{def:pre-thread} and~\ref{def:thread} of (pre-)threads directly adapt to the additives -- as they are not specific to \muMLLi -- and allow us to consider the following definition:

\begin{definition}
  A persistent slice is \defname{valid} if it is valid in the multiplicative sense\footnote{That is, every infinite branch of the slice is visited by a valid thread having its visible part contained in the branch.}.
  A \muMALLi pre-proof $\pi$ is \defname{valid} if all its persistent slicings are valid.
  \end{definition}

\begin{example}
\begin{figure*}[t]
  \vspace{.5cm}
  $$
\hspace{2cm}\pi_{\mathsf{filter}} = \quad
  \begin{prooftree}
    \Hypo{\pi_k}
    \Infer{1}{\vdash N}
    \Hypo{}
    \Infer{1}[\rax]{B \vdash B}
    \Hypo{\orig{adda} S \vdash S}
    \Infer{2}[\rnul,\rtensorl,\rnur,\rtensorr]{S\vdash S}
    \Infer{1}[\ronel]{1, S \vdash S}
    \Hypo{\orig{addb}{N, S \vdash S} }
    \Hypo{\pi_{neg}}
    \Infer{1}[]{B\vdash B}
    \Hypo{}
    \Infer{1}[\rax]{S\vdash S}
    \Infer{2}[\rnul,\rtensorl,\rnur,\rtensorr]{S\vdash S}
    \Infer{2}[\rcut]{N, S\vdash S}
    \Infer{1}[\rweakbooll]{N, B, S\vdash S}
    \Infer{1}[\rnul,\rtensorl]{N, S\vdash S}
    \Infer{2}[\rmul,\roplusl]{N, S\vdash S
      \begin{tikzpicture}[remember picture,overlay]
        \node[inner sep=0pt,outer sep=0pt] (addbb)  {\ensuremath{}};
        \draw [->,>=latex] (addb.north) .. controls +(50:7cm) and +(-5:11.5cm) .. (addbb.north east);
      \end{tikzpicture}\;}
    \Infer{2}[\rcut]{
      \begin{tikzpicture}[remember picture,overlay]
        \node[inner sep=0pt,outer sep=0pt] (addaa)  {\ensuremath{}};
        \draw [->,>=latex] (adda.north east) .. controls +(150:4cm) and +(190:5cm) .. (addaa.east);
      \end{tikzpicture}\;
      S\vdash S}
  \end{prooftree}
  $$
  \caption{\label{fig:pi-filter}Example of an additive circular proof}
\end{figure*}
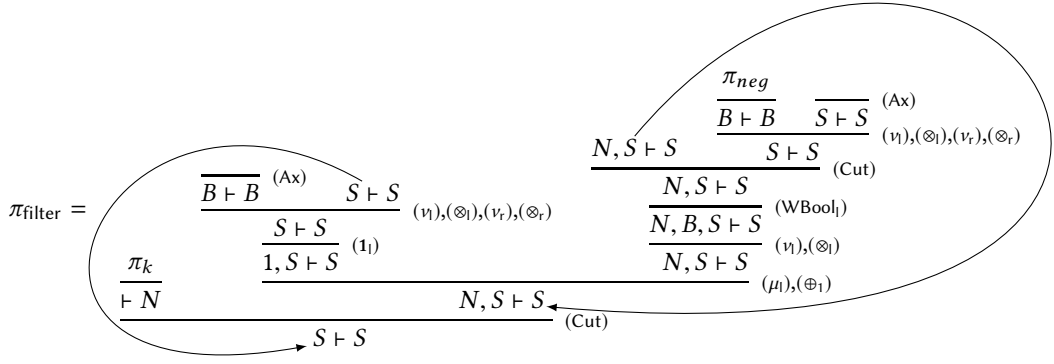

The circular pre-proof of Fig.~\ref{fig:pi-filter} is an example of a valid additive circular proof. It corresponds to the last program considered in the introduction.
\end{example}





%% file: bouncing-cut-elim.tex
\Section{Cut elimination theorem for \muMALLi}\label{sec:muMLLCutElimThm}

In this section, we shall establish our central result:

\restatabletheorem{th:mcutelim}{
Fair reduction sequences on \muMALLim proofs produce cut-free \muMALLi proofs.
}

For expository reasons, we focus on the multiplicative case here. The treatment of additives, while bringing new cases, is similar and can be found in Appendix~\ref{app:Additives}. 




The proof follows the same lines as the proof of cut elimination for straight threads \cite{BaeldeDS16}. 
We will only sketch it here and emphasize the new phenomena arising from the presence of axioms and bouncing threads.  The full proof can be found in Appendix~\ref{app:muMLLCutElimThm}.

The proof of Theorem~\ref{th:mcutelim} is in two parts.
We first prove that we cannot have a fair infinite reduction sequence made only of (unproductive) internal reductions.
Hence cut elimination is productive, \ie
reductions of \muMLLim proofs converge to cut-free \muMLLi pre-proofs.
We then establish that the obtained pre-proof is a valid proof.
In this section, we will only show productivity, validity of the resulting proof is shown in a similar way.

To show productivity, we proceed by contradiction, assuming that there exists a fair infinite sequence of internal reductions from a given proof $\pi$ of conclusion $\Gamma$. We will also assume w.l.o.g. that $\pi$ has only one multicut at the root. Note that since we perform only internal reduction rules, and since the latter do not duplicate multicuts, there is only one multicut progressing in the proof during this sequence of reductions. In the following, we refer to it as ``the'' multicut. 

\subsection{Trace of a reduction sequence} Let us start by introducing an important tool to analyse internal reduction sequences, called their \defname{trace}.  The trace of an internal reduction sequence is the set of sequents that occurred as a premise of the multicut rule during this reduction sequence. The conclusion sequent of the proof is additionally included in the trace.
By analyzing the reduction rules, it is easy to see that:

\begin{proposition}
  Given a $\muMLLi$ proof $\pi$, the trace of $\pi$ is a subtree (possibly with open leaves) of the original proof $\pi$: the trace of $\pi$ is the proof tree $\pi$ from which some branches may have been pruned and replaced by open leaves.
\end{proposition}

An example of a trace is shown below: sequents not in the trace are grayed.
\vspace{-.3cm}$$
\scalebox{.9}{
\begin{prooftree}
\Hypo{\vdots}
\Infer{1}[\rmu]{\vdash \mu X. X_{\beta\Ai}}
\Infer{1}[\rmu]{\vdash \mu X. X_\beta}

\Hypo{\vdots}
\Infer{1}[\rnu]{\vdash \mu X. X_{\beta^\bot\Ai\Ai}, \mu X. X_\gamma}
\Infer{1}[\rnu]{\vdash \mu X. X_{\beta^\bot\Ai}, \mu X. X_\gamma}

\Hypo{\textcolor{gray}{\vdots}}
\Infer{1}[\rnu]{\textcolor{gray}{\vdash \nu X. X_{\gamma^\bot\Ai}, \bot_\alpha}}
\Infer{1}[\rnu]{\vdash \nu X. X_{\gamma^\bot}, \bot_\alpha}

\Infer{2}[\rcut]{\vdash \nu X. X_{\beta^\bot\Ai}, \bot_\alpha}
\Infer{1}[\rnu]{\vdash \nu X. X_{\beta^\bot}, \bot_\alpha}

\Infer{2}[\rmcut]{\vdash \bot_\alpha}
\end{prooftree}
}$$


Before going further let us see how we will use the trace to get a contradiction:
\begin{itemize}
\item  We will define an extension of the proof system \muMLLi, and show that it is sound with respect to a boolean semantics. 
\item  Then we will show that the trace can be seen as a proof of a false sequent in this extended proof system. 
\item  This contradicts soundness and concludes the proof.
\end{itemize}

\savespace\subsection{The trace is almost a \muMLLi proof} As said above, we will need to see the trace as a genuine proof. In fact, it is almost a \muMLLi proof since it is a subtree of the original proof $\pi$, but is not completely a proof for the following reasons:
\begin{itemize}
\item The trace may have unjustified sequents:
this happens when a sequent $S$ enters the multicut during the reduction
sequence but never gets reduced. It will then be part of the trace but
the subtree of $\pi$ rooted in $S$ will not. This is for instance the case of the sequent $\vdash \nu X. X _{\gamma^\bot\Ai},  \bot_\alpha$ in the example above.
\item There is another reason why the trace might not
be a proof: its infinite branches may not be valid. The infinite 
branches of the trace are also infinite branches of the proof $\pi$, thus they 
are supported by valid bouncing threads of $\pi$. However, since the threads are bouncing, they might leave the branch and thus not be included in the trace.
\end{itemize}
We will show later how to handle the first problem of unjustified sequents. As for the second problem, we show that this actually never happens:
\restatableproposition{prop:ThreadsInTrace}{
Let $T$ be the trace of a reduction sequence starting from a proof $\pi$,
and let $\beta$ be an infinite branch of $T$. If $t$ is a bouncing thread of $\pi$ 
validating $\beta$, then $t$ is also a bouncing thread of $T$. 
}
This is one of the difficulties specific to the bouncing threads. This result is trivial with straight threads \cite{BaeldeDS16}, since threads belong to the branch they support.

\savespace\subsection{Truncated proof system}
To see the trace as a proof, we need to overcome the problem of unjustified sequents. For that, we will embed the trace in a proof system extending \muMLLi, called the \defname{truncated proof system}.

This proof system is parameterized by a partial function $\tau:\Addr \to \{\top, \llzero\}$ (from addresses to the formulas $\top, \llzero $) called a \defname{truncation}. To get a sound proof system, we impose a coherence condition on truncations: they should assign dual values to dual addresses. 
The rules of the truncated proof system are the same as those of $\muMLLi$, with an extra rule
which allows to replace an occurrence by its image in $\tau$. Pre-proofs and the validity condition are defined in the same way as \muMLLi.
The advantage of the truncated proof system is that it allows to close sequents easily:  if the address of an occurrence of the sequent is mapped to $\top$  by $\tau$, we can justify the sequent by a $\top$ rule.

The boolean semantics can be extended in the presence of truncations in a natural way: the occurrences whose addresses are in the domain of the truncation obtain as a boolean value their image by $\tau$. The rest of the boolean values are propagated through the connectives in the usual way.

 We show that the truncated proof system is sound for this semantics. 
Note that \muMLLi can be seen as a truncated proof system, where the truncation has empty domain. In this case the truncated boolean semantics coincides with the classical boolean semantics. Hence \muMLLi is sound for the boolean semantics.
\begin{theorem}
The proof system \muMLLi is sound for the boolean semantics. 
\end{theorem}

\savespace\subsection{Trace as a truncated proof} Let us see how to transform the trace into a proof in a truncated proof system. For this , we need to find a truncation $\tau$ that can allow us to close every unjustified sequent. In other words, we need to find a strategy for selecting an occurrence in each unjustified sequent, to which we will assign $\top$ by the truncation $\tau$. This strategy should define a coherent truncation in the sense that it should not assign $\top$ to two dual occurrences. 

In \cite{BaeldeDS16}, we have given such a strategy: we select the occurrence of the unjustified sequent which is principal in the proof $\pi$. The presence of axioms complicates the situation:  If $F$ is the occurrence that has been selected in an unjustified sequent, then its dual might appear in an axiom rule $\vdash F^\perp, G$.
 By coherence of the truncation, the address of $F^\perp$ must have been assigned $\llzero$ and the axiom rule cannot be soundly applied anymore. To justify the sequent $\vdash F^\perp, G$, we need to assign $\top$ to the address of $G$. Since the same can happen on the $G$ side, we need to show that it remains possible to define $\tau$ in a coherent way.

To get our desired contradiction, we need in addition for $\tau$ to assign $\llzero$ to the conclusion, thereby obtaining a proof of a false sequent. This needs to be done while still respecting the aforementioned constraints induced by axioms. 
\todoAS{préciser l'explication du paragraphe précédent: comment assigner 0.}

\savespace\subsection{Summary} To sum up, we have found a truncation $\tau$ i) which assigns $\llzero$ to the conclusion formula ii) for which the trace can be seen as a proof in the corresponding truncated proof system. Since the proof system is sound, we get a contradiction and this concludes the proof of productivity. 

%% file: decid-body.tex
\Section{Decidability properties of \muMLLo}\label{sec:decid-body}

\subsection{An operational approach to threads}\label{sec:stack}

In this section, we will explicate how threads can be recognized by a specific deterministic pushdown automaton reading only the weight of a pre-thread.
This will allow us to define the \emph{height} of a thread and the notion of \emph{constraint stack}.

Let $\Athread$ be the deterministic pushdown automaton described in Fig. \ref{fig:Athread}, on alphabet $\Sigma=\{\Al,\Ar,\Ai, \bar{\Al},\bar{\Ar},\bar{\Ai},\Aax,\Ccut, \Wwait\}$ and stack alphabet $\Gamma=\{\Al,\Ar,\Ai,\bot\}$ where $\bot$ is the empty stack symbol. 
The transitions are labelled ``$(a,\gamma) \mid \tau$'', where $a\in\Sigma$ is the input letter, $\gamma\in\Gamma$ is the topmost stack symbol, and $\tau$ is the action performed on the stack (no action if $\tau$ is not specified). If no stack symbol is specified, the stack is left unchanged. Symbol $x$ stands for an element in $\{\Al,\Ar,\Ai\}$. No acceptance condition is specified, meaning that any run is accepting. Only the absence of an available transition can cause the automaton to reject its input, for instance reading $\Al$ with topmost stack symbol $\Ar$ in state $\uparrow$.
The transition marked with a double arrow corresponds to the visible part of the thread.

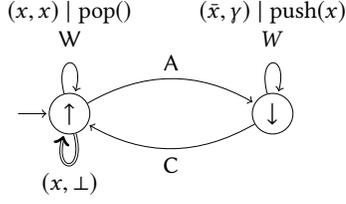
\begin{figure}[t]
\centering
\scalebox{.9}{
\begin{tikzpicture}[shorten >=1pt,node distance=3cm,on grid,auto,initial text=,
every state/.style={inner sep=0pt,minimum size=6mm}, every loop/.style={min distance=.5cm, looseness=10}]

\node[state,initial] (up) {$\uparrow$};
\node[state, right=of up] (down) {$\downarrow$};
\path[->]
	(up) edge[loop below,double] node{$(x,\bot)$} ()
	edge[loop above] node{$\begin{array}{c}(x,x)\mid\pop()\\\Wwait\end{array}$} ()
	edge[bend left] node {$\Aax$} (down)
	(down) 	edge[loop above] node{$\begin{array}{c}(\bar{x},\gamma)\mid\push(x)\\W\end{array}$} ()
	edge[bend left] node {$\Ccut$} (up)
	;
\end{tikzpicture}
}
\caption{The deterministic pushdown automaton $\Athread$}
\label{fig:Athread}
\end{figure}

\begin{lemma}
Let $t$ be a pre-thread. Then $t$ is a thread if and only if $\weight{t}$ is accepted by $\Athread$.
\end{lemma}

\begin{proof}
The constraints on the stack match the grammar from Def. \ref{def:thread}. 
\end{proof}

The stack of $\Athread$ will be referred to as the \emph{constraint stack}.

\savespace\subsection{Undecidability of bouncing validity}
\label{sec:undecidability}

In this section, we sketch why the validity condition is already undecidable for
 \muMLLo. This will motivate the following section introducing decidable subcriteria constituting a hierarchy of criteria while exact definitions, encodings and proof of undecidability are postponed to  Section \ref{appsubsec:undecproofs} for readability.
%

To show undecidability, we reduce from the halting problem for Minsky Machines, i.e. two-counter machines (2CM) able to perform increment, decrement, and zero test on the counters. The halting problem for 2CM is known to be $\Sigma^0_1$-complete \cite{Minsky61}.
The proof is only sketched here, and some technicalities have been abstracted away for clarity purposes. See Section \ref{appsubsec:undecproofs} for exact definitions and encodings.

We encode the halting problem of a 2CM $M$ using a bouncing thread. The thread of interest will always follow a formula $F=\nu X.(X\parr X)$ when going upwards, and its dual $G=\mu X.(X\tensor X)$ when going downwards. The idea is to use the constraint stack to encode the value of counters, and the position in the graph to encode the control state of the machine. The general shape of the main preproof $P$ performing the desired reduction is represented Fig. \ref{fig:mainproofbody}. Boldface formulas are those introduced in cuts, and grayed formulas are the ones that are not part of the thread of interest. We will also ignore addresses in this proof sketch, except those relevant to our encoding.

\begin{figure}[t]
 $$\scalebox{.9}{   \begin{prooftree}
	\Infer{0}[\rax]{\vdash \gG, \gray{F}}
	
        \Hypo{
          \orig{a}\vdash \gG, \tikzmark{d12}F}
      \Infer{2}[\mvw]{\vdash \gG, F\parr\tikzmark{d11} F}
      \Infer{1}[\mkrule{\textcolor{red}{\nu}}]{\vdash \gG,\tikzmark{d10}\mathbf{F}}
     
      \Hypo{\pi_R\tikzmark{d7}}
      \Infer{1}{\vdash \tikzmark{d8}G,\mathbf{F}\tikzmark{d6}}      
     
	\Hypo{\pi_M\tikzmark{d4}}
	\Infer{1}{\vdash \tikzmark{d5}\mathbf{G},\tikzmark{d3}F}
 \Infer{2}[\rcut]{\vdash \tikzmark{d9}\mathbf{G},\tikzmark{d2}F}
 \Infer{2}[\rcut]{
    \begin{tikzpicture}[remember picture,overlay]
    \node[inner sep=0pt,outer sep=0pt] (b)  {\ensuremath{}};
    \draw [->,>=latex] ($(pic cs:d12)+(-.1cm,.3cm)$) .. controls +(150:4.2cm) and +(185:5.6cm) .. (b.west);
    \end{tikzpicture}\;
    \vdash \gG,\tikzmark{d1}F}
    \end{prooftree}
    \begin{tikzpicture}[overlay,remember picture,-,line cap=round,line width=0.1cm]
   \draw[rounded corners, smooth=2,cyan, opacity=.25] ($(pic cs:d1)+(.1cm,.1cm)$)to ($(pic cs:d2)+(-.2cm,.-.4cm)$)  to ($(pic cs:d2)+(.2cm,.2cm)$)  to($(pic cs:d3)+(.2cm,.1cm)$) to($(pic cs:d4)+(0cm,.2cm)$)to($(pic cs:d4)+(-.4cm,.2cm)$) to($(pic cs:d5)+(.2cm,.1cm)$)  to($(pic cs:d6)+(-.1cm,.1cm)$) to($(pic cs:d7)+(0,.2cm)$) to($(pic cs:d7)+(-.4cm,.2cm)$) to($(pic cs:d8)+(0,.2cm)$)  to($(pic cs:d9)+(.2cm,.1cm)$)  to($(pic cs:d10)+(0,.1cm)$) to($(pic cs:d11)+(-.1cm,0cm)$) to($(pic cs:d12)+(.1cm,.2cm)$);    
   \end{tikzpicture}
     }$$
 \caption{A sketch of the main preproof $P$}
 \label{fig:mainproofbody}
\end{figure}
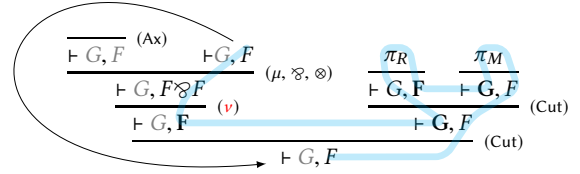

We build $P$ so that the only branch which is not clearly validated is the one going infinitely many times through the loop. A thread validating this branch (in blue in Fig. \ref{fig:mainproofbody}) must go through the two cuts, and bounce on axioms in $\pi_M$ and $\pi_R$.
The trajectory of this thread in $\pi_M$  will simulate the run of $M$, and it will be allowed to exit $\pi_M$ if and only if $M$ terminates.

We now give an example of one of the simplest gadgets used to perform this simulation: the increment gadget on the first counter.
Consider a state $p$ of the machine $M$, whose action is to increment the first counter and go to state $q$. Assume counter values $(n,m)$ are encoded by a constraint stack $\lc^n\rc\lc^m\rc$, where $\lc$ (resp. $\rc$) stands for a left (resp. right) constraint on the unfolding of $F$, i.e. a relative address $\Ai\Al$ (resp. $\Ai\Ar$). This means that to increment the first counter, we need to add a left constraint at the top of the stack.
This can be performed by the following gadget, where nodes labeled $(p)$ and $(q)$ encode the current control state:

\vspace{-.3cm}

$$\scalebox{.9}{
\begin{prooftree}
   \Hypo{(q)\vdash G,\mathbf{F},\gA}
   
   \Infer{0}[\rax]{\vdash G_\lc,F}
    	\Infer{0}[\rinf]{\vdash \gray{G_\rc,A}}
	\Infer{2}[\we]{\vdash G_\lc\tensor G_\rc,F,\gA}
	\Infer{1}[\mm]{\vdash \mathbf{G},F,\gA}
	
    \Infer{2}[\rAcut]{(p)\vdash G,F,\gA}
    \end{prooftree}}
$$

Here $A$ is an auxiliary formula $\nu X.(X\parr X)\tensor X$, that can be duplicated as required and used to build axiomless valid proofs, denoted by an $\infty$ metarule. The rule (Acut) denotes a cut combined with a duplication of $A$.
A thread entering node $(p)$ upwards with constraint stack $\lc^n\rc\lc^m\rc$ will enter node $(q)$ with constraint stack $\lc^{n+1}\rc\lc^m\rc$.
\smallskip

In order to fully simulate the run of $M$, we also need to design gadgets simulating increment on the second counter, as well as decrement and zero test on both counters.
The main difficulty lies in the tests performed by the machine: we want the thread to follow a conditional branching, depending on the value of the constraint stack. This can be done, but because of the linearity of the proof system, we cannot avoid leaving some extra constraints encoding the results of the tests. These ``garbage constraints'' will be collected by the thread on its path downwards in $\pi_M$, after the simulation of the machine is completed.
Since we want to finish with empty constraint, we need to erase these garbage constraints.
To do this, we add a second gadget $\pi_R$ performing the computation in a dual way: garbage constraints are fed to the thread, which rewinds the computation while erasing these unwanted constraints. All gadgets in $\pi_R$ are dual versions of those in $\pi_M$. This technique is reminiscent of the one used by Bennett \cite{Bennett73} to prove Turing-completeness of reversible Turing machines, where a history of the computation is produced to guarantee reversibility, then this history is erased by rewinding the computation.

We can finally exit this detour with no constraint, and perform a visible $\nu$-unfolding on the main branch (in red in Fig. \ref{fig:mainproofbody}), before looping back to the root of the proof.

The global pre-proof $P$ will be a valid proof according to the criterion if and only if the machine $M$ halts.

Notice that among the simplifications we made here for clarity of exposition, the auxiliary formula $A$ needed in some gadgets has been removed from the main preproof $P$.

\savespace\subsection{A hierarchy of decidable validity conditions}\label{sec:Decidability}

In order to recover a decidable criterion, we will consider restrictions on the constraint stack of valid threads.

\begin{definition}
If $t$ is a thread, we define its \emph{height} $h(t)\in\omega+1$ to be the supremum of the size of the stack of $\Athread$ along its run on $\weight{t}$.
\end{definition}

\begin{definition}
Let $k\in\N$. An infinite branch is \emph{$k$-valid} if there is a thread of height at most $k$ validating it. A proof $P$ is a \emph{$k$-proof} if every infinite branch of $P$ is $k$-valid.
\end{definition}


The two following theorems show that the height parameter $k$ induces a hierarchy of decidable criteria, whose union matches the full validity criterion.

\begin{theorem}\label{thm:kproof}
If $P$ is a valid circular proof of \muMLLo, there exists $k\in\N$ such that $P$ is a $k$-proof.
\end{theorem}

\begin{theorem}\label{thm:deckproof}
Given a circular pre-proof $P$ of \muMLLo and an integer $k$, it is decidable whether $P$ is a $k$-proof.
\end{theorem}

We now give a brief proof sketch to give an intuition on how to prove Theorems \ref{thm:kproof} and \ref{thm:deckproof}. See Appendix for details.

\begin{proof}(Sketch)
We will use the fact that once a starting point for a thread has been chosen, the thread evolves deterministically along the proof tree until a visible event occur. We define the notion of \emph{minimal shortcut} which is a part of a thread with no visible weight, bouncing on an axiom, and ending in the first point where the constraint stack is empty. It corresponds to an $\epsilon$-path.

By bounding the maximal height of the stack by $k$, we can detect loops or declare stack overflow, and we are able to compute the unique minimal shortcut (if it exists) for each starting point in the finite proof graph.
Now, checking validity of the proof can be done using an algorithm for straight threads \cite{DoumanePhD}, allowing them to take these shortcuts.

Theorem \ref{thm:kproof} is obtained by taking the maximal height reached by all minimal shortcuts of the proof graph.
\end{proof}

Combining Theorems \ref{thm:kproof} and \ref{thm:deckproof}, we obtain that validity of a circular pre-proof of \muMLLo is in $\Sigma^0_1$, i.e. recursively enumerable. Together with the reduction from Sec. \ref{sec:undecidability}, we obtain the following corollary:

\begin{corollary}\label{cor:sigma01}
The problem of deciding whether a circular pre-proof of \muMLLo is a proof is $\Sigma^0_1$-complete.
\end{corollary}

%% file: conclusion.tex
\Section{Conclusion}\label{sec:conclusion}

We have studied non-weelfounded and circular proofs of \muMALLi and defined an extended validity criterion for the  pre-proofs of \muMALLi compared to previous work by Baelde, Doumane and Saurin~\cite{BaeldeDS16}. We have shown that our criterion enjoys cut elimination and soundness, but reaches the barrier of undecidability: in the purely multiplicative fragment already, a parameter has to be bouned by an explicit value to make the criterion decidable.

For future work, we plan to investigate whether this decidability result still holds when adding the additive connectives.

We also want to extend these results to more relaxed criteria, for instance where the visible parts are only required to meet the validated branch infinitely often. 

A less sequential variant of circular proofs is also currently developed by De and Saurin \cite{infinets2019}, under the name \emph{infinets}: the canoncity and absence of commutation rules of proof nets may have good properties with respect to cut-elimination and we expect that bouncing validity may be fruitful in their setting.

Finally, the present work is a first step in improving the compositionality of circular proofs. In addition to strengthening our cut-elimination result us mentioned above, we plan to investigate how one can import results from sized types~\cite{abel2007mixed} or copattern~\cite{abel13popl} approach which also have good proprerties with respect to compositionality and may be used in an infinitary scenario~\cite{abelpientka:jfp15,abel-compo-cmcs16}.

%% file: appendix-reduction-rules.tex
\subsection{The multicut rule}
\label{app:multiCutRule}

 A new phenomenon occurs in the presence of axioms. Consider for instance the following pre-proof where $F=\phi_\alpha$ and $G=\phi_\beta$:
$$
\scalebox{.9}{
\begin{prooftree}
  \Infer{0}[\rax]{\vdash F, G^\perp}
\Hypo{\pi}
\Infer{1}[]{\vdash G, \Gamma}
\Hypo{\dots}
\Infer{3}[\rmcut]{\vdash F, \Gamma'}
\end{prooftree}}
$$
In the finitary cut-elimination procedure, we would reduce this multicut to the derivation labelled $\rmcutone$ below. Doing so, we have to perform a substitution on addresses
(denoted by $[\alpha/\beta]$) to relocate the subderivation $\pi$
on the required occurrence. Another option, described by the derivation $\rmcuttwo$ below, is to avoid the renaming by keeping a link explicitely in the multicut rule.
$$
\scalebox{.9}{
\begin{prooftree}
\Hypo{\pi[\beta/\alpha]}
\Infer{1}[]{\vdash F, \Gamma}
\Hypo{\dots}
\Infer{2}[\rmcutone]{\vdash F, \Gamma'}
\end{prooftree}
\qquad\text{or}\qquad\begin{prooftree}
\Hypo{\pi}
\Infer{1}[]{\vdash G\tikzmark{Gcut}, \Gamma}
\Hypo{\dots}
\Infer{2}[\rmcuttwo]{\vdash F\tikzmark{Fcut}, \Gamma'}
\end{prooftree}
\begin{tikzpicture}[overlay,remember picture,-,line cap=round,line width=0.1cm]
   \draw[rounded corners, smooth=2,red, opacity=.25] ($(pic cs:Fcut)+(-.1cm,.1cm)$)to ($(pic cs:Gcut)+(-.1cm,.1cm)$);    
   \end{tikzpicture}
   }
$$

 We choose the last option to avoid the
global renaming, which would complicate our technical development.

A multicut rule will now be written as: 
$$\scalebox{.9}{\begin{prooftree}
\Hypo{\vdash\Gamma_1\quad \dots\quad  \vdash\Gamma_n}
\Infer{1}[\rmcutpar]{\vdash\Gamma}
\end{prooftree}}$$
and comes with a function $\iota$ which shows how the occurrences of the conclusion are 
distributed over the premisses (modulo renaming), and a relation 
$\dbot$ specifiying which occurrences are cut-connected. A precise definition of the multicut rule is given in the appendix. 

\begin{definition}
  Given sequents $s, s_1, \dots, s_n$ where $n>0$
  and such that $s_i,s_j$ are disjoint for all $i\neq j$,
  a \defname{multicut} of conclusion $s$ and premisses
  $(s_i)_{i\in[1;n]}$ is given by an injection
  $\iota : s \mapsto \cup_{i\in[1;n]} s_i$ and a 
  symmetric relation $\dbot \subseteq (\cup_{i\in[1;n]} s_i)^2$ such that:
  \begin {itemize}
  \item For all $F \in s$, $\iota(F) \equiv F$.
  \item For all $F,G \in \cup_{i\in[1;n]} s_i$,
    $F \dbot G$ implies $F \equiv G^\perp$.
  \item $\mathrm{dom}(\dbot) =
    (\cup_{i\in[1;n]} s_i) \setminus \mathrm{im}(\iota)$.
  \item Given two sequents $s_i$ and $s_j$, we say that they are $\dbot$-connected
      on the formula occurrences $F,G$ when $F\in s_i$ and $G\in s_j$ such that $F \dbot G$.
      We say that they are $\dbot$-connected, and we write $s_i \dbot s_j$, 
      when they are $\dbot$-connected on some $F,G$.  The relation $\dbot$ on sequents
      must satisfy two conditions:
\begin{itemize}
\item two sequents must be $\dbot$-connected on at most one pair of 
  occurrences $F,G$;
\item the graph of the relation $\dbot$ must be connected and acyclic.
\end{itemize}      
    \end{itemize}
We write this multicut rule as: 
$$\begin{prooftree}
\Hypo{s_1\quad \dots \quad s_n}
\Infer{1}[\rmcutpar]{s}
\end{prooftree}$$
\end{definition}

\subsection{Cut elimination rules}
\label{app:CutElimRules}

We detail the rules of cut elimination introduced in \cref{sec:CutElimRules}.

 \begin{definition}
  \label{def:ExtReductions}
\defname{External reductions} are defined in \cref{fig:external}. 
In the first external rule, the sets $\C{C}_\Delta$ and $\C{C}_\Gamma$ are the 
subsets of $\C C$ which are respectively connected to $\Delta$ and $\Gamma$ 
respectively.
More precisely,
$$ 
\begin{array}{cc}  \C C_\Delta = &\{
   	 s \ | \ \exists s',
   	 s\dbot^* s' \text{ and } s'
	 \text{ is $\dbot$-connected to }\\
&	 \vdash \Delta,\Gamma, F\otimes G
	 \text{ on an occurrence of } \Delta
         \},
         \end{array}$$
where $\dbot^*$ is the transitive closure of the relation $\dbot$ on sequents.
$\C C_\Gamma$ is defined similarly.
\end{definition}

\begin{figure*}[htpb]
  \begin{center}
\scalebox{.95}{$\begin{array}{clc}
\begin{prooftree}
  \Hypo{\C C}
\Hypo{\vdash \Delta,F'}
\Hypo{\vdash \Gamma,G'}
\Infer{2}[\rtensor]{\vdash \Delta,\Gamma, F'\tensor G'} 
\Infer{2}[\rmcutpar]{\vdash \Sigma_\Delta, \Sigma_\Gamma, F\tensor G}
\end{prooftree}
&\underset{r}{\longrightarrow}& 
\begin{prooftree}
\Hypo{\C C_\Delta}
\Hypo{\vdash \Delta,F'}
\Infer{2}[\scriptsize\ensuremath{\mathsf{mcut(\iota', \perp\!\!\!\perp)}}]{\vdash \Sigma_\Delta, F}
\Hypo{\C C_\Gamma}
\Hypo{\vdash \Gamma,G'}
\Infer{2}[\scriptsize\ensuremath{\mathsf{mcut(\iota'', \perp\!\!\!\perp)}}]{\vdash \Sigma_\Gamma, G}
\Infer{2}[\rtensor]{\vdash \Sigma_\Delta, \Sigma_\Gamma, F\tensor G}
\end{prooftree}
\\[35pt]
\begin{prooftree}
\Hypo{\C C}
\Hypo{\vdash \Delta,F',G'}
\Infer{1}[\rparr]{\vdash \Delta, F'\parr G'} 
\Infer{2}[\rmcutpar]{\vdash \Sigma, F\parr G}
\end{prooftree}
&\underset{r}{\longrightarrow}&
\begin{prooftree}
\Hypo{\C C}
\Hypo{\vdash \Delta,F',G'}
\Infer{2}[\scriptsize\ensuremath{\mathsf{mcut(\iota', \perp\!\!\!\perp)}}]{\vdash \Sigma, F,G}
\Infer{1}[\rparr]{\vdash \Sigma, F\parr G}
\end{prooftree}
\\[35pt]
\begin{prooftree}
\Hypo{\C C}
\Hypo{\vdash \Delta, F'[\sigma X. F'/X]}
\Infer{1}[\rsigma]{\vdash \Delta, \sigma X. F'} 
\Infer{2}[\rmcutpar]{\vdash \Sigma, \sigma X. F}
\end{prooftree}
&\underset{r}{\longrightarrow}& 
\begin{prooftree}
\Hypo{\C C}
\Hypo{\vdash \Delta,F'[\sigma X. F'/X]}
\Infer{2}[\scriptsize\ensuremath{\mathsf{mcut(\iota', \perp\!\!\!\perp)}}]{\vdash \Sigma, F[\sigma X. F/X]}
\Infer{1}[\rsigma]{\vdash \Sigma, \sigma X. F}
\end{prooftree}
\end{array}$}
\\[20pt]
\scalebox{.95}{$
\begin{prooftree}
\Hypo{\C C}
\Hypo{\vdash \Delta}
\Infer{1}[\rbot]{\vdash \Delta, \bot_\beta} 
\Infer{2}[\rmcutpar]{\vdash \Sigma, \bot_\alpha}
\end{prooftree}
\underset{r}{\longrightarrow}
\begin{prooftree}
\Hypo{\C C}
\Hypo{\vdash \Delta}
\Infer{2}[\scriptsize\ensuremath{\mathsf{mcut(\iota', \perp\!\!\!\perp)}}]{\vdash \Sigma}
\Infer{1}[\rbot]{\vdash \Sigma, \bot_\alpha}
\end{prooftree}
\hspace{1cm}
\begin{prooftree}
\Infer{0}[\rone]{\vdash \llone_\beta}
\Infer{1}[\rmcutpar]{\vdash \llone_\alpha}
\end{prooftree}
\underset{r}{\longrightarrow}
\begin{prooftree}
\Infer{0}[\rone]{\vdash \llone_\alpha}
\end{prooftree}
$}
\\[15pt]
\end{center}
In the first reduction ($\rtensor/\rmcut$) we require that
$\iota(F\tensor G) = F'\tensor G'$ and take $\iota'$ and $\iota''$
that coincide with $\iota$ on $\Sigma_\Delta$ and $\Sigma_\Gamma$
respectively, and such that $\iota'(F)=F'$ and $\iota''(G)=G'$.
In the other reductions $\iota$ and $\iota'$ are similarly
constrained.
\caption{External reduction rules,
  where $r=(\mathsf{ext},F)$ and $F$ is the principal occurrence.}
\label{fig:external}
\end{figure*}

\begin{remark}
  Note that the $\rtensor/\rmcut$  external reduction yields
  multiple multicuts, though always on disjoint sub-trees.
  Thus, \muMLLim is stable by external reductions.
\end{remark}

In external reductions, we pushed a multicut away from the root,
above a logical rule. 
If we start with a \muMLLim pre-proof and apply a reduction sequence where
external rules are applied infinitely often to each multicut, we
will produce at the limit a cut-free proof. 
This is the reason why we say that external reductions are
\defname{productive}. This is not the case for the internal reduction rules
given next.

\begin{definition} \label{def:IntReductions}
  \defname{Internal reductions} are the \defname{principal} reductions
  given in \cref{fig:principal} together with the following two reductions:
  \begin{itemize}
    \item
      the merge $\rmcut/\rcut$ reduction
      $$\begin{array}{l}
        \scalebox{.9}{\begin{prooftree}
          \Hypo{\C C}
          \Hypo{\vdash \Delta,F}
          \Hypo{\vdash \Gamma,F^\bot}
          \Infer{2}[\rcut]{\vdash \Delta,\Gamma}     
          \Infer{2}[\rmcutpar]{\vdash \Sigma}
        \end{prooftree}}
        \\[10pt]
        \qquad\qquad\qquad
        \ \underset{r}{\longrightarrow}\qquad
        \scalebox{.9}{\begin{prooftree}
          \Hypo{\C C}
          \Hypo{\vdash \Delta,F}
          \Hypo{\vdash \Gamma,F^\bot}
          \Infer{3}[
             \scriptsize\ensuremath{
             \mathsf{mcut(\iota, \perp\!\!\!\perp')}}\xspace]{
             \vdash \Sigma}
        \end{prooftree}}
      \end{array}$$
      where $\dbot'$ extends $\dbot$ with
      $F\mathrel{\dbot'}F^\perp$ and
      $r=(\mathsf{merge},\{F,F^\bot\})$.
    \item
      the axiom reduction $\rmcut/\rax$
      $$\begin{array}{l}
        \scalebox{.9}{\begin{prooftree}
          \Hypo{\C C}
          \Infer{0}[\rax]{\vdash F, F'^\bot}
          \Hypo{\vdash F'', \Gamma}    
          \Infer{3}[\rmcutpar]{\vdash \Sigma}
        \end{prooftree}}
        \\[10pt]
        \qquad\quad\qquad\qquad\qquad
        \ \underset{r}{\longrightarrow}\ \qquad
        \scalebox{.9}{\begin{prooftree}
          \Hypo{\C C}
          \Hypo{\vdash F'', \Gamma}
          \Infer{2}[\scriptsize\ensuremath{\mathsf{mcut(\iota', \perp\!\!\!\perp')}}\xspace]{\vdash \Sigma}
        \end{prooftree}}
      \end{array}$$
      where $r=(\mathsf{CutAx},\{F,F'^\bot\})$,
      $F'^\bot\mathrel{\dbot}F''$
      and $\iota', \dbot'$ are defined as follows:
      \begin{itemize}
        \item for all $G \in \Sigma$,
          if $\iota(G)=F$ then $\iota'(G)=F''$, otherwise 
          $\iota'(G)=\iota(G)$;
        \item $\dbot'=\dbot\cup\{\{F'',G\}| \{F,G\}\in\dbot\}$.
\end{itemize}     
\end{itemize}
\end{definition}

\begin{figure*}[htpb] 
  \begin{center}
\scalebox{.92}{$\begin{array}{l}
\begin{prooftree}
\Hypo{\C C}
\Hypo{\vdash \Delta,F}
\Hypo{\vdash \Gamma,G}
\Infer{2}[\rtensor]{\vdash \Delta,\Gamma,F\tensor G}
\Hypo{\vdash \Theta,F'^\bot,G'^\bot}
\Infer{1}[\rparr]{\vdash \Theta,F'^\bot\parr G'^\bot} 
\Infer{3}[\rmcutpar]{\vdash \Sigma}
\end{prooftree}
\\[20pt] \hspace{4cm}
\underset{r}{\longrightarrow}\qquad
\begin{prooftree}
\Hypo{\C C}
\Hypo{\vdash \Delta,F}
\Hypo{\vdash \Gamma,G}
\Hypo{\vdash \Theta,F'^\bot,G'^\bot}
\Infer{4}[{\scriptsize\ensuremath{\mathsf{mcut(\iota, \perp\!\!\!\perp')}}}\xspace]{\vdash \Sigma}
\end{prooftree}
\\[15pt]
  \text{where }
  F\tensor G \mathrel{\dbot} G'^\bot\parr G'^\bot
  \text{ and $\dbot'$ coincides with $\dbot$ except for }
  F \mathrel{\dbot'} F'^\bot \text{ and } G \mathrel{\dbot'} G'^\bot
\\[35pt]
\begin{prooftree}
\Hypo{\C C}
\Hypo{\vdash \Delta, F'[\mu X.F'/X]}
\Infer{1}[\rmu]{\vdash \Delta,\mu X.F'}
\Hypo{\vdash \Gamma, {F}^\perp[\nu X.{F}^\perp/X]}
\Infer{1}[\rnu]{\vdash  \Gamma,\nu X.{F}^\perp} 
\Infer{3}[\rmcutpar]{\vdash \Sigma}
\end{prooftree}
\\[20pt] \hfill \underset{r}{\longrightarrow}\qquad
\begin{prooftree}
  \Hypo{\C C}
  \Hypo{\vdash \Delta, F'[\mu X.F'/X]}
\Hypo{\vdash \Gamma, {F}^\perp[\nu X.{F}^\perp/X]}
\Infer{3}[{\scriptsize\ensuremath{\mathsf{mcut(\iota, \perp\!\!\!\perp')}}}\xspace]{\vdash \Sigma}
\end{prooftree}
\\[15pt]
\text{where }
\mu X.F' \mathrel{\dbot} \nu X.{F}^\perp
\text{ and $\dbot'$ coincides with $\dbot$ except for }
F'[\mu X. F'/X] \mathrel{\dbot'} F^\bot[\nu X. F^\bot/X]
\\[35pt]
\begin{prooftree}
  \Hypo{\C C}
  \Hypo{\vdash \Gamma}
  \Infer{1}[\rbot]{\vdash \Gamma,\bot_\alpha}
  \Infer{0}[\rone]{\vdash \llone_{\beta}}
  \Infer{3}[\rmcutpar]{\Sigma}
\end{prooftree}
  \quad\underset{r}{\longrightarrow}\quad
\begin{prooftree}
  \Hypo{\C C}
  \Hypo{\vdash \Gamma}
  \Infer{2}[{\scriptsize\ensuremath{\mathsf{mcut(\iota, \perp\!\!\!\perp)}}}\xspace]{\Sigma}
\end{prooftree}
\quad \text{ where } \bot_\alpha \mathrel{\dbot} \llone_\beta
  \end{array}$}
  \end{center}
\caption{Principal reductions,
where $r=(\mathsf{princ},\{F,F'^\bot\})$ with $\{F, F'^\bot\}$ the principal occurrences that have been reduced.}
\label{fig:principal}
\end{figure*}

In internal reductions, the multicut remains at the root of the redex.
Thus, if a sequence of multicut reductions eventually involved only
internal reductions, it would not produce be productive.

The use of labels in reductions allows us to define in full details
our notion of reduction sequence and fairness.

\begin{definition} \label{def:fairnessfull}
  A \defname{reduction sequence} is a finite or infinite sequence
  $\sigma = (\pi_i,r_i)_{i\in1+\lambda}$ with $\lambda\in\omega+1$,
  where the $\pi_i$ are \muMLLim pre-proofs, the $r_i$ are labels
  identifying multicut reduction rules and, for all $i\in\lambda$,
  $\pi_i \underset{r_i}{\longrightarrow} \pi_{i+1}$.
  The sequence is \defname{fair} if for all $i\in\lambda$ and $r$
  such that $\pi_i \underset{r}{\longrightarrow} \pi'$ there is
  some $j\in\lambda$ such that $j \geq i$ and $\pi_j 
  \underset{r}{\longrightarrow} \pi_{j+1}$.
\end{definition}

%% file: appendix-cut-elim.tex
\subsection{Cut elimination for \muMLLi}
\label{app:muMLLCutElimThm}

\subsubsection{Trace of a reduction sequence}

If $\C R$ is a reduction sequence starting from $\pi$, we start by defining the \defname{trace} of $\C R$ to be the subtree of $\pi$ whose sequents occur in the reduction sequence as premisses of some multicut.
Note that each node of the trace corresponds to a well-formed inference:
indeed, if sequents $S$ and $S'$ are premisses of a same inference (which must
thus be a tensor or cut) and $S$ enters a multicut at some point in the 
reduction sequence, then $S'$ must also enter a multicut -- though not 
necessarily the same one.
However, the trace may have unjustified sequents:
this happens when a sequents $S$ enters the multicut during the reduction
sequence but never leaves it; it will then be part of the trace but
the subtree of $\pi$ rooted in $S$ will not.

The unjustified sequents of the trace are called its \defname{border sequents}.
 Note that a border sequent cannot be the conclusion of an axiom rule nor a
cut rule in the initial derivation $\pi$.
If this were the case, by fairness, it would have been 
absorbed by an $\rax/\rmcut$ or a $\rcut/\rmcut$ reduction respectively.
This allows to define the \defname{distinguished occurrence} of a border
sequent as the principal occurrence of the logical rule applied to the border 
sequent in $\pi$.

There is another reason why the trace of a reduction sequence might not
be a proof: its infinite branches may not be valid. The infinite 
branches of the trace are also infinite branches of the proof $\pi$, thus they 
are supported by valid threads of $\pi$, but these threads might not be 
included in the trace. We show that this actually never happens.

\subsubsection{The bouncing threads of the trace belong to the trace}

This section is dedicated to proving the following theorem.

\restatableproposition{prop:ThreadsInTrace2}{
Let $T$ be the trace of a reduction sequence starting from a proof $\pi$,
and let $\beta$ be an infinite branch of $T$. If $t$ is a thread of $\pi$ 
validating $\beta$, then $t$ is also a thread of $T$. 
}

We now introduce a
a useful technical tool called the \defname{residual of a pre-thread}.

\begin{definition}[Residual of a pre-thread]\label{def:residual}
  Let  $\C R$ be a finite reduction path starting from $\pi$ to $\pi'$
  and let $t$ be a pre-thread of $\pi$.
  Let $S'$ be the set of sequents of $\pi'$.
  We call the \defname{residual of $t$}
  after the reduction $\C R$ the pre-thread
  $t\cap \{(F,s,d)| s\in S', d\in \{\uparrow, \downarrow\}\}$.
\end{definition}

By definition, the length of the residual of $t$ is smaller than the length
of $t$.

\begin{proposition} \label{prop:insidemagic}
  Let  $\C R$ be a finite reduction path starting from $\pi$ to $\pi'$,
  let $T$ be its trace.
  Let $t$ be a $b$-thread of $\pi$ and $t'$ its residual after $\C R$.
  Then $t'$ is a $b$-thread.
  Furthermore, if $t$ is a $B$-path of $\pi$,
  then $t'$ is a $B$-path of $T$,
  and $t'$ has the same endpoints as $t$.
\end{proposition}

\begin{proof}
  It suffices to consider a single reduction step.
  Most of the claims follow from a simple inspection of the reduction rules.
  For the last one (i.e.\ $t$ and $t'$ have the same endpoints)
  we have to additionnally rule out the possibility that, if $s$ and $s'$
  are the endpoints of $t$, $s$ gets reduced at some point of the
  reduction while $s'$ does not: this could only happen if $s$ (or $s'$)
  was part of an $\rax/\rmcut$ reduction, but that would mean that our
  $B$-path can be extended into an $h$-path (if the axiom is at the beginning
  of the path) or the reverse of an $h$-path (if the axiom is at the end).
\end{proof}

\begin{definition}
  A pre-thread $t$ is a $B$-path of $\pi$ if:
  \begin{itemize}
    \item it is a maximal $b$-path of $\pi$,
      i.e.\ there is no $b$-path of $\pi$ which contains $t$ as an infix;
    \item it cannot be extended as an $h$-path or as the reverse of an
      $h$-path.
  \end{itemize}
\end{definition}

Intuitively, the second condition means that the path cannot be extended
by an axiom on either side, possibly after silent steps corresponding
to $\Wwait$ weights.

\restatablelemma{lem:LemmeMagique}{
  Let $\C R=\{\pi_i,r_i\}_{i\in\omega}$ be a reduction sequence and
  let $T=Tr(\C R)$ be its trace. 
  If $t$ is a $B$-path of $T$, then there is an index $i$ such that
  the endpoints of $t$ are mcut-connected in the multicut of $\pi_i$.
}

\begin{proof}
  As reductions are performed, the thread $t$ is simplified into
  \emph{residuals}. As long as these residuals remain non-empty, they are
  still $b$-paths in their respective derivations, and they keep the same
  endpoints because the only way to reduce one endpoint without the other
  is through an $\rax/\rmcut$ reduction. Moreover, the length of
  residuals only decreases. In fact, since $t$ is in the trace,
  it strictly decreases infinitely often. Thus, at some point, the two 
  endpoints of $t$ are directly mcut-connected.
\end{proof}

\begin{proof}[Proof of \cref{prop:ThreadsInTrace}]
  Let $t$ be a thread validating the branch $\beta$ in $\pi$. The visible part 
  of $t$ belongs to $\beta$, hence its belongs to the trace $T$. One needs to 
  prove that the hidden part belongs also to the trace.

  Suppose by contradiction that some
  $H=\psd{F_i}{s_i}{d_i}_{1\leq i\leq n}\in \hiddenP{t}$ leaves the trace.
  Thus there is an index $j\leq n$ such that $s_j$ is a border sequent of the 
  trace. Take the maximal such $j$. If $d_j = {\uparrow}$ then, since the
  hidden part $H$ ends in a sequent of $\beta$,
  there must be a position $k\in]j;n[$ where
  the path exits the subtree rooted in $s_j$ to re-enter $\beta$: we would
  then have $s_k = s_j$, contradicting the maximality of $j$.
  Hence, $d_j = {\downarrow}$.

  Since $H$ is a b-path, and $d_j=\downarrow$, then  there is $k>j$ such that 
  $H[j,k]$ is a $b$-path:
  keeping with the intuition that $b$-paths are well-bracketed words,
  the opening bracket at position $j$ must have a corresponding closing
  bracket at position $k$ such that the word in between is well-bracketed,
  thus a $b$-path.
  By maximality of $j$, $H[j,k]$ is in the trace.

  Let $k$ be maximal with this property.   Note that $H[j,k]$ is a maximal 
  $b$-path in the trace, since $s_j$ is a border sequent and $k$
  is chosen to be maximal. Let us show that $H[j,k]$ cannot be extended to an 
  $\epsilon$-path of the trace.  
  Suppose by contradiction that this is the case. Since $s_j$ is a border 
  sequent, $H[j,k]$ can be extended only on the right. Thus there is $l>k$ 
  such that $H[j,l]$ is an $\epsilon$-path. Since $H[j,l]$ starts with a 
  downward direction, and since $\epsilon$-paths have the same direction in 
  their endpoints, we have that  $d_l={\downarrow}$. 
  Thus by the same reasoning as before, there is $m>l$ such that $H[l,m]$ is a 
  $b$-path. Thus $H[j,m]$ is a $b$-path, which contradicts the maximality of 
  $k$. 

  We can now apply \cref{lem:LemmeMagique} to $H[j,k]$:
  at some point of the reduction, the sequents $s_j$ and $s_k$
  are mcut-connected through the occurrences $F_j$ and $F_k$. Note that 
  $s_{k+1}$ belongs to $T$ (by maximality of $j$)
  and that $F_{k+1}$ is a strict sub-occurrence of $F_k$ (otherwise, this 
  would contradict the maximality of $k$).
  \todoDB{Cela ne marche que si les b-paths peuvent finir par des $\Wwait$,
  ce qui n'est pas le cas\ldots}
  Since $F_{k+1}$ is in the trace,   
  this means that $F_j$ has been reduced which is not possible since $s_j$ is 
  in the border of the trace. 
\end{proof}

In order to view a trace as a \muMLLim proof, we shall devise a way to
justify its border sequents. Intuitively, we will identify each distinguished
occurrence with the true constant $\top$. To achieve this formally,
we introduce a \emph{truncated} proof system in the next section.
Before that, let us mention a key technical result, which builds on the
intuition that $b$-paths are simplified during cut elimination.

\subsection{Truncated proof system} \label{subsec:Truncated ProofSystem}

The truncated proof system builds on a truncation that forces
a semantics on particular occurrences.

\begin{definition}
  A \defname{truncation} $\tau$ is a partial function from
  $\Sigma^*$ to $\{\top,\llzero\}$ such that:
  \begin{itemize}
    \item For any $\alpha \in \Sigma^*$,
      if $\alpha\in\dom(\tau)$, then $\alpha^\perp\in \dom(\tau)$ and
      $\tau(\alpha)=\tau(\alpha^\perp)^\perp$.
    \item If $\alpha\in\dom(\tau)$ then for any $\beta\in\Sigma^+$,
      $\alpha.\beta\notin\dom(\tau)$.
  \end{itemize}
\end{definition}

\begin{definition}
  Given a truncation $\tau$, the infinitary \defname{proof system \muMLLit} is 
  obtained by taking all the rules of \muMLLi together with the following rule for $\top$:
  $$
  \begin{prooftree}
  \Infer{0}[\rtop]{\vdash \Gamma, \top_\alpha}
  \end{prooftree}
  $$
   with the following proviso.
  The rules of \muMLLi only apply when the address of their principal occurrence is not in
  the domain of $\tau$, otherwise the following rule has to be applied:
$$
\begin{prooftree}
\Hypo{\vdash \tau(\alpha)_{\alpha i},\Delta}
\Infer{1}[\rt]{ \vdash \phi_\alpha, \Delta}
\end{prooftree}
\qquad
\text{if }\alpha\in \dom(\tau)
$$
The notions of thread and validity are the same as in \muMLLi.
\end{definition}

As in \cite{BaeldeDS16} we define a classical truth semantics for
our truncated proof system.
Truncated occurrences (i.e.\ whose address is in 
$\dom(\tau)$) are assigned their value under $\tau$.
The semantics of a unit $\top$ or $0$ is itself.
Then this semantics is propagated to more complex formulas inductively,
interpreting $\parr$ as disjunction, $\tensor$ as conjunction, and $\mu,\nu$ 
as least and greatest fixed points respectively.The semantics of an occurrence $F$ under a truncation $\tau$
is noted $\interp{F}$. 
We establish, in the same way as in \cite{BaeldeDS16} that \muMLLit is sound wrt.\ this semantics:

\restatableproposition{prop:soundness}{
  If $\vdash \Gamma$ is provable in $\muMLLit$, then
  $\interp{F}=\top$ for some $F\in \Gamma$.
}

\todoDB{quelles différences avec straight threads ?}
Since \muMLLi is a sub-system of \muMLLit, we obtain as a corollary that
\muMLLi is sound wrt.\ the boolean semantics.

\subsection{From traces to truncated proofs}
\label{subsec:TransformingTraceIntoProof}

\begin{definition}
  Let $\pi$ be a pre-proof.
  We define the relations $\approx_\pi$ and $\dbot_\pi$ as follows:
  \begin{itemize}
    \item $F\approx_\pi G$ if there is an $h$-path in $\pi$ from $F$ to $G$,
      or from $G$ to $F$.
    \item $F\dbot_\pi G$ if there is a $b$-path in $\pi$ between $F$ and $G$.
  \end{itemize}
\end{definition}

The relations $\dbot_\pi$ and $\approx_\pi$ are symmetric ---
note that the reverse of a $b$-path from $F$ to $G$ is a $b$-path from
$G$ to $F$.
The relation $\approx_\pi$ is reflexive.

\restatableproposition{prop:TruncationWellDefined}{
  Let $\C R$ be a reduction. The trace of $\C R$ cannot contain
  an occurrence $F$ and two distinguished occurrences $G$ and $H$
  such that $F\dbot_T G$ and $F\approx_T H$.
}

\begin{proof}
  We proceed by contradiction.
  Let $t_2$ be the $b$-path from $G$ to $F$, starting with a $\downarrow$
  direction and ending with $\uparrow$.
  Let $t_1$ be the path from $F$ to $H$. It must be an $h$-path, starting
  and ending with $\uparrow$. Indeed, the reverse of an $h$-path would reach
  $H$ with a $\downarrow$ which is absurd since $H$ is a distinguished
  occurrence of a border sequent.

  Let $t$ be the $b$-path obtained by gluing the path from $G$ to $F$
  with the path from $F$ to $H$.
  This path is a $B$-path of $T$, since its endpoints are distinguished
  formulas, so they are in the border of the trace.
  By applying \cref{lem:LemmeMagique}, there is a point in the reduction
  where the occurrences $G$ and $H$ are directly mcut-connected. 

  Since $G$ and $H$ are distinguished, they are principal occurrences of the 
  rules applied to their border sequents in $\pi$. Thus, considering the
  point of the reduction where they are directly mcut-connected,
  there is an internal redex on $G$ and $H$. By fairness it is reduced,
  which contradicts the fact that they are distinguished occurrences
  of border sequents.  
\end{proof}

By \cref{prop:TruncationWellDefined} we can define the truncation
and truncated proof associated to a trace.

\begin{definition} Let $\C R$ be a reduction sequence and $T$ be its trace. The truncation $\tau$ associated with $\C R$ is defined by setting:
\begin{itemize}
\item $\tau(F)=\llzero$ if there is a distinguished occurrence $G$ such that $F\dbot_T G$,
\item $\tau(F)=\top$ if there is a distinguished occurrence $G$ such that $G\approx_T F$. 	
\end{itemize}
\end{definition}

\begin{definition}
Let $T$ be the trace of a infinite internal reduction sequence starting from $\pi$, and let $\tau$ be the truncation associated to this reduction. The \defname{truncated proof} $\pi_\tau$ is obtained from $T$ by replacing every border sequent $\vdash \phi_\alpha, \Gamma$, whose distinguished occurrence is $\phi_\alpha$ by the following derivation:
$$\begin{prooftree}
\Infer{0}[\rtop]{\vdash \top_{\alpha.i},  \Gamma}
\Infer{1}[\rt]{\vdash \phi_\alpha,  \Gamma}
\end{prooftree}$$
\end{definition}

It is now easy to establish productivity of cut elimination.
\restatableproposition{productivity}{
Any fair reduction sequence produces a \muMLLi pre-proof.
}

\begin{proof}
  By contradiction, consider a fair infinite sequence of internal multicut
  reductions starting from $\pi$. Let $\pi_\tau$ be the truncated proof of
  its trace. 
  Since no external reduction occurs,
  it means that an occurrence $F$ in the conclusion of $\pi_\tau$ can only
  be principal in an $\rax/\rmcut$ reduction of the considered sequence.
  If $\iota$ is the injection associated to the multicut after that reduction,
  the same observation holds for $\iota(F)$, and so on.
  In short, any occurrence $F' \approx_{\pi_\tau} F$
  will never be principal in a logical rule. Hence we can replace all
  these occurrences by occurrences of $\bot$.
  We thus obtain a proof of the sequent $\vdash \bot,\ldots,\bot$ which
  contradicts the soundness of \muMLLit.
\end{proof}

\subsubsection{Proof of cut elimination}

We have shown in \cref{productivity} that multicut reduction is productive.
To establish cut-elimination (\cref{th:mcutelim}), it only remains to prove that
the resulting (cut-free) pre-proof is actually a valid proof.

\begin{proof}[Proof of \cref{th:mcutelim}]
  Let $\pi$ be a \muMLLim proof of conclusion $\vdash A$,
  and $\pi'$ the cut-free pre-proof obtained by \cref{productivity},
  \ie the limit of the multicut reduction process.
  Any branch of $\pi'$ corresponds to a multicut reduction path.
  For the sake of contradiction, assume that $\pi'$ is invalid.
  It must thus have an invalid infinite branch $\beta=(s_i)_{i\in\omega}$,  corresponding
  to an infinite reduction path $\C R$.
  Let $\tau$ and $\theta$ be the associated truncation and
  truncated proof in \muMLLit.

  We set $\mathsf{Froz}(\beta)$ to be the set of occurrences of $\beta$
  which are never principal. For convenience we will use the weakening rule:
  $$\begin{prooftree}
    \Hypo{\vdash \Gamma}
    \Infer{1}[\rweak]{\vdash \Gamma,\Delta} 
  \end{prooftree}$$
  Weakening is indeed admissible as long as the derivation on which it
  is applied contains an infinite branch,
  since one can let the weakened occurrences "travel" into this infinite branch.
  Without loss of generality, we now assume that all occurrences of
  $\mathsf{Froz}(\beta)$ have been weakened away in $\beta$.

We define the truncation $\tau'$ to be the truncation obtained by extending $\tau$ as follows.
For every occurrence $F_1\parr F_2$  which is principal in $\beta$, we set, for $i\in\{1,2\}$, $\tau'(F_i)=\llzero$ if $F_i \in \mathsf{Froz}(\beta)$.

Let $\beta_{\geq i}$ be the suffix of $\beta$ starting from the $i^{th}$ element. 
If $\vdash \Gamma$ is the conclusion of $\beta_{\geq i}$, then for every  $\Delta\subseteq \Gamma$,  we define coinductively the proof ${\beta_{\geq i}}^\bot(\Delta)$ of conclusion $\vdash \Delta^\bot$ as follows. We proceed by case analysis on the rule applied to the conclusion of $\beta$.  This rule can be either a logical rule or a weakening. 
If the rule is a weakening, then it is simulated by a weakening rule.
If it is a logical rule, then let $F$ be its principal occurrence. We set  $\Delta'=\Delta\setminus\{F\}$. We have either: 
\begin{itemize}
\item $F =\sigma X. G$ where $\sigma\in\{\mu, \nu\}$.  We set $\overline{\sigma}$ to be the dual of $\sigma$, and:
$$\beta_{\geq i}^\bot(\Delta)=\begin{prooftree}
\Hypo{\beta_{\geq i+1}^\bot(\Delta', G[F/X]) }
\Infer{1}[]{\vdash \Delta'^\bot, G^\bot[F^\bot/X]}
\Infer{1}[\scriptsize ($\overline{\sigma}$)]{\vdash \Delta'^\bot,F^\bot}
\end{prooftree}$$
\item \label{def:beta-tensor} $F= G\otimes H$. Suppose wlog.\ that $G\in 
  s_{i+1}$ (\ie $H$ left the branch $\beta$). We set $\Delta''=\Delta'\cap 
  	s_{i+1}$ and:
$$\beta_{\geq i}^\bot(\Delta)=\begin{prooftree}
\Hypo{\beta_{\geq i+1}^\bot(\Delta'', G) }
\Infer{1}[]{\vdash \Delta''^\bot, G^\bot}
\Infer{1}[\rweak]{\vdash \Delta'^\bot, G^\bot, H^\bot}
\Infer{1}[\rparr]{\vdash \Delta'^\bot,F^\bot}
\end{prooftree}$$
\item $F=G\parr H$ and  $G \in \dom(\tau')$ and $H\notin \dom(\tau')$,
  or symmetrically. We set:
$$\beta_{\geq i}^\bot(\Delta)=
\begin{prooftree}

\Infer{0}[\rtp, \rtop]{\vdash  G^\bot}

\Hypo{\beta_{\geq i+1}^\bot(\Delta', H)}
\Infer{1}[]{\vdash \Delta'^\bot, H^\bot}
\Infer{2}[\rtensor]{\vdash \Delta'^\bot, F^\bot}
\end{prooftree}$$
\item $F=G\parr H$, $G \in \dom(\tau')$ and $H\in \dom(\tau')$. In this case we set:
$$\beta_{\geq i}^\bot(\Delta)=
\begin{prooftree}

\Infer{0}[\rtp, \rtop]{\vdash  G^\bot}

\Infer{0}[\rtp, \rtop]{\vdash \Delta'^\bot, H^\bot}
\Infer{2}[\rtensor]{\vdash \Delta'^\bot, F^\bot}
\end{prooftree}$$
\item $F=G\parr H$, $G \notin \dom(\tau')$ and $H\notin \dom(\tau')$. In this case we set:
$$\beta_{\geq i}^\bot(\Delta)=
\begin{prooftree}
\Hypo{\beta_{\geq i+1}^\bot(G)}
\Infer{1}[]{\vdash  G^\bot}

\Hypo{\beta_{\geq i+1}^\bot(\Delta', H)}
\Infer{1}[]{\vdash \Delta'^\bot, H^\bot}
\Infer{2}[\rtensor]{\vdash \Delta'^\bot, F^\bot}
\end{prooftree}$$
\end{itemize}

Observe now that
$\beta^\bot:=\beta_{\geq 0}^\bot$ is a proof in the truncated proof system 
$\muMLLi_{\tau'}$ of conclusion $A^\bot$. Indeed, its threads (which
are necessarily straight threads since no cuts and axioms are involved)
are the duals of the threads of the branch $\beta$, which are by hypothesis 
not valid.  

Since $\muMLLi_{\tau'}$ is sound, we have that $\interp{A^\bot}=\top$. But the 
truncated proof $\theta$ is a $\muMLLi_\tau$ proof of conclusion $A$, and again by soundness 
we have that $\interp{A}=\top$: contradiction.
\end{proof}

%% file: appendix-additives.tex
\subsection{Extending \muMLLi cut-elimination to the additives}\label{app:Additives}

In this appendix, we give details on the proof of cut-elimination theorem for full \muMALLi. First, we provide precise definition for the sliced proof system and the associated partial cut-reduction relation and the introduce persistent slices. We can then formulate precisely the additive validity criterion and prove the cut-elimination theorem.
The proof schema of additive cut-elimination follows the same pattern as in the multiplicative case but we check that the multiplicative result can indeed be lifted. Most of the definitions can be very straightforwardly lifted to the additive but for the soundness result some work has to be done: we show that we do not need a full soundness result but that a soundness result wrt. a specific class of derivations, called $\tau$-adapted proofs, which then allows us to prove productivity of cut-elimination and preservation of validity by fair-reduction sequences.

%


%
\hide
$  \pi_k = \begin{prooftree}
    \Hypo{}
    \Infer 1[\rax$\dagger$]{\vdash T\tikzmark{add6},T^\perp\tikzmark{add7}}
    \Infer 1[\rbot]{\vdash \bot,T\tikzmark{add5},T^\perp\tikzmark{add8}}
    \Hypo{}
    \Infer 1[\rax$\ddagger$]{\vdash T,T^\perp}
    \Infer 1[\rbot]{\vdash \bot,T,T^\perp}
    \Infer 1[\scriptsize{$\color{red}{\rmu^k}$}]{{\vdash \bot, T,T^\perp}}
    \Infer 2 [\rwith]{\vdash \bot\with\bot, T\tikzmark{add4},T^\perp\tikzmark{add9}}
    \Hypo{\vdash S,\orig{a}T\tikzmark{add10}}
    \Infer 2[\rcut]{\vdash \bot\with\bot, S,T\tikzmark{add3}}
    \Infer 1[\rmu,\rparr]{\vdash S,T\tikzmark{add2}}
    \Infer 1[\rnu]{\vdash S,T\tikzmark{add1}
          \begin{tikzpicture}[remember picture,overlay]
    \node[inner sep=0pt,outer sep=0pt] (b)  {\ensuremath{}};
    \draw [->,>=latex] (a.north east) .. controls +(20:1.2cm) and +(0:4.5cm) .. (b.east);
  \end{tikzpicture}\;}
  \end{prooftree} \begin{tikzpicture}[overlay,remember picture,-,line cap=round,line width=0.1cm]
   \draw[rounded corners, smooth=2,cyan, opacity=.25] ($(pic cs:add1)+(0cm,.1cm)$) to ($(pic cs:add2)+(-0.1cm,.1cm)$) to($(pic cs:add3)+(0,.1cm)$) to($(pic cs:add4)+(.1cm,-0.1cm)$) to($(pic cs:add5)+(0,-.1cm)$)to($(pic cs:add6)+(-0.1cm,.3cm)$)to($(pic cs:add7)+(-0.1cm,.3cm)$)to($(pic cs:add8)+(-.2cm,.1cm)$)to($(pic cs:add9)+(0,.1cm)$)to($(pic cs:add9)+(1cm,-.1cm)$)to($(pic cs:add10)+(-0.1cm,-.1cm)$)to($(pic cs:add10)+(0,.3cm)$); 
  \end{tikzpicture}
$

With $(\pi_k)_{k\geq 0}$ proofs, one faces three situations for cut-elimination:
\\
(i) from $\pi_0$ it is productive and produces a valid proof;
\\
(ii) from $\pi_1$  it is productive and produces an invalid pre-proof;
\\
(iii) from $\pi_k$, for $k\geq 2$ it is not even productive.

Indeed, in these examples, each $\pi_k$ contains exactly one infinite
branch which is supported by a thread on $T$ bouncing on the left-most
axiom and this thread is valid.
Still, this does not ensure productivity nor validity.
For instance, $\pi_1$ reduces to the following pre-proof:
$$  \begin{prooftree}
    \Hypo{\orig{a}\vdash S,T}
\Infer 1[\rbot]{\vdash \bot,S,T}
    \Hypo{\vdash S,\orig{aa}T}
    \Infer 1[\rbot]{\vdash \bot, S,T}
    \Infer 2[\rwith]{\vdash \bot\with\bot, S,T}
    \Infer 1[\rmu,\rparr]{\vdash S,T
              \begin{tikzpicture}[remember picture,overlay]
    \node[inner sep=0pt,outer sep=0pt] (bb)  {\ensuremath{}};
    \draw [->,>=latex] (aa.north east) .. controls +(10:2cm) and +(0:4.5cm) .. (bb.east);
  \end{tikzpicture}\;}
    \Infer 1[\rnu]{          \begin{tikzpicture}[remember picture,overlay]
    \node[inner sep=0pt,outer sep=0pt] (b)  {\ensuremath{}};
    \draw [->,>=latex] (a.north east) .. controls +(160:2cm) and +(180:2.5cm) .. (b.east);
  \end{tikzpicture}\;\vdash S,T}
  \end{prooftree}
  $$
which is not valid since any infinite branch which takes only finitely many times the left back-edge is invalid. 

To understand the problem, consider the first step of cut-reduction which is a $\rcut/\rwith$
commutation and copies the right-premiss of the cut ({\it ie.} the
non-wellfounded part of the proof): the resulting pre-proof contains two
infinite branches now, but only one thread to validate them.
While the left copy can be validated by the original thread, the
right copy does not contain a residual of the original thread. Of course,
one might consider a thread originated in the cut inference, but that
will not suffice to ensure validity, nor productivity, as $\pi_2$ examplifies: its right-most branch does not produce any rule.

\endhide

\subsubsection{Sliced proof system and its cut-reduction}

To solve the previous issue, we will make use of slices, originally introduced by Girard in his seminal paper and later used in the analysis of interaction and cut-elimination of linear logic in the setting of Ludics~\cite{girard01locus,Terui11} or in the design of additive proof-nets~\cite{HvG}.

%
%
%
%

  \begin{definition}[Additive slice]
    A \defname{sliced pre-proof} is a pre-proof built on a variant of $\muMALLi$, $\muSMALLi$, where the inference rule $\rwith$ has been replaced by the following two rules:
    $$\begin{prooftree}
      \Hypo{\vdash A, \Gamma}
      \Infer 1[\rwithl]
             {\vdash A\with B, \Gamma}
    \end{prooftree}
    \qquad \begin{prooftree}
      \Hypo{\vdash B, \Gamma}
      \Infer 1[\rwithr]
             {\vdash A\with B, \Gamma}
    \end{prooftree}$$
\end{definition}

 \begin{definition}[Slicing of a pre-proof]
    To a $\muMALLi$ sequent (pre-)proof, one can associate a set of slices
    by keeping, for each $\rwith$ inference, only one of its premisses and replacing the $\with$ by the corresponding inference in $\rwithl,\rwithr$. More precisely, a \defname{slice of $\pi$} is any $\muSMALLi$ derivation obtained from $\pi$ by applying corecursively one of the following two reductions (the other inferences are treated homomorphically):
$$    \scalebox{.9}{
$
    \begin{prooftree}
      \Hypo{\pi_1}
      \Infer 1{\vdash A_1, \Gamma}
      \Hypo{\pi_2}
      \Infer 1{\vdash A_2, \Gamma}
      \Infer 2[\rwith]{\vdash A_1\with A_2, \Gamma}
    \end{prooftree}
    ~\longrightarrow~
    \begin{prooftree}
      \Hypo{\pi_1}
      \Infer 1{\vdash A_1, \Gamma}
      \Infer 1[\rwithl]
             {\vdash A_1\with A_2, \Gamma}
        \end{prooftree}, \begin{prooftree}
      \Hypo{\pi_2}
      \Infer 1{\vdash A_2, \Gamma}
      \Infer 1[\rwithr]
             {\vdash A_1\with A_2, \Gamma}
        \end{prooftree}
$}$$
that is:
$$
\scalebox{.8}{$
Sl\left(\begin{prooftree}
      \Hypo{\pi_1}
      \Infer 1{\vdash A_1, \Gamma}
      \Hypo{\pi_2}
      \Infer 1{\vdash A_2, \Gamma}
      \Infer 2[\rwith]{\vdash A_1\with A_2, \Gamma}
    \end{prooftree}
\right) =
\left\{\begin{prooftree}
\Hypo{\pi'_i}
      \Infer 1{\vdash A_i, \Gamma}
      \Infer 1[\rwithi]
             {\vdash A_1\with A_2, \Gamma}
\end{prooftree},
\begin{array}{c}
  \pi'_i\in Sl(\pi_i),\\
    i\in\{1,2\}
\end{array}
\right\}
$}
$$        
  \end{definition}

\subsubsection{Cut-reductions for sliced proofs}
 
 Cut-reduction rules for slices of $\muSMALLi$ are identical to
 those for $\muMALLi$ except for the sliced additives. In this case,
 one may have a problematic situation when a $\rwithl$ shall interact
 with a $\roplusr$: cut-elimination cannot be performed.  
 Among the several ways to cope with this problem, we choose here to
%
 introduce a special inference, $\rdai$, a generalized axiom rule
 allowing to derive any sequent,
 which denotes the fact that a bad interaction occurred.

$$
\begin{prooftree}
  \Hypo{}
  \Infer 1[$\rdai$]{\vdash \Gamma}
\end{prooftree}
$$

This does not impact the technical development since this serves essentially the purpose of defining those slices which avoid the mismatch.

Considering the $\rdai$ inference, cut-reductions for slices are specified as follows\footnote{This reduction is straightforwardly extended to multicuts and some care shall be taken in treating $\rdai$, in particular any cut involving \rdai{} is reduced to \rdai itself as standard in ludics.}:

 \begin{definition}[Cut reductions for slices]
   The sliced additive principal case is reduced as follows, 
   \hide
 \noindent\scalebox{.75}{
$\begin{prooftree}
\Hypo{\pi_i}
      \Infer 1{\vdash A_i, \Gamma}
      \Infer 1[\rwithi]
             {\vdash A_1\with A_2, \Gamma}
\Hypo{\pi'_j}
      \Infer 1{\vdash A^\perp_j, \Gamma}
      \Infer 1[\roplusj]
             {\vdash A^\perp_1\oplus A^\perp_2, \Delta}
             \Infer 2 [\rcut]{\vdash \Gamma, \Delta}
    \end{prooftree}
    $
    $\rightarrow \left\{
\begin{array}{cc}
   \begin{prooftree}\Hypo{}\Infer 1[$\rdai$]{\vdash \Gamma,\Delta}\end{prooftree} & \text{if }i\neq j\\
\begin{prooftree}
\Hypo{\pi_i}
      \Infer 1{\vdash A_i, \Gamma}
\Hypo{\pi'_i}
      \Infer 1{\vdash A^\perp_i, \Delta}
             \Infer 2 [\rcut]{\vdash \Gamma, \Delta}
    \end{prooftree}  & \text{if }i = j\\      
    \end{array}
    \right.$
 }
 \endhide
if $\{A^\perp_1\with A^\perp_2, {A'}_1\oplus {A'}_2\}\in\dbot$, with $r = (\mathsf{princ}, \{A^\perp_1\with A^\perp_2, {A'}_1\oplus {A'}_2\})$.
$$\begin{prooftree}
\Hypo{\C C}
\Hypo{\pi_i}
\Infer 1
{\vdash A^\perp_i, \Gamma}
\Infer 1[\rwithi]
{\vdash A^\perp_1\with A^\perp_2, \Gamma}
\Hypo{\pi'_j}
\Infer 1{\vdash {A'}_j, \Gamma}
\Infer 1[\roplusj]
{\vdash {A'}_1\oplus {A'}_2, \Delta}
\Infer 3 [{\scriptsize\ensuremath{\mathsf{mcut(\iota, \perp\!\!\!\perp)}}}\xspace]{\vdash \Sigma}
    \end{prooftree}
    $$

 $$\underset{r}{\longrightarrow}
 \left\{
\begin{array}{cc}
   \begin{prooftree}\Hypo{}\Infer 1[$\rdai$]{\vdash \Sigma}\end{prooftree} & \text{if }i\neq j\\[12pt]
\begin{prooftree}
\Hypo{\C C}
\Hypo{\pi_i}
      \Infer 1{\vdash A^\perp_i, \Gamma}
\Hypo{\pi'_i}
      \Infer 1{\vdash A'_i, \Delta}
             \Infer 3 [{\scriptsize\ensuremath{\mathsf{mcut(\iota, \perp\!\!\!\perp')}}}\xspace]{\vdash \Sigma}
\end{prooftree}  & \text{if }i = j\\
\text{where }\dbot'=\dbot\cup\{\{A^\perp_i,{A'}_i\}\}&
    \end{array}
    \right.$$
 \end{definition}

Notions of $b$-paths and $\epsilon$-paths can be naturally extended to additive slices.

\subsubsection{Persistent slices}

To state the validity criterion for the additives, one needs to describe \defname{persistent} slices that will never produce a $\rdai$:
 
 \begin{definition}[Persistent slice]
   Given a slice $\pi$, a $\rwithi$ rule of principal formula $A_1\with A_2$
   occurring in $\pi$ is said to be
   \defname{well-sliced} if
   no $b$-path starting down from the $A_1\with A_2$ occurrence of
   this sequent ends in a
   $A^\perp_1\oplus A^\perp_2$ which principal formula of a $\roplusj$ inference with $i\neq j$.
   A slice is \defname{persistent} if all its $\rwithi$ occurrences are well-sliced. \end{definition}

\begin{lemma}
In a persistent slice, $\rwithi$ rules are characterized by the following property: (i) either there exists a $b$-path starting in $A_i$, (ii) or there is a maximal pre-thread $t$ starting from $A_i$ such that $\weight t$ is prefix of a word in $\C{B}$ ends in the conclusion sequent or in the conclusion of a $\top$ rule (iii) or no such maximal pre-thread $t$ (starting from $A_i$ such that $\weight{t}$ is prefix of a word in $\C{B}$) exists and they mutually extend into an infinite pre-thead.
  \end{lemma}

\begin{proof}
  By case distinction, the fourth case being disabled by the condition of well-sliced $\rwithi$ rules.
\end{proof}

In establishing the cut-elimination result, an intermediate proof system will be useful, that of partially sliced \muMALLi pre-proofs, in which both $\rwith, \rwithl$ and $\rwithr$ occur:

%
%

The reader will notice that the additive $\with$-inferences occurring in the trace of a reduction path will always be sliced inferences ($\rwithl$ or $\rwithr$).

 \begin{proposition}
   All reducts of a persistent slice are $\rdai$-free.
 \end{proposition}

 \begin{proof}
The property relies on the simple observation that (i) any cut-redcution step from a persistent slice results in a persistent slice and (ii) if there is a reduction step from slice $S$ which creates a $\rdai$ rule, then $S$ is non-persistent.
 \end{proof}

 
%
%
 
\begin{proposition}[Pull-back property]
  If $\pi \rightarrow^* \pi'$ (resp. $\pi \rightarrow^\omega \pi'$) and $S'$ is a slice of $\pi'$, then there is a
  slice $S$ of $\pi$ such that $S\rightarrow^* S'$ (resp. $S \rightarrow^\omega S'$). 
\end{proposition}

\begin{proof}
  In the case of the finitary reduction, this is a well-known property of slices.

 For the infinite fair reductions, it results from the fact that fair reductions are necessary strongly convergent in the sense of infinitary rewriting and therefore one can find a point to which in the reduction in which any point of the resulting slice is being produced and trace it back. The obtained slice is of course persistent since it reduces to $S'$.
\end{proof}



\subsubsection{Additive bouncing validity criterion}

Definitions~\ref{def:pre-thread} and~\ref{def:thread} of (pre-)threads directly adapt to the additives as they are not specific to the multiplicative fragment. 

%

\begin{definition}[Validity]
  A slicing is \defname{valid} if it is persistent and if it is valid in the multiplicative sense\footnote{That is, every infinite branch of the slicing is visited by a valid thread having its visible part contained in the branch.}.
A \muMALLi pre-proof $\pi$ is \defname{valid} if all its persistent slicings are valid.
  \end{definition}


\subsubsection{Additive cut-elimination theorem}

We now state the cut-elimination theorem and
give a shema of the proofs.

\begin{theorem} Fair infinite cut-reduction on \muMALLi proofs is productive and produces valid proofs.
\end{theorem}

\paragraph{Schema of the proof}
For cut-elimination, the proof goes by contradiction: assuming that we have a non productive fair cut-elimination, we may assume that it consists only of internal reduction steps and the trace of this cut-elimination is actually a slice with open premisses. It is actually contained in a persistent slice of $\pi$. As a consequence, there is an infinite branch of $\pi$ which is entirely visited  by the trace and this branch is visited by a thread thanks to additive validity. 
By adapting the truncated proof system to the additives and proving a restricted soundness result, we transform the valid persistent slide in a truncated derivation of the empty sequent by pruning the conclusion formulas which are never principal in the trace, from which results the contradiction.

For proving validity, it goes also by contradiction: assume the produced cut-free proof of $\vdash F$  $\pi'$ contains a persistent slice $S'$ containing an invalid branch $\beta'$.
By the pull-back property, we find a persistent slice $S$ of $\pi$ reducing to $S'$ which is valid by assumption.
From the invalid branch $\beta'$ one can build a cut-free proof $\beta'^\perp$ of $\vdash F^\perp$ together with a truncation $\tau'$ ensuring that $F$ is interpreted a $0$ while validity of S and adaptation wrt. $\tau'$ ensures that $F$ is interpreted as $\top$, a contradiction.

We now establish productivity of \muMALLi{} cut-elimination:

\begin{theorem} Fair reduction sequence of $\muMALLi$ are productive.
\end{theorem}

To do so, first notice that the following notions of appendix~\ref{app:muMLLCutElimThm} can be straightforwardly adapted to the additive case (or to additive slices):
\begin{itemize}
\item Definition~\ref{def:residual} and Lemma~\ref{lem:LemmeMagique} adapt without change to {\it persistent slices} as it is not specific to the multiplicative case (but persistency is needed).
\item 
  Proposition~\ref{prop:ThreadsInTrace} applies to the trace of a persistent slice $\pi$.
\item Missing and unjustified sequents  can be extended to reduction paths of $\muMALLi$ pre-proofs (non-sliced) for the $\rwith$ connective as done already in \cite{BaeldeDS16}.
\item While truncations need no adaptation, the truncated semantics shall be adapted to the additives by adding the following clauses:\\
  $$      \scalebox{1}{$
  \begin{array}{l}
      \interp{(\phi\with \psi)_\alpha}^{\C{E}}=
   \interp{\phi_{\alpha\Al}}^{\C{E}}\wedge\interp{\psi_{\alpha\Ar}}^{\C{E}}, ~
   \\[12pt]
   \interp{(\phi\oplus \psi)_\alpha}^{\C{E}}=
   \interp{\phi_{\alpha\Al}}^{\C{E}}\vee\interp{\psi_{\alpha\Ar}}^{\C{E}}
   \end{array}$}
  $$
\item The truncated proof system is extended in the most natural way (as in\cite{BaeldeDS16}).
\item Truncation induced by a reduction path is lifted to the additive case and it is well-defined since the additive inference would simply add a tricky  case of {\it missing} sequent for a $\with$ premiss erased when reducing a $\rcut/\rwith$ cut but there cannot be an $\epsilon$-path in this case. Therefore the only case to treat is that of distinguished occurrences of unjustified sequents of type 1 which works as for the multiplicative.
\end{itemize}

As for adapting soundness (Prop.~\ref{prop:soundness}), we actually
  do not need the full soundness result but only soundness wrt. a class of derivations that we introduce now:
  \begin{definition}
    Given a truncation $\tau$, a $\tau$-adapted $\muMALLit$ derivation $\pi$ is a $\muMALLit$ pre-proof such that
    (i) for all $(\phi\with\psi)_\alpha$ occurring in $\pi$, either $\alpha\in Dom(\tau)$ or $\{\alpha \Al,\alpha \Ar\}\cup \tau^{-1}(\top)\neq\emptyset$.
    (ii) given a $\rwith$ occurring in $\pi$ of conclusion sequent $s$ and principal formula $(\phi_1\with \phi_2)_\alpha$, if $\alpha \Al \not\in Dom(\tau)$ (resp. $\alpha \Ar \not\in Dom(\tau)$) there is
    no $b$-path starting down from the $(\phi_1\with \phi_2)_\alpha$ occurrence of $s$ ending in a
   $(\phi^\perp_1\oplus \phi^\perp_2)_\beta$ which principal formula of a $\roplusj$ inference with $i\neq j$.
  \end{definition}

  We will use soundness for $\tau$-adapted proofs:

  \begin{proposition}\label{prop:soundnesstauadapted}
    Given a truncation $\tau$ and a valid $\tau$-adapted $\muMALLit$ derivation $\pi$ of conclusion $\vdash \Gamma$, there exists a formula $F\in\Gamma$ such that $\interp{F}=\top$.
    \end{proposition}

  \begin{proof}
    The soundness proof for $\muMLLit$ of proposition~\ref{prop:soundness} can be extended to this setting:
    first notice that the relation $\leq$ on pointed sequents that is used to transfer the marking through $\epsilon$-path can be extended and we will use it in the following consistantly with the underlying slice $\tau$-adaptation suggests.

    Construction of $(s_i)$ shall now treat the additive inferences: (i) for the $\roplus$ rule, nothing is to be changed since the rule is unary: $s_{i+1}$ is the premiss of $s_i$; (ii) for the $\rwith$ inference on $s_i$ of principal occurrence $G=H\with K$, we reason as in~\cite{BaeldeDS16}: since $[f_i(G)]=\llzero$ and $f_i(G)$ is of the form $H_m\with K_m$ where $H_m$ and $K_m$ are marking of $H$ and $K$ respectively, then either $[H_m]=\llzero$ or $[K_m]=\llzero$.
    Moreover, by $\tau$-adaptation, we know that the address of one of $H,K$ is in $\tau^{-1}(\top)$ and therefore we necessarily choose the other disjunct suppose w.l.o.g. that $[H_m]=\llzero$. We set $s_{i+1}$ to be the
    premiss of $s_i$ that contains $H$.

    The definition of the sequence of markings $(f_i)$ is trivially extended (more precisely the clauses for the $\with$ and $\oplus$ are those used in~\cite{BaeldeDS16}.

The multiplicative soundness argument can be carried over in this setting since by $\tau$-adaptation, the branch we have built is part of the persistent valid slice induced by $\tau$ (that is the purpose for the $\tau$-adaptation requirement) and we therefore have a thread for which the decreasing sequence of ordinal can be applied concluding soundness for those derivations. 
  \end{proof}

  We can finally establish productivity of cut-elimination:

\begin{proof}[sketch] 
  Let $\pi$ be a $\muMALLi$ valid proof.
By contradiction, assume that $\pi$ has a  fair infinite sequence of internal reductions.
{\it Wlog.} we can assume that all reductions steps from $\pi$ are internal.

For each $\rwith$ inference of $\pi$, the trace of this cut-elimination contains at most one premiss {\it ie.} it is contained in a slice: one can actually find a persistent slice $S$ which contains the trace.

The previous remark ensures that for each $\with$ formula principal in the trace, the truncation $\tau$ of the reduction path contains one of its subformulas in its domain and that the truncated proof $\pi_\tau$ associated with the trace is $\tau$-adapted.

Since the reduction contains only internal reductions, the conclusion formulas of $\pi_\tau$ are never principal in the $\pi_\tau$ and therefore we can erase them resulting in $\pi'_\tau$ which is a $\tau$-adapted $\muMALLit$ valid derivation of the empty sequent which cannot be by soundness (prop~\ref{prop:soundnesstauadapted}). 
\end{proof}

\hide
- en effet, une tranche partielle dont tous les avec tranches sont dans des b-path tranches contient une tranche persistente.

The existence of such a persistent slice is done in the following way:
We shall build (possibly by transfinite induction) a strictly $\leq^{\mathsf sl}$-decreasing sequence $(S_i)_{i\in\lambda}$ of partial sliced such that all sliced additive are persistent. $S_0$ is defined to be the partial slice formed of the trace together with the subproofs rooted in the leaves of the trace (in particular all with rules present in the trace are sliced while all those coming from the upper part of the proof are not). Assuming that $S_i$ is built, $S_{i+1}$ is built from $S_i$ by choosing some bottomost persistently sliceable inference of $S_i$ and slicing it adequately wrt the persistency condition. To do so, we make sure that if no such a persistently sliceable $\rwith$ can be found then $S_i$ is actually a slice and the construction is over.
Indeed, assuming that there is at least one (non sliced) $\rwith$ inference
in $S_i$, let us consider the bottommost such inferences. We assert that some
of them is persistently sliceable. Indeed, given a bottommost occurrence of a $\rwith$ (of principal formula $A\with B$) which is not persistently sliceable, there is a pre-thread $t$ starting down $A\with B$, such that $\weight t$ is prefix of a word in $\C B$ and $t$ ends in a (bottommost) non-sliced $\rwith$ rule.
If no bottommost $\rwith$ rule were persistently sliceable, the previous observation would result in the existence of a cycle of $A_i\with B_i$ formulas each principal inferences of non-sliced $\rwith$ rule, such that a maximal pre-thread $t_i$ given by the observation above ends in the sequent of principal formula $A_{i+1}\with B_{i+1}$, with $A_1\with B_1 = A_n\with B_n$. Such a cycle is incompatible with the multiplicative nature of the branching rules of sliced inferences (which are all $\rtensor$ or $\rcut$).
From this contradiction we know that there is a persistently sliceable $\rwith$ inference that is used to build $S_{i+1}\subseteq S_i$.
For a limit ordinal $\lambda$, $S_\lambda = \cup_{i\in\lambda} S_i$

\begin{itemize}
\item prop that in any partial persistent slice, there exists a pers. sliceable $\rwith$ 
  \item decreasing sequence of partial persistent slices $(S_i)_{i\in\omega}$
\item show that $S= \cap_i S_i$ is a (persistent) slice and that $S \supset T$.
\item by validity, every infinite branch of S is valid, if beta in the trace, the supporting bouncing thread is also in the trace (as for multiplicative) and we conclude
\end{itemize}

Indeed, given an \rwith{} rule either it is part of the trace or not.
  In the first case, at most one of its premisses is in the trace: if exactly one of the premisses is in the trace that means that the corresponding \roplus{} (associated by the duality $\dbot$ of the multicut) reaches the premisses of the multicut and  interact through an internal step with the $\rwith$ in a way which adequate with the slice proposed by the trace. If none of its premisses are part of the trace, then we can freely pick one in such a way to be compatible with the potential $\roplus$ inference.
  that means that the corresponding $\roplus$ inference never reaches the multicut. 



  As a consequence, there is an infinite branch of $\pi$ which is entirely visited  by the trace which contradict non-productivity by trivially adapting the multiplicative reasoning in presence of sliced additive inferences.

  Assume the proof is unproductive. 
Wlog, we may assume that no commutation rule ever occurs. 
Consider the trace of the cut reduction. It is contained in a persitent slice and by productivity of the multiplicative case extended to the sliced additive, we conclude the expected contraction.
\end{proof}
\endhide

\begin{theorem}
  Given $\pi$ a $\muMALLi$ proof, any fair mcut-reduction from $\pi$  produces
  a $\muMALLi$ proof.
\end{theorem}

\begin{proof}[Sketch]
  Let $\pi$ be a \muMALLim{} proof of conclusion $\vdash F$ and $\pi'$ the cut-free preproof resulting from the previous property.

  By contradiction, assume $\pi'$ is non-valid. That means there exists a slice $S'$ of $\pi'$ ($\pi'$ is cut-free so there is no persistency assumption applying here) and an infinite branch $\beta'$ of $S'$ such that $\beta'$ is supported by no valid thread.
  By the pull-back property, $S'$ has been built by reducing a persistent slice $S$ of $\pi$ and therefore $\beta'$ corresponds to a reduction path from $S$.
  Since we are working in a persistent additive slice, the multiplicative validity can be extended in order to extract an infinite branch $\beta$ invalid in $S$.
  The construction of the proof $\beta'^\perp$ used to obtain the contradiction is lifted to the additive case as follows:
  \begin{itemize}
  \item first one shall define truncation $\tau'$ not only by considering the occurrences of $F_1\parr F_2$ which are principal in $\beta'$, but also those of $F_1\oplus F_2$ and extend the truncation to $\tau'(F_i)=0$ if   $F_1\oplus F_2$  is the principal occurrence of a $\roplusj$ rule with $i\neq j$.
    \item then the construction of the dual of branch $\beta$,  $\beta^\perp$, is extended with the following clauses to the definition of page~\pageref{def:beta-tensor}: 

      \begin{itemize}
      \item $F= G\with H$. Suppose \emph{wlog.}\xspace that $G\in s_{i+1}$ (\ie $H$ left the branch $\beta$).
$$\beta_{\geq i}^\bot(\Delta)=\begin{prooftree}
\Hypo{\beta_{\geq i+1}^\bot(\Delta', G) }
\Infer{1}[]{\vdash \Delta'^\bot, G^\bot}
\Infer{1}[\roplusl]{\vdash \Delta'^\bot,F^\bot}
\end{prooftree}$$
\item $F=G\oplus H$ and  $G \in \dom(\tau')$ and $H\notin \dom(\tau')$. We set:
$$\beta_{\geq i}^\bot(\Delta)=
\begin{prooftree}

\Infer{0}[\rtp, \rtop]{\vdash \Delta'^\bot, G^\bot}

\Hypo{\beta_{\geq i+1}^\bot(\Delta', H)}
\Infer{1}[]{\vdash \Delta'^\bot, H^\bot}
\Infer{2}[\rwith]{\vdash \Delta'^\bot, F^\bot}
\end{prooftree}$$
        \end{itemize}
         \end{itemize}

$\beta'^\perp$ is cut-free and valid (therefore valid in the sense of \cite{BaeldeDS16}) since its threads are dual of the threads of $\beta$ (which are invalid).
  
$\beta'^\perp$ is $\tau'$-adapted (as $\beta^\perp$ is cut-free, this amounts to checking that the appropriate premiss of each $\rwith$ has its formula in the domain of $\tau'$ which is by design) and 
by proposition~\ref{prop:soundnesstauadapted}, we conclude the desired contradiction since one the one hand we have $\interp{F}= \top$ and on the other hand we have  $\interp{F^\perp}= \top$.
\end{proof}

%% file: appendix-undecidability-as.tex
\subsection{(Un)decidability properties}
\label{app:undecidability}

In this appendix, we prove the undecidability of the general bouncing criterion and introduce a hierarchy of decidable sub-criteria.

\subsection*{Proof of decidability for the bounded height criterion}
\label{app:decidability}

We start by detailing some structure of proofs, via the notion of \emph{shortcut}.
\begin{definition}
A \emph{shortcut} is a finite pre-thread $t=uv$ where $\weight{u}\in W^*A$ and $v$ is a b-path.
A shortcut $t$ is \emph{minimal} if no strict prefix of $t$ is a shortcut.
\end{definition}

\noindent Notice that if $t$ is a shortcut, then $t$ is an $\epsilon$-path.

We will note $(F,s)$ a pointed sequent: $s$ is a sequent of the proof, and $F\in s$.
We want to be able to follow threads where shortcuts have been removed. These threads behave like straight threads, except on cuts where they allowed to jump from the starting point of a shortcut to its end.
We will now formalize a description of such ``jumping'' threads.

Let $\Sigma_{\textit{jump}}=\{\Wwait,\Ai,\Al,\Ar,c_l,c_r\}$, where $c_l$ (resp. $c_r$) stands for left (resp. right) cut occurrences.
A word $\tau\in\Sigma_{\textit{jump}}^*$ will be called a \defname{relative address}.
If $P$ is a preproof and $(F,s)$ is a pointed sequent of $P$, then a relative address $\tau$ points to another pointed sequent $\tau@(F,s)$ in $P$. We define this by induction on $\tau$:
\begin{itemize}
\item If $\tau=\epsilon$ then $\tau@(F,s)=(F,s)$.
\item If $\tau=\Wwait\tau'$ then $\tau@(F,s)=\tau'@(F,s')$, where $s'$ is the premiss of $s$ containing $F$.
\item If $\tau=\Ai \tau'$, $F=\sigma X. G$ (for some $\sigma\in\{\mu,\nu\}$) is principal in $s$ with premiss $s'$, then $\tau@ F= \tau'@ (G[F/X],s')$.
\item If $\tau=\Al \tau'$ (resp. $\tau= \Ar\tau')$, $F=G \star H$ (for some 
$\star\in\{\parr,\tensor\}$) is principal in $s$, then $\tau@ (F,s)=\tau'@ (F',s')$ where $F'=G$ (resp. $F'=H$) and $s'$ is the premiss of $s$ containing $F'$.
\item If $\tau=c_l\tau'$ (resp. $\tau=c_r\tau'$), and the rule applied to $s$ in $P$ is a cut, then $\tau@(F,s)=\tau'@(F',s')$, where $F'$ is the occurrence introduced by the cut on the left (resp. right) premiss $s'$ 
 of this cut.
\item Otherwise, $\tau@(F,s)$ is undefined. 
\end{itemize}

\begin{lemma}
Let $P$ be a circular pre-proof. If $t$ is a minimal shortcut from $(F,s)$ to $(F',s')$, then $F\equiv F'$. Moreover, there is a relative address $\tau$ such that $(F',s')=\tau@(F,s)$. This $\tau$ is called the \emph{effect} of $t$ and noted $\effect(t)$.
For each pointed sequent $(F,s)$, there is at most one minimal shortcut starting in $(F,s)$
\end{lemma}
\begin{proof}
Since the weight of minimal shortcut starts with $W^*A$, no choice is possible before an axiom is encountered. When going downwards, constraints will be pushed on the constraint stack, and the shortcut is again uniquely defined. When going upwards (i.e. after having seen a cut, and before the next axiom), two cases can occur. Either the constraint stack is not empty, and therefore it uniquely determines the path followed by the shortcut, or it is empty, which marks the end of the minimal shortcut $t$.
The relative address $\tau$ is given by the position of the end of $t$ relatively to the beginning of $t$ in the proof tree.
The fact that $F\equiv F'$ follows from the fact that any shortcut is an $\epsilon$-path.
\end{proof}

If $P$ is a pre-proof and $(F,s)$ is a pointed sequent in $P$, we note $\short(F,s)$ the minimal shortcut starting in $(F,s)$ if it exists. If not, we fix $\short(F,s):=\epsilon$. We also fix $\effect(\epsilon)=\epsilon$. 
We will abbreviate $\effect(\short(F,s))$ by $\effect(F,s)$ to lighten notations.

\begin{remark}
If $P$ is a circular pre-proof, and $(F,s),(F',s')$ are two pointed sequents of $P$ corresponding to the same occurrence in the finite graph of $P$, then $\effect(F,s)=\effect(F',s')$.
\end{remark}
This remark allows us to compute only finitely many effects: one for each pointed sequent in the finite proof graph.

\begin{definition}
An \emph{s-thread} is a sequence $(F_i,s_i,\uparrow)_{i\in\omega}$ that obeys the same rules as a thread going only upwards, with some relaxation in the constraints between $(F_i,s_i,\uparrow)$ and $(F_{i+1},s_{i+1},\uparrow)$ defining a pre-thread.
Indeed we add a new clause allowing the s-thread to take minimal shortcuts: 
$(F_{i+1},s_{i+1})$ can be reached from $(F_i,s_i)$ following the relative address $\effect(F_i,s_i)$, i.e. $(F_{i+1},s_{i+1})=\effect(F_i,s_i)@(F_i,s_i)$.

The weight of an s-thread is defined by generalizing the definition of weight of a thread, matching this new clause with $w_i=\Wwait$. The notion of visible part and validity of an s-thread is then induced by this definition.
\end{definition}
Notice that the visible part of an s-thread is obtained by simply removing steps introduced by this new clause, corresponding to shortcuts.

\begin{lemma}\label{lem:s-thread}
An infinite branch is validated by a thread if and only if it is validated by an s-thread.
\end{lemma}
\begin{proof}
The s-thread is obtained from the thread by compressing minimal shortcuts and replacing them with the new clause.
Conversely, the thread can be obtained from the s-thread by replacing the new clause with minimal shortcuts.
This transformation preserves the visible part.
\end{proof}

We now give the proof of Theorem \ref{thm:kproof}, stating that any valid proof of \muMLLo is a $k$-proof for some $k\in\N$.

\begin{proof}
Let $P$ be a valid proof of \muMLLo.
Each pointed sequent in $P$ can be annotated with its effect, and with the height of its minimal shortcut  (or with $(\epsilon,0)$ if this minimal shortcut does not exist). 

By Lemma \ref{lem:s-thread}, all infinite branches of $P$ are validated by s-threads, following effects annotating $P$.
Let $k$ be the maximal height annotating a pointed sequent in $P$. We obtain that $P$ is a $k$-proof.
\end{proof}

The rest of the section is devoted to proving Theorem \ref{thm:deckproof}, stating that given a pre-proof $P$ and an integer $k$, it is decidable whether $P$ is a $k$-proof.

\begin{proof}
Let us note $$\effect_k(F,s)=\left\{\begin{array}{ll} \effect(F,s)&\text{if $\short(F,s)$ has height $\leq k$}
\\ \epsilon&\text{otherwise}\end{array}\right.$$
For each $(F,s)$ in the graph, $\effect_k(F,s)$ can be computed in a finite time.
Indeed, it suffices to follow the only possible thread starting in $(F,s)$ in the graph of $P$, until we find a minimal shortcut or we detect a failure. Reasons for failure are:
\begin{itemize}
\item The weight does not begin with $\Wwait^*\Aax$, i.e. an unfolding happen before the first axiom,
\item the constraint stack gets higher than $k$,
\item we detect a loop: the same pointed sequent is visited twice with identical stack content.
\end{itemize}

This corresponds to turning the automaton $\Athread$ into a DFA, by bounding the size of the stack to $k$, and accept words of the form $\Wwait^*\Aax\Sigma^* C$ by empty stack.
This allows us to annotate each pointed sequent $(F,s)$ with $\effect_k(F,s)$. 

We can now verify that the pre-proof is a $k$-proof, using a nondeterministic parity automaton $\A_s$ reading branches of $P$ and guessing the existence of an s-thread. This automaton is identical to the one in \cite{DoumanePhD} for straight threads, except that it can follow effects. While following the relative address given by an effect, the thread is considered hidden, and the action preformed on it does not influence the accepting condition of $\A_s$. Since the length of effects is globally bounded (there are finitely many of them), the relative adresses to follow can be stored in the state space of the automaton. The pre-proof $P$ is a $k$-proof if and only if $\A_s$ accepts all branches.
\end{proof}

\subsubsection{Details on the undecidability proof}
\label{appsubsec:undecproofs}

We show that the validity condition is already undecidable for the proof system \muMLLo.


We reduce from the halting problem for two-counter machines (2CM), known to be $\Sigma^0_1$-complete \cite{Minsky61}.

Here is a brief outline of the proof.

We start by recalling the definition of 2CM in the next section. These are finite-state deterministic machines manipulating two counters, able to perform Zero test, increment and decrement on each counter.

We then show how to encode the halting problem of a 2CM $M$ using bouncing threads. The idea is to use the constraint stack to encode the value of counters, and position in the graph to encode the control state. Gadgets allow to increment or decrement each counter. 
The main difficulty lies in the tests performed by the machine: we want to design a conditional branching on the thread, depending on the value of the constraint stack. This can be done, but because of the linearity of the proof system, we cannot avoid leaving some extra constraints encoding the results of the tests, that will be collected by the thread later.
Since we want to finish with empty constraint, we need to erase this extra information. 
To do this, we add a second gadget performing the computation in a dual way: results of tests are fed to the thread, that rewinds the computation while erasing these extra constraints.
 We can finally exit the detour with (almost) no constraints, and perform a visible $\nu$-unfolding on the main branch, before looping back to the root of the proof.

The global pre-proof will be a valid proof according to the criterion if and only if the machine $M$ halts.


\medskip\subsubsection*{Two Counter Machines}

A 2CM $M$ is a tuple $(Q,q_0,q_f,\delta)$ where $Q$ is a finite set of states, $q_0$ is the initial state, $q_f$ is the final state, and $\delta$ is the transition function.
The machine has access to two counters storing nonnegative integer values. The counters are initialized to $0$. 

The possible actions of the machine are the following, where $\tau\in\{1,2\}$ identifies one of the counters:
\begin{itemize}
\item $\Inc_\tau(q)$ : increment counter $\tau$, and jump to state $q$
\item $\Dec_\tau(q)$ : decrement counter $\tau$, and jump to state $q$
\item $\T_\tau(q_Z,q_P)$: if the current value of counter $\tau$ is $0$, jump to $q_Z$, else jump to $q_P$.
\end{itemize}

Let $\Act=\{Inc_\tau(q)\mid \tau\in \{1,2\}, q\in Q\}\cup\{Dec_\tau(q)\mid \tau\in \{1,2\}, q\in Q\}\cup\{\T_\tau(q_Z,q_P)\mid \tau\in\{1,2\}, q_Z,q_P\in Q\}$ be the set of possible actions.

The transition function of $M$ is a function $\delta:Q\setminus\{q_f\}\to\Act$, specifying which action is executed when each state is reached. No action is mapped to $q_f$, since the run stops when $q_f$ is reached.

A \emph{configuration} of the machine $M$ is a triple $(p,k[1],k[2])\in Q\times\N^2$ specifying the current state and the values for the two counters.

A \emph{run} of the machine $M$ is a sequence of configurations $(p_i,k[1]_i,k[2]_i)_{0\leq i\leq n}$ such that $p_0=q_0$, $k[1]_0=k[2]_0=0$, $p_n=q_f$, and consistent with $\delta$, i.e. for all $i\in[0,n-1]:$, we have

\begin{itemize}
\item if $\delta(p_i)=\Inc_\tau(q)$ then $p_{i+1}=q$ and $k[\tau]_{i+1}=k[\tau]_i+1$.
\item if $\delta(p_i)=\Dec_\tau(q)$ then $p_{i+1}=q$ and $k[\tau]_{i+1}=k[\tau]_i-1$.
\item if $\delta(p_i)=\T_\tau(q_Z,q_P)$, then $k[\tau]_{i+1}=k[\tau]_i$, and
	\begin{itemize} 
	\item if $k[\tau]_i=0$ then $p_{i+1}=q_Z$.
	\item if $k[\tau]_i>0$ then $p_{i+1}=q_P$.
	\end{itemize}
\end{itemize}

In all cases the other counter is left unchanged, i.e. $k[3-\tau]_{i+1}=k[3-\tau]_{i}$

Without loss of generality, we can also assume that the run ends with both counter values equal to $0$, i.e. $k[1]_n=k[2]_n=0$.

The next theorem states that the halting problem is undecidable for Two Counter Machines.

\begin{theorem}\label{thm:2CM}\cite{Minsky61}
Given a Two Counter Machine $M=(Q,q_0,q_f,\delta)$, it is undecidable to determine whether $M$ has a run, by a reduction from Turing Machines halting problem.
\end{theorem}

\medskip\subsubsection*{From machines to proofs}

We will now encode the halting problem for 2CM into the problem of deciding whether a preproof is a proof.

We fix a machine $M=(Q,q_0,q_f,\delta)$.

We will build a preproof $P$ such that the leftmost infinite branch can be validated by a bouncing thread if and only if there exists a run of the machine $M$. All the other branches of $P$ will be validated by non-bouncing threads.


We will use throughout the proof the formulas $F,G$, where $F=\nu X.(X\parr X)$, $G=F^\perp=\mu X.(X\tensor X)$, and auxiliary formulas $A=\nu X. (X\parr X)\tensor X$ and $B=\mu X.(X\parr A\parr A)$. Their addresses will sometimes be omitted, keep in mind that a letter can represent different occurrences in a proof tree.

The thread will always follow a formula $F$ when going upwards, and $G$ when going downwards. The formula $A$ will be used to ensure that all infinite branches except the leftmost one are validated by non-bouncing threads.

The conclusion of the proof $P$ is the sequent $G,F,B$.

The idea of the construction is to use a bouncing thread to encode a run of $M$, by storing the current configuration on $M$ in the stack of constraints that the thread must satisfy.

The general shape of the preproof $P$ is given in Figure \ref{fig:mainproof}. By convention, formulas introduced in cuts will always be $F$ on the left and $G$ on the right. This means that a thread going upwards following a formula $F$ will always turn right on cuts, bounce on axioms in the right part, and finally come back to visit the left part of the cut.

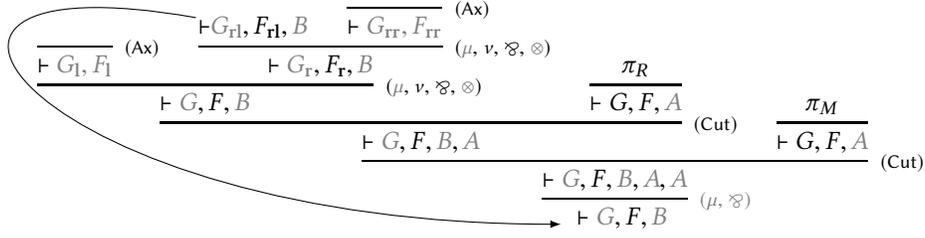
\begin{figure*}[h]
    $$~\hspace{3cm}   \begin{prooftree}
	\Infer{0}[\rax]{\vdash \gray{G_\lc,F_\lc}}
	
        \Hypo{
          \orig{bemainpreproof}\vdash \gray{G_{\rc\lc}}, F_{\rc\lc},\gB}
      \Infer{0}[\rax]{\vdash \gray{G_{\rc\rc},F_{\rc\rc}}}
      \Infer{2}[\mnvwgray]{\vdash \gray{G_\rc},F_\rc,\gB}
      
      \Infer{2}[\mnvwgray]{\vdash \gG,F,\gB}      
      
      \Hypo{\pi_R}
      \Infer{1}{\vdash G,F,\gA}      
      
      \Infer{2}[\rcut]{\vdash \gG,F,\gB,\gA}

	\Hypo{\pi_M}
	\Infer{1}{\vdash G,F,\gA}

      \Infer{2}[\rcut]{\vdash \gG,F,\gB,\gA,\gA}
       \Infer{1}[\gray{\muvee}]{     \begin{tikzpicture}[remember picture,overlay]
        \node[inner sep=0pt,outer sep=0pt] (tgmain)  {\ensuremath{}};
        \draw [->,>=latex] (bemainpreproof.north west) .. controls +(170:5cm) and +(180:7cm) .. (tgmain.west);
         \end{tikzpicture}\;
         \vdash \gG,F,\gB}
    \end{prooftree} $$
  \caption{The main preproof $P$}\label{fig:mainproof}
\end{figure*}

Auxiliary formulas are grayed to emphasize the trajectory of the thread of interest.

The thread on the branch with infinitely many $(\star)$ must use formula $F$, as it is the only one performing a $\nu$-unfolding. This means it has to go through the two cuts and bounce on axioms in $\pi_M$ and $\pi_R$.
All other infinite branches are validated by non-bouncing threads.

As described in the outline, the goal of $\pi_M$ is to use the bouncing thread stemming from $F$ to simulate a run of $M$ in a deterministic way, accumulating ``garbage constraints'' for every test. The role of $\pi_R$ is to erase these garbage constraints, by mirroring the behaviour of $\pi_M$.

After $\pi_R$, the formula $F$ is unfolded twice before looping back to the root. The first unfolding is used to match a leftover constraint, that cannot be erased in $\pi_R$ for technical reasons detailed later. The second unfolding contributes to the visible part of the thread.

\medskip\subsubsection*{Encoding of counters in the constraint stack}\label{sec:encoding}
We describe here how the preproof $\pi_M$ will be able to simulate a run of $M$ via a thread following formula $F$.

When $F$ is unfolded, a thread following $F$ can either go to the left disjunct or to the right, corresponding to weights $\Ai\Al$ or $\Ai\Ar$. We will use the alias $\lc$ for $\Ai\Al$ and $\rc$ for $\Ai\Ar$, since unfoldings will always alternate with left/right choices. 
For instance if $u$ is of the form $\lc v$, then when the thread goes up and encounters an unfolding of $F$, it has to follow the left disjunct, and update the constraint from $u$ to $v$. Constraints are updated according to the stack of the pushdown automaton $\Athread$ described in Section \ref{sec:stack}. Therefore, we will refer to the word $u\in\{\lc,\rc\}^*$ storing the current constraints as the \emph{constraint stack}.



We are now ready to detail how configurations of the machine will be encoded in the constraint stack.
A counter of value $n$ will be encoded by the sequence of constraints $(\rc\lc^n\rc)$. Therefore, when the thread is simulating the run of $M$, its constraint stack is of the form $(\rc\lc^n\rc)(\rc\lc^m\rc)$ to denote that $k[1]=n$ and $k[2]=m$.

The current state is not encoded in the constraint stack, but in the current position of the thread in the preproof: for instance if the thread is in the node labeled $(q_0)$ in the proof graph, then the current state of the corresponding run is $q_0$.

``Garbage constraints'' are constraints that will not be part of the constraint stack during the simulation of the run of $M$, but will be pushed on the stack by the thread when the simulation has succeeded, i.e. after the node labeled $(q_f)$ has been reached by the thread. In order to exit the cut, the thread must go down from $(q_f)$ to $(q_0)$, and that is where garbage constraints may be pushed. These extra constraints will encode the results of all tests performed during the computation. The result of a test is $\rc$ if the tested counter was Zero, and $\lc$ otherwise, so in general it is a letter $X\in\{\rc,\lc\}$. Garbage constraints will be of the form $\rc X$ for a test on the first counter, and of the form $\rc\lc^k\rc\rc X$ for a test on second counter, where $X$ is the result of the test and $k$ is the value of the first counter.

\medskip\subsubsection*{Auxiliary metarules for $\pi_M$}

In order to make the construction readable and modular, we start by describing metarules that will be used as building blocks throughout the section.

\medskip\subsubsection*{Rules relative to $A$}

We will use the $\infty$ notation to denote an infinite tree obtained by unfolding a proof of $A$. This yields a valid cut-free proof tree that contains no axiom.

We explicit this from sequent $\Gamma,A$, where $\Gamma$ can be any sequent.
%
%
%
%
%
%
%
%
%

%
%
$$
 \begin{array}{rl}
\begin{prooftree}
\Infer{0}[\rinf]{\vdash \Gamma,A}
\end{prooftree}
\equiv
& 
 \begin{prooftree}

       \Hypo{(\vartriangle)}
         
		\Hypo {(\spadesuit)}        
		\Hypo {(\spadesuit)}         
         \Infer{2}[\mkrule{\tensor,\parr}]{\vdash A,(A\parr A)\tensor A}
         \Infer{1}[\rnu]{(\spadesuit)\vdash A,A}
		\Infer{2}[\mkrule{\tensor,\parr}]{\vdash \Gamma,(A\parr A)\tensor A}
	\Infer{1}[\rnu]{(\vartriangle)\vdash \Gamma,A}     
      
    \end{prooftree}  
    
 \end{array}
$$
We will often want to duplicate $A$ using its $\parr$ connective. Let us define the following metarule, valid for any sequent $\Gamma,A$:
$$
 \begin{array}{rl}
\begin{prooftree}
\Hypo{\vdash \Gamma,A, A}
\Infer{1}[\rAA]{\vdash \Gamma,A}
\end{prooftree}
\equiv
& 
 \begin{prooftree}

      \Hypo{\vdash \Gamma,A, A}
     
         \Infer{1}[\vv]{\vdash \Gamma, A\parr A}
         \Infer{0}[\rinf]{\vdash A}
		\Infer{2}[\we]{\vdash \Gamma,(A\parr A)\tensor A}
	\Infer{1}[\rnu]{\vdash \Gamma,A}     
      
    \end{prooftree}  
    
 \end{array}
$$
We will define an alias for a cut rule allowing to duplicate the $A$ formula on both sides, noted $(\Acut)$. The effect of rule $(\Acut)$ is simply a cut for formulas $F,G$, so it can be considered as such for threads of interests following $F$ upwards and $G$ downwards.
$$
 \begin{array}{rl}
\begin{prooftree}
\Hypo{\vdash G,F,\gA}
\Hypo{\vdash G,F,\gA}
\Infer{2}[\rAcut]{\vdash G,F,\gA}
\end{prooftree}
\equiv
& 
 \begin{prooftree}
\Hypo{\vdash G,F,\gA}
\Hypo{\vdash G,F,\gA}
	\Infer{2}[\rcut]{\vdash G, F, \gA,\gA}     
      \Infer{1}[\rAA]{\vdash G,F,\gA}
      
    \end{prooftree}  
    
 \end{array}
$$
\medskip\subsubsection*{Copying left and right constraints}

We now describe a helpful metarule: the expanding rule $\rexp$, represented in Fig. \ref{fig:exp}. It allows to unfold $F$ and $G$ once, pairing the left (resp. right) unfolding of $F$ with the left (resp. right) unfolding of $G$. This ensures that if the thread goes left (resp. right) upwards on $F$, it will also go left (resp. right) downwards on $G$. This can therefore be understood as a copying operator: the bit of the constraint stack will be the same before and after bouncing on the current cut. 
As before, we will gray formulas that will always be avoided by the thread of interest.

\begin{figure*}[h]
  $$
  ~\hspace{4cm}
\begin{array}{cc}
\begin{prooftree}
\Hypo{\vdash G_\lc,F_\lc,\gA}
\Hypo{\vdash G_\rc,F_\rc,\gA}
\Infer{2}[\rexp]{\vdash G,F,\gA}
\end{prooftree}
\equiv
& 
 \begin{prooftree}

      \Hypo{\vdash G_\lc,F_\lc,\gA}
      
       \Hypo{\vdash G_\rc,F_\rc,\gA}

	\Infer{2}[\we]{\vdash G_\lc\tensor G_\rc,F_\lc, F_\rc, \gA,\gA}     
   
	\Infer{1}[\mnv]{\vdash G,F,\gA,\gA}
      \Infer{1}[\rAA]{\vdash G,F,\gA}

    \end{prooftree} 
 \end{array}
 $$
 \caption{the metarule $\rexp$}
  \label{fig:exp}
\end{figure*}
In most constructions, the left (resp. right) conjunct of $F$ will be paired with the left (resp. right) conjunct of $G$ when both are expanded, so we will often omit the labels. The left expansions will be represented on the left branch of the proof.
%
%
%

We give metarule $\ml$ (resp. $\mr$) forcing the reading and copying of a left (resp. right) constraint on the stack. This is enforced by preventing the thread from bouncing on an axiom if the forbidden bit is read, thanks to the axiomless infinite proofs described by the $\rinf$ metarules.

$$
\begin{array}{rl}
\begin{prooftree}
\Hypo{\vdash G,F,\gA}
\Infer{1}[\ml]{\vdash G,F,\gA}
\end{prooftree}
\equiv
& 
\begin{prooftree}
 				\Hypo{\vdash G_\lc,F_\lc,\gA}
 				\Infer{0}[\rinf]{\vdash \gray{G_\rc,F_\rc,A}}     			 
     \Infer{2}[\rexp]{\vdash G,F,\gA}
  \end{prooftree}
  
\\
~
\\
  
\begin{prooftree}
\Hypo{\vdash G,F,\gA}
\Infer{1}[\mr]{\vdash G,F,\gA}
\end{prooftree}
\equiv
& 
\begin{prooftree}
				\Infer{0}[\rinf]{\vdash \gray{G_\lc,F_\lc,A}}  
 				\Hypo{\vdash G_\rc,F_\rc,\gA} 				   			 
     \Infer{2}[\rexp]{\vdash G,F,\gA}
  \end{prooftree}
  \end{array}
$$

\medskip\subsubsection*{Constraint introduction}

We might also want to push a right constraint without popping one. For instance this is needed at the beginning when the stack is empty.
This can be done via the following metarules $\pi_{\rci}$ and $(\rci)$, where $\mathsf i$ stands for ''Introduction''. We define both variants because we might want to use one or the other depending on the context.

Let $\pi_{\rci}$ be the following preproof:

$$\begin{array}{cc}
\begin{prooftree}
\Hypo{\pi_{\rci}}
\Infer{1}{\vdash G,F,\gA}
\end{prooftree}
\equiv
&
\begin{prooftree}
	\Infer{0}[\rinf]{\vdash \gray{G_\lc,A}}
	\Infer{0}[\rax]{\vdash G_\rc,F}
	\Infer{2}[\we]{\vdash G_\lc\tensor G_\rc,F,\gA}
	\Infer{1}[\mm]{\vdash G,F,\gA}
\end{prooftree}
\end{array}
$$

This preproof allows a thread to go upwards following $F$ without any event, bounce on an axiom, and go downwards following $G$ while pushing a right constraint.

We now combine $\pi_{\rci}$ with a cut in order to go back to a thread going upwards, having a constraint stack starting with an extra $\rc$.

$$\begin{array}{cc}
\begin{prooftree}
\Hypo{\vdash G,F,\gA}
\Infer{1}[\mri]{\vdash G,F,\gA}
\end{prooftree}
\equiv
&
\begin{prooftree}
\Hypo{\vdash G,F,\gA}
	\Hypo{\pi_{\rci}}
	\Infer{1}{\vdash G,F,\gA}
\Infer{2}[\rAcut]{\vdash G,F,\gA}
\end{prooftree}
\end{array}
$$
This means that the effect of the rule $(\rci)$ on the constraint stack can be summarized by $\epsilon\mapsto \rc$, adding an extra $\rc$ at the top of the constraint stack.

The above rule can be similarly defined to introduce a $\lc$ constraint instead, by simply switching the $G,F$ axiom to the left premiss instead of the right in $\pi_{\rci}$. This dual version will be noted by $\pi_{\lci}$ and metarule $(\lci)$.

%
%
%
%

%
\medskip\subsubsection*{The initialization metarule}

\newcommand{\rinit}{\mkrule{\mathsf{init}}}
\newcommand{\rinitp}{\mkrule{\mathsf{init'}}}

We describe here the first metarule encountered in $\pi_M$. Its role is to initialize the constraint stack to encode two counters with value $0$, and go to the $(q_0)$ node to start the simulation of the run of $M$.

It must therefore allow the thread to enter with empty constraint stack and exit the cut with a constraint stack $(\rc\rc)(\rc\rc)$.
This is straightforward now that we have the $(\rci)$ rule:
$$
\begin{array}{cc}
\begin{prooftree}
\Hypo{(q_0)\vdash G,F,\gA}
\Infer{1}[\rinit]{\vdash G,F,\gA} 
\end{prooftree}
\equiv 
&
   \begin{prooftree}
   \Hypo{(q_0)\vdash G,F,\gA}

 	\Infer{1}[\mri]{\vdash G,F,\gA}
 	\Infer{1}[\mri]{\vdash G,F,\gA}
    \Infer{1}[\mri]{\vdash G,F,\gA}
    \Infer{1}[\mri]{\vdash G,F,\gA}
    \end{prooftree}
\end{array}
$$

In the following, we describe how to build the preproof $\pi_M$, by connecting nodes of the form $(p)$ to their successors in the computation.
So for each $p\in Q$, a preproof of ``local root'' $(p)$ will be built with hypotheses of the form $(q)$ with $q\in Q$. Once such a preproof has been built for each node $(p)$, the hypotheses $(q)$ are connected to their corresponding ``local root'' node trough back loops.

%
\medskip\subsubsection*{Encoding the action $\Inc$}

\newcommand{\rinc}[1]{\mkrule{\mathsf{Inc}_{#1}}}
\newcommand{\counter}{{\mathsf{counter}}}
\newcommand{\rcount}{\mkrule{\counter}}

We now assume that we are at a node of the proof graph labeled by $(p)$, where $p$  is a state of the machine $M$, and that the current constraint stack encodes the counter value as described earlier, i.e. $(\rc\lc^{k[1]}\rc)(\rc\lc^{k[2]}\rc)$.

Assume $\delta(p)=\Inc_1(q)$ with $q\in Q$.
We will build a metarule $(\Inc_1)$ updating the configuration by acting on the constraint stack, and ending up in node $(q)$:

$$
\begin{array}{cc}
\begin{prooftree}
\Hypo{(q)\vdash G,F,\gA}
\Infer{1}[\rinc1]{(p)\vdash G,F,\gA} 
\end{prooftree}
\equiv 
&
   \begin{prooftree}
   \Hypo{(q)\vdash G,F,\gA}
    \Hypo{\pi_{\lci}}
    \Infer{1}{\vdash G_\rc,F_\rc,\gA}
    \Infer{1}[\mr]{\vdash G,F,\gA}
    \Infer{2}[\rAcut]{(p)\vdash G,F,\gA}
    \end{prooftree}
\end{array}
$$

Notice that in order to bounce on the axiom, the thread must see $\rc$ while going up, and $\rc\lc$ on the way down. The rest of the current stack (of the form $\lc^*\rc(\rc\lc^*\rc)$) is left unchanged.
This rule turns a stack of the form $\rc u$ into $\rc\lc u$, thereby incrementing the first counter. We will abbreviate this action $\rc\mapsto \rc\lc$.

We might also need to increment the second counter, which is deeper in the stack.
For this, let us devise another auxiliary metarule $(\counter)$, allowing us to skip the part of the stack encoding the first counter.

$$
\begin{array}{cc}
\begin{prooftree}
\Hypo{\vdash G,F,\gA}
\Infer{1}[\rcount]{\vdash G,F,\gA} 
\end{prooftree}
\equiv 
&
   \begin{prooftree}
   \Hypo{(\dagger)\vdash G_{\rc\lc},F_{\rc\lc},\gA}
   \Hypo{\vdash G_{\rc\rc},F_{\rc\rc},\gA}
    \Infer{2}[\rexp]{(\dagger)\vdash G_\rc,F_\rc,\gA}
    \Infer{1}[\mr]{\vdash G,F,\gA}
    \end{prooftree}
\end{array}
$$

This metarule processes constraints of the form $(\rc\lc^*\rc)$ on the way up, and these same constraints will be copied back when the thread returns downwards from the formula $G$ of the hypothesis.
This means that this gadget allows the right premiss sequent to access the encoding of the second counter, while leaving the first one untouched.

We can now give the pre-proof allowing to increment the second counter, by simply adding the $(\counter)$ metarule at the appropriate place:

$$
\begin{array}{cc}
\begin{prooftree}
\Hypo{(q)\vdash G,F,\gA}
\Infer{1}[\rinc2]{(p)\vdash G,F,\gA} 
\end{prooftree}
\equiv 
&
   \begin{prooftree}
   \Hypo{(q)\vdash G,F,\gA}
    \Hypo{\pi_{\lci}}
    \Infer{1}{\vdash G_\rc,F_\rc,\gA}
    \Infer{1}[\mr]{\vdash G,F,\gA}
    \Infer{1}[\rcount]{\vdash G,F,\gA}
    \Infer{2}[\rAcut]{(p)\vdash G,F,\gA}
    \end{prooftree}
\end{array}
$$

The effect of this metarule is $(\rc\lc^*\rc)\rc\mapsto (\rc\lc^*\rc)\rc\lc$

This achieves the treatment of states performing an increment. For all nodes $(p)$ where $p$ performs an increment of counter $\tau$ before going to $q$, we link node $(p)$ of the proof with node $(q)$ through the metarule $(\Inc_\tau)$.

\medskip\subsubsection*{Encoding the action $\Dec$}\label{sec:Dec}

\newcommand{\rdec}[1]{\mkrule{\mathsf{Dec}_{#1}}}

Assume $\delta(p)=\Dec_1(q)$ with $q\in Q$.

This means we want to build a metarule with action $\rc\lc\mapsto \rc$ on the stack, in order to decrease the value of the first counter by $1$.

This is done by the following metarule $\rdec1$:


$$
\begin{prooftree}
\Hypo{(q)\vdash G,F,\gA}   
    \Infer{0}[\rinf]{\gray{G_\lc,F_\lc, F_{\rc\rc},A}}
    \Infer{0}[\rax]{G_\rc, F_{\rc\lc}}
    \Infer{2}[\we]{\vdash G_\lc\tensor G_\rc, \gray{F_\lc}, F_{\rc\lc},\gray{F_{\rc\rc}}, \gA}
    \Infer{1}[\nuvee]{\vdash G_\lc\tensor G_\rc, \gray{F_\lc}, F_\rc, \gA}
	\Infer{1}[\mnv]{\vdash G,F,\gA}
\Infer{2}[\rAcut]{(p)\vdash G,F,\gA}
 \end{prooftree}
$$

Notice that if the thread does not start with $\rc\lc$, it gets lost in an $(\infty)$ proof, corresponding to a failure of the run.

If $\delta(p)=\Dec_2(q)$, the construction is similar, using again the metarule $(\counter)$ to leave the first counter untouched and access the second one. This is done by the following metarule $\rdec2$.


$$
\begin{prooftree}
\Hypo{(q)\vdash G,F,\gA}   
    \Infer{0}[\rinf]{\gray{G_\lc,F_\lc, F_{\rc\rc},A}}
    \Infer{0}[\rax]{G_\rc, F_{\rc\lc}}
    \Infer{2}[\we]{\vdash G_\lc\tensor G_\rc, \gray{F_\lc}, F_{\rc\lc},\gray{F_{\rc\rc}}, \gA}
    \Infer{1}[\nuvee]{\vdash G_\lc\tensor G_\rc, \gray{F_\lc}, F_\rc, \gA}
	\Infer{1}[\mnv]{\vdash G,F,\gA}
	\Infer{1}[\rcount]{\vdash G,F,\gA}
\Infer{2}[\rAcut]{(p)\vdash G,F,\gA}
 \end{prooftree}
$$

\medskip\subsubsection*{Encoding the action $\T$}\label{sec:Test}

\newcommand{\rT}[1]{\mkrule{\T_{#1}}}

It remains to describe how to modify the constraint stack for actions of type $\T$.

\medskip\subsubsection*{Test on the first counter}\label{sec:test2}

Let us assume first that $\delta(p)=\T_1(q_Z,q_P)$.

Notice that a zero test can be performed by simply testing whether the stack starts with $\rc\rc$ or with $\rc\lc$, i.e. by identifying the second letter of the stack.
Therefore, these first two letters can be used to branch to the result of the test. 
However, since they must be still be part of the encoding, we need to reintroduce them in the stack after having read them.
We define the metarule $(\T_1)$ accordingly:

$$
   \begin{prooftree}
   \Hypo{(q_P)\vdash G,F,\gA}
    \Infer{1}[\mli]{\vdash G,F,\gA}
    \Infer{1}[\mri]{\vdash G_{\rc\lc},F_{\rc\lc},\gA}
   \Hypo{(q_Z)\vdash G,F,\gA}
     \Infer{1}[\mri]{\vdash G,F,\gA}
    \Infer{1}[\mri]{\vdash G_{\rc\rc},F_{\rc\rc},\gA}
    \Infer{2}[\rexp]{\vdash G_\rc,F_\rc,\gA}
    \Infer{1}[\mr]{(p)\vdash G,F,\gA} 
    \end{prooftree}
$$
%

This metarule allows the thread to go to $(q_Z)$ if the counter was zero, or to $(q_P)$ if the counter was strictly positive, leaving the stack unchanged in both cases.

Notice that this metarule also leaves some garbage constraints in the following sense: when going back down from $(q_Z)$ (resp. $(q_P)$) to $(p)$, the thread will push extra constraints $\rc\rc$ (resp. $\rc\lc$) on top of the pile, due to the $(\exp)$ and $(\rc)$ rule in $(\T_1)$.

\medskip\subsubsection*{Test on the second counter}
\newcommand{\pishift}{\pi_{\mathrm{shift}}}
\newcommand{\move}{\mathsf{move}}
\newcommand{\pimove}{\pi_{\move}}
\newcommand{\rmove}{\mkrule{\mathsf{result}}}
\newcommand{\cop}{\mathsf{copy}}
\newcommand{\rcopl}{\mkrule{\cop_\lc}}
\newcommand{\perm}{\mathsf{prep}}
\newcommand{\rperm}{\mkrule{\perm}}

We now assume that $\delta(p)=\T_2(q_Z,q_P)$, and we want to encode the corresponding metarule, linking $(p)$ to $(q_Z)$ and $(q_P)$ in the pre-proof $\pi_M$.

This is more tricky, because we need to access the relevant bit encoding the result of this test, and copy the value of the first counter after it to restore the stack. This corresponds to copying an unbounded amount of information, so this cannot be done directly in the same way as in the previous construction for $\T_1$.

We therefore design auxiliary gadgets allowing us to copy the information bit by bit.
The result $T$ of a test will be encoded by $T=\rc\rc\rc$ for zero and $T=\rc\rc\lc$ for not zero.

The pre-proof $\pishift$, represented Fig. \ref{fig:pishift} has effect $\lc^{k+1} T\mapsto \lc^k T\lc$, with $T\in\{\rc\rc\rc,\rc\rc\lc\}$. 
The proof can be built thanks to the following table, that explicits the transformation of the relevant prefix constraint stack:


%
%
%
%
%
%

$$\begin{array}{|l|c|c|c|}
\hline
\text{before: }&\lc\lc & \lc\rc\rc\lc & \lc\rc\rc\rc \\
\hline
\text{after: }&\lc\lc  & \rc\rc\lc\lc & \rc\rc\rc\lc\\
\hline
\end{array}
$$

\begin{figure*}[h]
$$
\begin{array}{cc}
\scalebox{1}{
\begin{prooftree}
\Hypo{\pishift}
\Infer{1}{\vdash G,F,\gA} 
\end{prooftree}
}
\equiv 
&
\scalebox{1}{
\begin{prooftree}
\Hypo{(\bullet)\vdash \Gl,\Fl,\gA}
\Infer{0}[\rax]{\vdash \Gr,\Fr}
\Infer{2}[\mnvw]{\vdash G,F,\gA}

		\Infer{0}[\rax]{\vdash \Gll, \Fll}
		\Infer{0}[\rinf]{\vdash \gray{\Glr,\Fr, A}}
		\Infer{2}[\mw]{\vdash \Gl, \Fll,\gray{\Fr},\gA}  
    \Hypo{\piaux}	
    \Infer{1}{\vdash \Gr,\Flr,\gA}
    
\Infer{2}[\we]{\vdash \Gl\tensor \Gr,\Fll,\Flr,\gray{\Fr},\gA,\gA}  

    \Infer{1}[\mnvnv]{\vdash G,F,\gA,\gA}
    \Infer{1}[\rAA]{\vdash G,F,\gA}

\Infer{2}[\rAcut]{(\bullet)\vdash G,F,\gA} 
\end{prooftree}
}
\end{array}
$$

$$
\begin{array}{cc}
\scalebox{1}{
\begin{prooftree}
\Hypo{\piaux}
\Infer{1}{\vdash \Gr,\Flr,\gA}
\end{prooftree}
}
\equiv
&
\scalebox{1}{
\begin{prooftree}	
\Infer{0}[\rinf]{\vdash \gray{\Grl,F_{\lc\rc\lc},A}} 	
\Infer{0}[\rax]{\vdash G_{\rc\rc\lc\lc}, F_{\lc\rc\rc\lc}}	 
		\Infer{0}[\rinf]{\vdash \gray{G_{\rc\rc\lc\rc},A}}		 
		\Infer{2}[\mw]{\vdash G_{\rc\rc\lc}, F_{\lc\rc\rc\lc},\gA}
\Infer{0}[\rax]{\vdash G_{\rc\rc\rc\lc}, F_{\lc\rc\rc\rc}}
		\Infer{0}[\rinf]{\vdash \gray{G_{\rc\rc\lc\rc},A}} 
		\Infer{2}[\mw]{\vdash G_{\rc\rc\rc}, F_{\lc\rc\rc\rc},\gA}				
		\Infer{2}[\mnvw]{\vdash \Grr,F_{\lc\rc\rc},\gA}	
		\Infer{2}[\rexp]{\vdash \Gr,\Flr,\gA}
\end{prooftree}
}
\end{array}
$$
\caption{The pre-proof $\pishift$ with its auxiliary $\piaux$}
\label{fig:pishift}
\end{figure*}

Notice that all axioms in the right part of the cut in $\pishift$ are paired according to the table above.

The main interesting phenomenon in $\pishift$ occurs on the $(\bullet)$ loop.
In the case where the result starts with $\lc$,  this loop allows to enter the cut again, after the popping of one $\lc$ constraint. This will therefore perform the wanted transformation on the constraint stack. Let us take an explicit example to see this gadget at work: consider a thread entering $\pishift$ with constraint $u=\lc\lc\rc\rc\rc$.
The thread will bounce on axiom $\Gll,\Fll$, leaving this constraint unchanged, and it will enter the $(\bullet)$ node with constraint $\lc\rc\rc\rc$. This time, the detour will enter $\piaux$, and will pop $\lc\rc\rc\rc$ and push $\rc\rc\rc\lc$ onto the stack. When exiting the cut, the thread will have a constraint starting with $\rc$ and therefore will immediately bounce on the axiom. When going back, it will push back the first $\lc$ on the way down to the original root. It will finally exit with constraint $\lc\rc\rc\rc\lc$, which is the wanted result of the mapping $\lc^k T\mapsto \lc^{k-1} T\lc$.

We can now iterate $\pishift$ in order to move the result $T$ on top of the stack.

Let us start with an auxiliary metarule $(\cop_\lc)$ (Fig. \ref{fig:copl}) copying the first letter of the stack if it is $\lc$, and do nothing if the stack starts with $\rc$. I.e. it has action $\left\{\begin{array}{c} \lc\mapsto\lc\lc \\ \rc\mapsto\rc\end{array}\right.$ on the stack. It will actually be the case that if the constraint starts with $\rc$, it starts with $\rc\rc$.

\begin{figure*}[h]
$$
~\hspace{3cm}\begin{array}{cc}
\begin{prooftree}
\Hypo{\vdash G,F,\gA}
\Infer{1}[\rcopl]{\vdash G,F,\gA} 
\end{prooftree}
\equiv 
&

\begin{prooftree}
\Hypo{\vdash G,F,\gA}
\Infer{0}[\rax]{\vdash \Gll, \Fl,\gA}
\Infer{0}[\rinf]{\vdash \Glr,\gA}	
  
		\Infer{2}[\mw]{\vdash \Gl, \Fl,\gA}
		\Infer{0}[\rax]{\vdash \Gr, \Fr}
\Infer{2}[\we]{\vdash \Gl\tensor \Gr,\Fl,\Fr,\gA}  

    \Infer{1}[\mnv]{\vdash G,F,\gA}
	
\Infer{2}[\rAcut]{\vdash G,F,\gA} 
\end{prooftree}
\end{array}
$$
\caption{The metarule $\rcopl$}
\label{fig:copl}
\end{figure*}
We now build the proof $\pimove$ (Fig. \ref{fig:pimove}), iterating $\pishift$, allowing to copy an unbounded quantity of information past the test result $T$. The proof $\pimove$ has the following action on the stack: $\left\{\begin{array}{c}\lc^{k} T\mapsto \lc^k T \lc^{k}\\ \rc\mapsto\rc\end{array}\right.$. The principle is to first duplicate the leading $\lc$, so that the extra occurrence can be used to test whether we want to perform a shifting using $\pishift$.
If the stack starts with $\rc$ (actually with $\rc\rc$), then the $\rcopl$ metarule will do nothing, and the thread will just bounce on the right axiom. 
This allows us to iteratively call $\pishift$, until we reach the encoded result $T$.
When exiting this gadget, the last $\pishift$ leaves constraint $T\lc^k$, and one extra $\lc$ is collected by each $(\bigstar)$ loop. That is why the constraint afer $\pimove$ is $\lc^kT\lc^k$.

\begin{figure*}[h]
$$
  ~\hspace{3cm}
  \begin{array}{cc}
\begin{prooftree}
\Hypo{\pimove}
\Infer{1}{\vdash G,F,\gA} 
\end{prooftree}
\equiv 
&

\begin{prooftree}
		\Hypo{(\bigstar)\vdash G,F,\gA}
		\Hypo{\pishift}
    \Infer{2}[\rAcut]{\vdash \Gl,\Fl,\gA}
    \Infer{0}[\rax]{\vdash \Gr,\Fr}
\Infer{2}[\mnvw]{\vdash G,F,\gA} 
\Infer{1}[\rcopl]{(\bigstar)\vdash G,F,\gA} 
\end{prooftree}
\end{array}
$$
\caption{The metarule $\pimove$}
\label{fig:pimove}
\end{figure*}
Another auxiliary metarule $(\perm)$ (Fig. \ref{fig:perm}) will allow us to prepare the input for $\pimove$ from the standard counter encoding, i.e. performing action $\rc\lc^k\rc\rc X \mapsto \rc\lc^k(\rc\rc X)\rc\rc X$, with $X\in\{\rc,\lc\}$. The parenthesized expression is the $T$ that we will want to move to the top. Notice that this corresponds to a copying of the counter $\rc\lc^k\rc$, followed by a mapping $\rc\rc X\mapsto \rc X\rc\rc X$. The dots in Fig. \ref{fig:perm} represents $(\infty)$ proofs, not detailed for concision.

\begin{figure*}[h]
$$
  ~\hspace{3cm}
  \begin{array}{cc}
\scalebox{1}{
\begin{prooftree}
\Hypo{\vdash G,F,\gA}
\Infer{1}[\rperm]{\vdash G,F,\gA} 
\end{prooftree}
}
\equiv 
&
\scalebox{1}{
\begin{prooftree}
\Hypo{\vdash G,F,\gA}

	\Infer{0}[\rinf]{\gray{\Gl,\Fl,\Frl,A}}

	\Infer{0}[\rax]{\vdash G_{\rc\lc\rc\rc\lc}, F_{\rc\rc\lc}}	
	\Infer{0}[\rax]{\vdash G_{\rc\rc\rc\rc\rc}, F_{\rc\rc\rc}}
	\Hypo{\dots}
	\Infer{3}[\mkrule{\tensor,\mu, \parr_A}$^*$]{ \vdash \Gr,  F_{\rc\rc\lc}, F_{\rc\rc\rc},\gA}

\Infer{2}[\mw]{\vdash G, \gray{\Fl}, \gray{\Frl}, F_{\rc\rc\lc}, F_{\rc\rc\rc}, \gA,\gA}
\Infer{1}[\mkrule{(\nu,\parr)^3,\parr_A}]{\vdash G,F,\gA}
\Infer{1}[\rcount]{\vdash G,F,\gA}

\Infer{2}[\rAcut] {\vdash G,F,\gA}
\end{prooftree}}
\end{array}
$$
\caption{The metarule $\perm$}
\label{fig:perm}
\end{figure*}

Let us now combine these gadgets to define a metarule $\rmove$, moving the test result at the wanted place, while producing extra stack content before it. This metarule has action $\rc\lc^kT\mapsto \lc^k T\lc^kT$, with $T\in\{\rc\rc\lc,\rc\rc\rc\}$, and leaves a garbage constraint $\rc$ that will be seen later on the way down.

$$
\begin{array}{cc}
\begin{prooftree}
\Hypo{\vdash G,F,\gA}
\Infer{1}[\rmove]{\vdash G,F,\gA} 
\end{prooftree}
\equiv 
&

\begin{prooftree}
	\Hypo{\vdash G,F,\gA}
    \Hypo{\pimove}
    \Infer{1}{\vdash G,F,\gA} 
\Infer{2}[\rAcut]{\vdash G_\rc,F_\rc,\gA} 
\Infer{1}[\mr]{\vdash G,F,\gA} 
\Infer{1}[\rperm]{\vdash G,F,\gA} 
\end{prooftree}
\end{array}
$$

We can detail the stack modifications in $\rmove$: $(\rc\lc^k)T\trans{\rperm}{} \rc\lc^k TT \trans{\mr}{} \lc^k TT \trans{\pimove}{} \lc^k T\lc^k T$.

We can finally build the metarule for $\T_2$, performing the wanted test and leaving garbage constraint of the form $\rc\lc^k\rc\rc X$ with $X\in \{\lc,\rc\}$.
$$
\begin{prooftree}
   
	\Hypo{(\clubsuit)\vdash \Gl,\Fl,\gA}
   
   \Hypo{(q_P)\vdash G,F,\gA}
   \Infer{1}[\mri]{\vdash G_{\rc\rc\lc},F_{\rc\rc\lc},\gA}
   \Hypo{(q_Z)\vdash G_,F,\gA}
   \Infer{1}[\mri]{\vdash G_{\rc\rc\rc},F_{\rc\rc\rc},\gA}
    \Infer{2}[\rexp]{\vdash \Grr,\Frr,\gA}
    \Infer{1}[\mr]{\vdash \Gr,\Fr,\gA}
    \Infer{2}[\rexp]{(\clubsuit)\vdash G,F,\gA}
    \Infer{1}[\rmove]{(p)\vdash G,F,\gA} 
    \end{prooftree}
$$

%
%
%

The principle of this gadget is the following: after preparing the stack via the $\rmove$ metarule, the $(\clubsuit)$ loops and the $\mr$ rule first pop the garbage prefix $\lc^k\rc\rc$. The following bit $X\in\{\lc,\rc\}$ is the wanted test result, and allows us to enter $(q_Z)$ or $(q_P)$ with a remaining stack that encodes the next configuration of the machine (after adding the leading $\rc$ to complete the valid encoding).

\medskip\subsubsection*{Final state $q_f$}
It remains to describe what happens to a thread entering the node labelled by the final state $q_f$. We will simply allow it to finally bounce on an axiom, thereby starting a downwards path that will gather all the garbage constraints, exit $\pi_M$, and enter the second cut and the proof $\pi_R$.

We just need to evacuate the formula $A$ from the sequent. Since we know that the constraint stack starts with $\rc$, this can be done in the following way:

$$
\begin{prooftree}
\Infer{0}[\rinf]{\vdash \gray{ \Gl,\Fl,A}}
\Infer{0}[\rax]{\vdash \Gr,\Fr}
\Infer{2}[\mnvw]{(q_f)\vdash G,F,\gA}
\end{prooftree}
$$

%
%
%
%
%
%
%
%

\medskip\subsubsection*{Exiting the proof $\pi_M$}

This concludes the description of the proof $\pi_M$, starting with $\begin{prooftree}\Hypo{(q_0)\vdash G,F,A}\Infer{1}[\rinit]{\vdash G,F,A}\end{prooftree}$, and built as described by linking state-labelled nodes using rules we defined for the action performed by each state according to $\delta$.

This construction ensures the following Lemma:

\begin{lemma}\label{lem:exitconstraints}
A thread entering $\pi_M$ with empty constraint will be able to exit it if and only if $M$ has a run.
If this is the case, the constraint after the exit is a word $u_1u_2\dots u_k \rc^4$, where $k$ is the number of tests performed, and $u_i$ encodes the results $X\in\{\lc,\rc\}$ of the $i^\mathit{th}$ test in the following way:
\begin{itemize}
\item if the test is on the first counter, then $u_i=\rc X$,
\item if the test is on the second counter, then $u_i=\rc\lc^k\rc\rc X$, where $k$ is the value of the first counter at the time of the test.
\end{itemize}
\end{lemma}

\begin{proof}
It is straightforward to prove by induction of the length of the thread/run: a thread entering $\pi_M$ with empty constraint can reach node $(p)$ with constraint $(\rc\lc^n\rc)(\rc\lc^m\rc)$ if and only if their is a partial run of $M$ reaching state $p$ with counter values $n,m$. Moreover, the garbage constraints that will be pushed back from $(p)$ on the way back to the root of $\pi_M$ encode the results of tests as described in the statement of the Lemma.
The $\rc^4$ following garbage constraints is the constraint stack reached in $(q_f)$ at the end of computation, encoding two counters of value $0$, as we assumed $M$ ends with this configuration.
\end{proof}

However, we want the thread to be back on the main branch with a bounded number of constraints. We therefore must erase all these garbage constraints. This will be the role of the pre-proof $\pi_R$.

\medskip\subsubsection*{The reverse simulation proof $\pi_R$}
\newcommand{\dual}{\mathit{dual}}

The goal of the proof $\pi_R$ will be to erase the garbage constraints instead of creating them.
To achieve this, we will aim at building a dual version of the pre-proof $\pi_M$.
A convenient way to think about it is the following: we try to reproduce the proof $\pi_M$, but considering this time that the thread of interest $t'$ originates in the sequent $G$. Therefore it will always follow $G$ upwards and $F$ downwards. We call such a thread a $G$-thread, and the previous version used in $\pi_M$ a $F$-thread.

The principle is that if the $G$-thread $t'$ creates garbage constraints $u$, i.e. has a visible weight $u\in\Sigma^*$ then its dual, the identical thread considered in reverse and originating in $F$, has a visible weight $\dual(u)$, where the function $\dual:\Sigma^*\to\Sigma^*$ is defined as follows.
Let $v\mapsto\bar{v}:\Sigma^*\to\Sigma^*$ be the length-preserving morphism defined on letters by the following correspondence:

$$\begin{array}{|l||c|c|c|c|c|c|c|c|c|}
\hline
&&&&&&&&&\\[-1em]
x:& ~l~ & ~r~ & ~i~ & ~\bar{l}~ &~\bar{r}~ &~\bar{i}~ & ~A~ & ~C~ & ~W~ \\
\hline
&&&&&&&&&\\[-1em]
\bar{x}:& ~\bar{l}~ &~\bar{r}~ &~\bar{i}~& ~l~ & ~r~ & ~i~  & ~A~ & ~C~ & ~W~ \\
\hline
\end{array}
$$

Let $v^R$ be the reverse of a word $v$, defined by induction: $\epsilon^R=\epsilon$ and if $(u,a)\in\Sigma^*\times\Sigma$ then $(ua)^R=a(u^R)$.

We now define $\dual(u)=\bar{u}^R$.

\begin{lemma}\label{lem:dual}
Let $t$ be a thread from $\varphi_\alpha$ to $\psi_\beta$, and $\dual(t)$ be the identical thread considered in the other direction, from $\psi_\beta$ to $\varphi_\alpha$.
Then $\visP{t^R}=\dual(\visP{t})$
\end{lemma}

Let $t_M$ be the $F$-thread of the main pre-proof $P$ going through the $\pi_M$ cut and ending with the cut rule, and $t_R$ be the analog $F$-thread for the $\pi_R$ cut. We want to build $\pi_R$ such that $\visP{(t_R)^R}=\visP{t_M}$. By Lemma \ref{lem:dual}, this implies $\visP{t_R}=\dual(\visP{t_M})$, and moreover $\visP{t_M}\subseteq \{\bar{l},\bar{r},\bar{i}\}^*$, and $\weight{t_M}$ ends with $C$. Therefore $\visP{t_Mt_R}=\epsilon$.

In the following, we will describe how gadgets of $\pi_M$ are dualized to create $\pi_R$, where the run of the machine is simulated by a $G$-thread.

\medskip\subsubsection*{Dual auxiliary metarules}

Rules relatives to $A$ are left unchanged, so we will freely used $\rinf, \rAA,\rAcut$.

The rules $\rexp,\mr,\ml$ also stay identical.

The main difference will occur when we want to introduce or delete a constraint, since they will now be reversed.
In particular, in the previous construction, we saw that introducing a constraint could be done at will without any assumption, but removing a constraint as done in $\rdec1$ needed the presence of a known bit occuring before the constraint to be removed.

We therefore redefine constraint introduction gadgets $\pi_{\rci'}$ and $\mrip$, with the notable change that the effect on stack is now $\rc\mapsto \rc\rc$, i.e. we always assume that the stack starts with $\rc$.
We keep the convention for cuts, i.e. the newly introduced cut formulas are written on the inside (left $F$ and right $G$).

$$
\begin{array}{cc}
\begin{prooftree}
\Hypo{\pi_{\rci'}}
\Infer{1}{\vdash G,F,\gA} 
\end{prooftree}
\equiv 
&
\begin{prooftree}
   
    \Infer{0}[\rinf]{\gray{\Gl,\Fl, \Frl,A}}
    \Infer{0}[\rax]{\Gr, \Frr}
    \Infer{2}[\we]{\vdash \Gl\tensor \Gr, \gray{\Fl}, \gray{\Frl},\Frr, \gA}
    \Infer{1}[\nuvee]{\vdash \Gl\tensor \Gr, \gray{\Fl}, \Fr, \gA}
	\Infer{1}[\mnv]{\vdash G,F,\gA}

 \end{prooftree}
\end{array}
$$

$$
\begin{array}{cc}
\begin{prooftree}
\Hypo{\vdash G,F,\gA}
\Infer{1}[\mrip]{\vdash G,F,\gA} 
\end{prooftree}
\equiv 
&
\begin{prooftree}
   
    \Hypo{\pi_{\rci'}}
	\Infer{1}{\vdash G,F,\gA}

\Hypo{\vdash G,F,\gA}
\Infer{2}[\rAcut]{\vdash G,F,\gA}
 \end{prooftree}
\end{array}
$$

It is easily verified that a $G$-thread entering with constraint $\rc$ will exit with constraint $\rc\rc$.

As before, we also define the left analog $\pi_{\lci'}$ and $\mlip$, with effect $\rc\mapsto\rc\lc$.

\medskip\subsubsection*{Initialisation of $\pi_R$}

We now want to initialize the dual thread. However, since the introduction rule needs to assume an $\rc$ constraint already on the stack, we will need to assume this for the $G$-thread entering $\pi_R$. This means that in the end, the $F$-thread exiting $\pi_R$ via $G$ (the real thread of interest) will have a leftover constraint $\rc$, that we will need to evacuate on the main branch, as done in the main pre-proof $P$.

Thus, we only need to add three $\rc$ constraints to the one already assumed:

$$
\begin{array}{cc}
\begin{prooftree}
\Hypo{(q_0')\vdash G,F,\gA}
\Infer{1}[\rinitp]{\vdash G,F,\gA} 
\end{prooftree}
\equiv 
&
   \begin{prooftree}
   \Hypo{(q_0')\vdash G,F,\gA}

 	\Infer{1}[\mrip]{\vdash G,F,\gA}
 	\Infer{1}[\mrip]{\vdash G,F,\gA}
    \Infer{1}[\mrip]{\vdash G,F,\gA}
    \end{prooftree}
\end{array}
$$

We will use the $(p')$ notation with $p\in Q$ for state-labelled nodes of $\pi_R$, to distinguish them from nodes in $\pi_M$.

\medskip\subsubsection*{Dual encoding of action $\Inc$}

\newcommand{\rincp}[1]{\mkrule{\mathsf{Inc'}_{#1}}}

Assume $\delta(p)=\Inc_1(q)$ with $q\in Q$.
We want to define a metarule $(\Inc_1')$ updating the configuration reached in $(p')$ by acting on the constraint stack, and ending up in node $(q')$, with effect $\rc\mapsto\rc\lc$ on the stack. Therefore it suffices to use the $\lc$ introduction defined before:

$$
\begin{array}{cc}
\begin{prooftree}
\Hypo{(q')\vdash G,F,\gA}
\Infer{1}[\rincp1]{(p')\vdash G,F,\gA} 
\end{prooftree}
\equiv 
&
   \begin{prooftree}
   \Hypo{(q')\vdash G,F,\gA}
    \Infer{1}[\mlip]{(p')\vdash G,F,\gA}
    \end{prooftree}
\end{array}
$$

To manipulate the second counter instead, we can use the $\rcount$ rule which is identical as in $\pi_M$. Thus the metarule $\rincp2$ is defined as follows:

%
%

$$
   \begin{prooftree}
   
    \Hypo{\pi_{\lci'}}
    \Infer{1}{\vdash G,F,\gA}
    \Infer{1}[\rcount]{\vdash G,F,\gA}
    
    \Hypo{(q')\vdash G,F,\gA}
    \Infer{2}[\rAcut]{(p')\vdash G,F,\gA}
    \end{prooftree}
$$

The effect of this metarule on a $G$-thread is $(\rc\lc^*\rc)\rc\mapsto (\rc\lc^*\rc)\rc\lc$

\medskip\subsubsection*{Dual encoding of action $\Dec$}
\newcommand{\rdecp}[1]{\mkrule{\mathsf{Dec'}_{#1}}}

We now need to encode a metarule having the effect $\rc\lc\mapsto\rc$ on $G$-thread.

This is done by metarule $\rdecp1$, described here:

$$
   \begin{prooftree}
   
   \Infer{0}[\rax]{\vdash \Grl,\Fr}
      \Infer{0}[\rinf]{\vdash \gray{\Grr,A}}
    \Infer{2}[\mw]{\vdash \Gr,\Fr,\gA}
 	\Infer{1}[\mr]{\vdash G,F,\gA}
    \Hypo{(q')\vdash G,F,\gA}
    \Infer{2}[\rAcut]{(p')\vdash G,F,\gA}
    \end{prooftree}
$$

%
To manipulate the second counter, again we just need to have a $(\counter)$ metarule, as shown in metarule $\rdecp2$:

%

$$
   \begin{prooftree}
   
   \Infer{0}[\rax]{\vdash \Grl,\Fr}
      \Infer{0}[\rinf]{\vdash \gray{\Grr,A}}
    \Infer{2}[\mw]{\vdash \Gr,\Fr,\gA}
 	\Infer{1}[\mr]{\vdash G,F,\gA}
 	\Infer{1}[\rcount]{\vdash G,F,\gA}
    \Hypo{(q')\vdash G,F,\gA}
    \Infer{2}[\rAcut]{(p')\vdash G,F,\gA}
    \end{prooftree}
$$

\medskip\subsubsection*{Dual encoding of action $\T{}$}
\newcommand{\rTp}[1]{\mkrule{\T'_{#1}}}
\newcommand{\rXcop}{\mkrule{\rc X\mathrm{cop}}}
If $\delta(p)=\T_1(q_Z,q_P)$, we can define the Test rule in a similar way as in $\pi_M$.
However, extra care is needed for introduction rules, as they cannot be used as freely as before.

Thus, we first define the metarule $\rXcop$ (Fig. \ref{fig:rXcop}), with effect $\rc X\mapsto \rc X\rc X$ on a $G$-thread, for $X\in\{\lc,\rc\}$.

\begin{figure*}[h]
$$
\begin{array}{cc}
\begin{prooftree}
\Hypo{ \vdash G,F,\gA}
\Infer{1}[\rXcop]{\vdash G,F,\gA} 
\end{prooftree}
\equiv 
&
   \begin{prooftree}
   
	\Infer{0}[\rinf]{\vdash \gray{\Gl,\Gamma_F,A}}
   
	\Infer{0}[\rax]{\Grl,F_{\rc\lc\rc\lc}}
	\Infer{0}[\rax]{\Grr,F_{\rc\rc\rc\rc}}   
   \Infer{2}[\mw]{\Gr,F_{\rc\lc\rc\lc},F_{\rc\rc\rc\rc}}
    \Infer{2}[\mw]{\vdash G,\gray{\Gamma_F},F_{\rc\lc\rc\lc},F_{\rc\rc\rc\rc},\gA}
    \Infer{1}[\mkrule{(\nu,\parr)^4}]{\vdash G,F,\gA}
    \end{prooftree}
\end{array}
$$
\caption{The metarule $\rXcop$}
\label{fig:rXcop}
\end{figure*}

We can now give the metarule $\rTp1$, that uses the first $\rc X$ output by $\rXcop$ to perform the zero test:

$$
   \begin{prooftree}
   \Hypo{(q_P')\vdash \Grl,\Frl,\gA}
   \Hypo{(q_Z')\vdash \Grr,\Frr,\gA}
    \Infer{2}[\rexp]{\vdash \Gr,\Fr,\gA}
    \Infer{1}[\mr]{\vdash G,F,\gA}
    \Infer{1}[\rXcop]{(p')\vdash G,F,\gA} 
    \end{prooftree}
$$

If the $G$-thread enters with constraint prefix $\rc\rc$, this prefix will be copied, the first copy will be read, put in the garbage to be collected later, and the $G$-thread will enter $(q_Z')$. Similarly if the constraint started with $\rc\lc$, the $G$-thread will go to $(q_P')$ with unchanged stack and garbage $\rc\lc$.

To treat the case of the second counter, i.e. $\delta(p)=\T_2(q_Z,q_P)$, we just need to describe dual versions of all gadgets from $(\T_2)$ in $\pi_M$. We design them so that the effect on the $G$-thread is exactly the same as the one of their original version on the $F$-thread.

\newcommand{\rmovep}{\mkrule{\mathsf{result}'}}
\newcommand{\rcoplp}{\mkrule{\cop_\lc'}}
\newcommand{\rpermp}{\mkrule{\perm'}}

The only difficulty to keep in mind is that the introduction rules now assume that the constraint stack starts with $\rc$.
We start with $\pishift'$ described Fig. \ref{fig:pishiftp}, where the pairing between $G$ and $F$ formulas in axioms is recalled in the table:
$$\begin{array}{|l|c|c|c|}
\hline
\text{before: }&\lc\lc & \lc\rc\rc\lc & \lc\rc\rc\rc \\
\hline
\text{after: }&\lc\lc  & \rc\rc\lc\lc & \rc\rc\rc\lc\\
\hline
\end{array}
$$

\begin{figure*}[h]
$$
\begin{prooftree}

    \Hypo{\piaux'}	
     \Infer{1}{\vdash \Gl,\Fll,F_{\rc\rc\lc\lc},F_{\rc\rc\rc\lc},\gA}
    \Infer{0}[\rinf]{\vdash \gray{\Gr,\Flr,\Frl,F_{\rc\lc\rc},F_{\rc\rc\rc},A}}
    \Infer{2}[\mw]{\vdash G,\Fll,\gray{\Flr},\gray{\Frl},\gray{F_{\rc\lc\rc}},F_{\rc\rc\lc\lc},\gray{F_{\rc\rc\lc\rc}},F_{\rc\rc\rc\lc}\gray{F_{\rc\rc\rc\rc}},\gA,}
\Infer{1}[\mkrule{(\nu,\parr)^6}]{\vdash G,F,\gA}

\Hypo{(\bullet)\vdash \Gl,\Fl,\gA}
\Infer{0}[\rax]{\vdash \Gr,\Fr}
\Infer{2}[\mnvw]{\vdash G,F,\gA}	

\Infer{2}[\rAcut]{(\bullet)\vdash G,F,\gA} 
\end{prooftree}
$$
\caption{The pre-proof $\pishift'$}
\label{fig:pishiftp}
\end{figure*}

For readability, we separated $\piaux'$, described Fig. \ref{fig:piauxp}:
\begin{figure*}[h]
$$
\begin{array}{cc}
\scalebox{.9}{
\begin{prooftree}
\Hypo{\piaux'}
 \Infer{1}{\vdash \Gl,\Fll,F_{\rc\rc\lc\lc},F_{\rc\rc\rc\lc},\gA}
 \end{prooftree}}
\equiv
&
\scalebox{.9}{
\begin{prooftree}	
	 
		\Infer{0}[\rax]{\vdash \Gll, \Fll}
		\Infer{0}[\rinf]{\vdash \gray{G_{\lc\rc\lc}, A}}		
		
		\Infer{0}[\rax]{\vdash  G_{\lc\rc\rc\lc},F_{\rc\rc\lc\lc}}
	
		\Infer{0}[\rax]{\vdash G_{\lc\rc\rc\rc},F_{\rc\rc\rc\lc}}
		\Infer{2}[\mw]{\vdash G_{\lc\rc\rc},F_{\rc\lc\lc},F_{\rc\rc\lc}}
		
		\Infer{2}[\mw]{\vdash \Glr,F_{\rc\lc\lc},F_{\rc\rc\lc},\gA}

    \Infer{2}[\mw]{\vdash \Gl,\Fll,F_{\rc\rc\lc\lc},F_{\rc\rc\rc\lc},\gA}
\end{prooftree}}
\end{array}
$$
\caption{The pre-proof $\piaux'$}
\label{fig:piauxp}
\end{figure*}

As before, the effect of $\pishift'$ on the $G$-thread is $\lc^{k+1} T\mapsto \lc^k T \lc$ with $T\in\{\rc\rc\rc,\rc\rc\lc\}$.
We continue with the dual $\rcoplp$ of $\rcopl$, represented Fig. \ref{fig:cop}, with action $\left\{\begin{array}{c} \lc\mapsto\lc\lc \\ \rc\rc\mapsto\rc\rc\end{array}\right.$. Notice that we now use the fact that in the present context, a stack starting with $\rc$ actually starts with $\rc\rc$, as mentioned when defining the gadget $\rcopl$.

\begin{figure*}[h]
$$
\begin{array}{cc}
\scalebox{.9}{
\begin{prooftree}
\Hypo{\vdash G,F,\gA}
\Infer{1}[\rcoplp]{\vdash G,F,\gA} 
\end{prooftree}}
\equiv 
&
\scalebox{.9}{
\begin{prooftree}

\Infer{0}[\rax]{\vdash \Gl, \Fll}
		\Infer{0}[\rinf]{\vdash \gray{\Grl,\Flr,\Frl,\gA}}	
		\Infer{0}[\rax]{\vdash \Grr, \Frr}
		\Infer{2}[\mw]{\vdash \Gr, \gray{\Flr},\gray{\Frl},\Frr,\gA} 
\Infer{2}[\we]{\vdash G,\Fll,\gray{\Flr},\gray{\Frl},\Frr,\gA}  

    \Infer{1}[\mkrule{(\nu,\parr)^2}]{\vdash G,F,\gA}
	
\Hypo{\vdash G,F,\gA}
\Infer{2}[\rAcut]{\vdash G,F,\gA} 
\end{prooftree}}
\end{array}
$$
\caption{The metarule $\rcoplp$}
\label{fig:cop}
\end{figure*}

We turn to $\pimove'$, represented Fig. \ref{fig:pimovep}, iterating $\pishift$ with effect $\lc^{k} T\mapsto \lc^k T \lc^{k}$ and $\rc\rc\mapsto\rc\rc$, functioning along the same principle as $\pimove$:
\begin{figure*}[h]
$$
\begin{array}{cc}
\scalebox{.9}{
\begin{prooftree}
\Hypo{\pimove'}
\Infer{1}{\vdash G,F,\gA} 
\end{prooftree}}
\equiv 
&
\scalebox{.9}{
\begin{prooftree}
		\Hypo{\pishift'}
		\Infer{1}{\vdash G,F,\gA} 
		\Hypo{(\bigstar')\vdash G,F,\gA}
    \Infer{2}[\rAcut]{\vdash \Gl,\Fl,\gA}
    \Infer{0}[\rax]{\vdash \Gr,\Fr}
\Infer{2}[\mnvw]{\vdash G,F,\gA} 
\Infer{1}[\rcoplp]{(\bigstar')\vdash G,F,\gA} 
\end{prooftree}}
\end{array}
$$
\caption{The pre-proof $\pimove'$}
\label{fig:pimovep}
\end{figure*}

%

The rule $\rpermp$ (Fig. \ref{fig:permp}) will again allow us to prepare the input for $\pimove'$, using the same pairing as before.
This rule has effect $(\rc\lc^k\rc)\rc X\mapsto (\rc\lc^k\rc)\rc X \rc\rc X$, with $X\in\{\lc,\rc\}$.
We will note $\Gamma_F$ for a sequent composed of several non-pertinent occurrences of formula $F$, for concision.
\begin{figure*}[h]
$$
\begin{array}{cc}\scalebox{.9}{
\begin{prooftree}
\Hypo{\vdash G,F,\gA}
\Infer{1}[\rpermp]{\vdash G,F,\gA} 
\end{prooftree}}
\equiv 
&
\scalebox{.9}{
\begin{prooftree}
	\Infer{0}[\rinf]{\gray{\Gl,\Gamma_F,A}}

	\Infer{0}[\rax]{\vdash \Grl,  F_{\rc\lc\rc\rc\lc}}	
	\Infer{0}[\rax]{\vdash \Grr, F_{\rc\rc\rc\rc\rc}}
	\Infer{2}[\mw]{ \vdash \Gr, F_{\rc\lc\rc\rc\lc}, F_{\rc\rc\rc\rc\rc}}

\Infer{2}[\mw]{\vdash G, \gray{\Gamma_F}, F_{\rc\lc\rc\rc\lc}, F_{\rc\rc\rc\rc\rc}, \gA}
\Infer{1}[\mkrule{(\nu,\parr)^5}]{\vdash G,F,\gA}
\Infer{1}[\rcount]{\vdash G,F,\gA} 

\Hypo{\vdash G,F,\gA} 

\Infer{2}[\rAcut]{\vdash G,F,\gA} 
\end{prooftree}}
\end{array}
$$
\caption{The metarule $\rpermp$}
\label{fig:permp}
\end{figure*}
We continue with the metarule $\rmovep$, having action $\rc\lc^kT\mapsto \lc^k T\lc^kT$ on the $G$-thread, with $T\in\{\rc\rc\lc,\rc\rc\rc\}$, and leaving a garbage constraint $\rc$ that will be seen later on the way down.
$$
\begin{array}{cc}\scalebox{.9}{
\begin{prooftree}
\Hypo{\vdash G,F,\gA}
\Infer{1}[\rmovep]{\vdash G,F,\gA} 
\end{prooftree}}
\equiv 
&
\scalebox{.9}{
\begin{prooftree}
 \Hypo{\pimove'}
 \Infer{1}{\vdash G,F,\gA} 
	\Hypo{\vdash G,F,\gA}
\Infer{2}[\rAcut]{\vdash \Gr,\Fr,\gA} 
\Infer{1}[\mr]{\vdash G,F,\gA} 
\Infer{1}[\rpermp]{\vdash G,F,\gA} 
\end{prooftree}}
\end{array}
$$
\newcommand{\rXr}{\mkrule{\mathsf{rXr}_i}}
Since the dual form of the $\mri$ rule used in $\rT2$ needs extra care, we will devise another auxiliary gadget rule $\rXr$ (Fig. \ref{fig:rXr}), with effect $\rc X \mapsto \rc X \rc$ with $X\in\{\lc,\rc\}$. This is to restore the stack content after the bit $X$ encoding the test result to its original value starting with $\rc$.
\begin{figure*}[h]
$$
\begin{array}{cc}\scalebox{.9}{
\begin{prooftree}
\Hypo{ \vdash G,F,\gA}
\Infer{1}[\rXr]{\vdash G,F,\gA} 
\end{prooftree}}
\equiv 
&
\scalebox{.9}{
   \begin{prooftree}
   
	\Infer{0}[\rinf]{\vdash \gray{\Gl,\Gamma_F,A}}
   
	\Infer{0}[\rax]{\Grl,F_{\rc\lc\rc}}
	\Infer{0}[\rax]{\Grr,F_{\rc\rc\rc}}   
   \Infer{2}[\mw]{\Gr,F_{\rc\lc\rc},F_{\rc\rc\rc}}
    \Infer{2}[\mw]{\vdash G,\gray{\Gamma_F},F_{\rc\lc\rc},F_{\rc\rc\rc},\gA}
    \Infer{1}[\mkrule{(\nu,\parr)^3}]{\vdash G,F,\gA}
    \end{prooftree}}
\end{array}
$$
\caption{The metarule $\rXr$}
\label{fig:rXr}
\end{figure*}
We can finally build the metarule for $\T_2$ (Fig. \ref{fig:T2p}), performing the wanted test and leaving garbage constraint of the form $\rc\lc^k\rc\rc X$ with $X\in \{\lc,\rc\}$. 
\begin{figure*}[h]
$$
\begin{array}{cc}
\scalebox{.9}{
\begin{prooftree}
   \Hypo{(q_P)\vdash G,F,\gA}
   \Hypo{(q_Z)\vdash G,F,\gA}
\Infer{2}[\rTp2]{(p)\vdash G,F,\gA} 
\end{prooftree}}
\equiv 
&
\scalebox{.9}{
   \begin{prooftree}
   
	\Hypo{(\clubsuit')\vdash \Gl,\Fl,\gA}
   \Hypo{(q_P) \vdash G_{\rc\rc\lc},F_{\rc\rc\lc},\gA}
\Hypo{(q_Z) \vdash G_{\rc\rc\rc},F_{\rc\rc\rc},\gA}

    \Infer{2}[\rexp]{\vdash \Grr,\Frr,\gA}
    \Infer{1}[\mr]{\vdash \Gr,\Fr,\gA}
    \Infer{1}[\rXr]{\vdash \Gr,\Fr,\gA}
    \Infer{2}[\rexp]{(\clubsuit')\vdash G,F,\gA}
    \Infer{1}[\rmovep]{(p)\vdash G,F,\gA} 
    \end{prooftree}}
\end{array}
$$
\caption{The metarule $\rTp2$}
\label{fig:T2p}
\end{figure*}
After entering $\rTp2$ with stack content $u$, the $G$-thread will reach node $(q_Z')$ (resp. $(q_P')$) if the second counter value is zero (resp. not zero), with same stack content $u$.

\medskip\subsubsection*{Dual final state $q_f'$}
As before, when reaching the node $(q_f')$, we just need bounce on an axiom after evacuating the formula $A$ from the sequent. The same gadget as in $\pi_M$ can be used:

$$
\begin{prooftree}
\Infer{0}[\rinf]{\vdash \gray{ \Gl,\Fl,A}}
\Infer{0}[\rax]{\vdash \Gr,\Fr}
\Infer{2}[\mnvw]{(q_f')\vdash G,F,\gA}
\end{prooftree}
$$

\medskip\subsubsection*{Correctness of the pre-proof $P$ and conclusion}

By construction, if the machine $M$ does not halt, then the $F$-thread entering $\pi_M$ will never exit it, so it cannot validate the branch looping through the $(\star)$ nodes, and the pre-proof $P$ is not a proof.
Conversely, if $M$ halts, then the $F$-thread will exit $\pi_M$ with a constraint $u$ encoding the results of the tests as described in Lemma \ref{lem:exitconstraints}. Moreover, the $G$-thread $t'$ entering $\pi_R$ with a single constraint $\rc$ will exit on the same node with the same constraints $u$. This means that the $F$-thread going through $\pi_M$ and $\pi_R$ will exit $\pi_R$ with a single constraint $\rc$. This constraint will be popped on the first unfolding of $F$, and the second (left) unfolding of $F$ will be on the visible part of the thread. The thread then reaches node $(\star)$, and loops back to the root to go through the same path infinitely many times.

Thus the $F$-thread starting in the root validates the branch containing infinitely many $(\star)$ nodes. All other infinite branches of $P$ are validated by non-bouncing thread following formulas $A$, that can also be originated in the root of $P$, in the formula $B$.

Therefore, we showed the following theorem:
\begin{theorem}
The pre-proof $P$ is a proof if and only if the 2CM $M$ has a run.
\end{theorem}
By Theorem \ref{thm:2CM}, we obtain that deciding whether a circular pre-proof of \muMLLo is a proof is undecidable, and more precisely it is $\Sigma^0_1$-hard.
In the next section, we show that this problem is recursively enumerable, so we obtain Cor. \ref{cor:sigma01}, stating that deciding validity of a circular pre-proof of \muMLLo is $\Sigma^0_1$-complete.